\documentclass[12pt]{article}
\usepackage{amsmath}
\usepackage{amsfonts} 
\usepackage{graphicx,psfrag,epsf}
\usepackage{enumerate}
\usepackage{natbib}
\usepackage{url} 
\usepackage{hyperref}
\usepackage{graphicx}
\usepackage{amsthm}
\usepackage{bbm}
\usepackage{soul}
\usepackage{paralist}
\usepackage{mathtools}
\usepackage{adjustbox}
\usepackage{booktabs}
\usepackage{xcolor}
\usepackage{amssymb}
\usepackage{multirow}  
\usepackage{array}     
\usepackage{physics}
\usepackage{tablefootnote}
\usepackage{enumitem}
\usepackage{amssymb}

\usepackage{lipsum}

\allowdisplaybreaks

\newtheorem{theorem}{Theorem}
\newtheorem{lemma}{Lemma}

\newtheorem{assumption}{Assumption}
\newtheorem{proposition}{Proposition}

\theoremstyle{definition}

\theoremstyle{remark}
\newtheorem{remark}{Remark}

\newcommand{\E}{\mathbb{E}}

\newcommand{\pr}{\mathbb{P}}
\newcommand{\corr}{\mathbb{C}\mathrm{or}}
\newcommand{\cov}{\mathbb{C}\mathrm{ov}}

\newcommand{\astar}{\alpha^{*}}
\newcommand{\atheta}{\alpha_{\theta}}
\newcommand{\aphi}{\alpha_{\phi}}

\newcommand{\Var}{\mathbb{V}{\rm ar}}

\usepackage{tikz}
\usepackage{color}

\usetikzlibrary{shapes,decorations,arrows,calc,arrows.meta,fit,positioning}
\tikzset{
    -Latex,auto,node distance =1 cm and 1 cm,semithick,
    state/.style ={ellipse, draw, minimum width = 0.7 cm},
    point/.style = {circle, draw, inner sep=0.04cm,fill,node contents={}},
    bidirected/.style={Latex-Latex,dashed},
    el/.style = {inner sep=2pt, align=left, sloped}
}

\definecolor{edgeblue}{HTML}{0072B2}
\definecolor{edgeorange}{HTML}{D55E00}
\definecolor{edgeteal}{HTML}{009E73}

\usepackage{setspace}
\usepackage{bibspacing}
\newcommand{\blind}{0}

\begin{document}

\def\spacingset#1{\renewcommand{\baselinestretch}%
{#1}\small\normalsize} \spacingset{1}


\if0\blind
{
  \title{\bf 
Identification and estimation of causal mechanisms in cluster-randomized trials with post-treatment confounding using Bayesian nonparametrics
  }
  \author{Yuki Ohnishi$^{1*}$, Michael J. Daniels$^{2\star}$, Lei Yang$^{2\dagger}$, and Fan Li$^{1 \ddagger}$\\
    $^1$Department of Biostatistics,  
    Yale University School of Public Health, \\
    New Haven, CT, United States\\
    $^2$Department of Statistics, University of Florida, \\
    Gainesville, FL, United States\\
    \\
    $^*$yuki.ohnishi@yale.edu, $^\star$daniels@ufl.edu, $^\dagger$yang.lei@ufl.edu, $^\ddagger$fan.f.li@yale.edu}
  \maketitle
} \fi

\if1\blind
{
  \bigskip
  \bigskip
  \bigskip
  \begin{center}
    {\LARGE\bf The nested enriched Dirichlet process for causal mediation in cluster-randomized trials with post-treatment confounding}
\end{center}
  \medskip
} \fi

\bigskip

\begin{abstract}
Causal mediation analysis in cluster-randomized trials (CRTs) is essential for explaining how cluster-level interventions affect individual outcomes, yet it is complicated by interference, post-treatment confounding, and hierarchical covariate adjustment. We develop a Bayesian nonparametric framework that simultaneously accommodates interference and a post-treatment confounder that precedes the mediator. Identification is achieved through a multivariate Gaussian copula that replaces cross-world independence with a single dependence parameter, yielding a built-in sensitivity analysis to residual post-treatment confounding. For estimation, we introduce a nested common atoms enriched Dirichlet process (CA-EDP) prior that integrates the Common Atoms Model (CAM) to share information across clusters while capturing between- and within-cluster heterogeneity, and an Enriched Dirichlet Process (EDP) structure delivering robust covariate adjustment without impacting the outcome model. We provide formal theoretical support for our prior by deriving the model's key distributional properties, including its partially exchangeable partition structure, and by establishing convergence guarantees for the practical truncation-based posterior inference strategy. We demonstrate the performance of the proposed methods in simulations and provide further illustration through a reanalysis of a completed CRT.
\end{abstract}

\noindent%
\vfill

\newpage
\spacingset{1.9} 

\section{Introduction}
\label{sec:intro}
\subsection{Background and literature review}
Cluster randomized trials (CRTs), where groups of individuals are randomized to treatment conditions, are a cornerstone of intervention research in public health, education, and social policy \citep{Turner2017}. Clusters can be schools, hospitals, or communities, and are chosen to be the unit of randomization 
when individual randomization is infeasible or risks contamination.
Beyond simply determining \emph{if} an intervention is effective, a central goal of scientific inquiry is to understand \emph{how} it works. 
Causal mediation analysis can reveal these mechanisms by disentangling intervention pathways \citep{Imai2010, Daniel2015}, but two features of CRTs complicate standard approaches: interference within clusters and confounding by intermediate variables that occur after treatment assignment.

In CRTs, the intervention is implemented at the cluster level: all individuals within the same cluster receive the same treatment according to the cluster's assignment. \emph{Interference} arises when one individual's treatment-induced behaviors or states affect another individual's outcome within that cluster. 
In mediation analysis of CRTs, a key channel of interference is the mediator itself: an individual-level mediator  can influence peers' outcomes through shared environment, social norms, or resource reallocation. This motivates decomposing mediation effects into (i) an \emph{individual} component operating through a unit's own mediator and (ii) a \emph{spillover} component operating through others' mediators within the cluster \citep{VanderWeele2013_CRT}.
A second complication arises from the presence of \emph{post-treatment confounders}: intermediate variables that occur after treatment assignment and influence both the mediator and outcome. Such variables are not uncommon (e.g., mandated clinic visits that deliver growth monitoring and nutrition counseling), and if unaccounted for, can induce bias through uncontrolled mediator-outcome confounding. 
Identification in the presence of post-treatment confounding relies on strong cross-world independence assumptions \citep{VanderWeele2014_effect_decomposition}, which are inherently untestable from the observed data. Careful adjustment and sensitivity analyses are thus essential to ensure the credibility of causal mediation findings.

A few prior efforts have developed causal mediation methods to address interference in CRTs, all of which assume away post-treatment confounding. For example, \citet{VanderWeele2013_CRT} decomposed the natural indirect effect into a spillover mediation effect and an individual mediation effect, and established nonparametric identification. \citet{cheng2024} studied the semiparametric efficiency theory and proposed conditionally doubly robust estimators for these estimands. \citet{ohnishi2025} proposed a Bayesian nonparametric (BNP) framework for assessing mediation and spillover effects with multiple unordered mediators in CRTs. They leveraged a newly-developed nested dependent Dirichlet process (nDDP) as the inferential engine for complex multilevel data. Although finer estimands were proposed to investigate mechanisms via multiple mediators, their identification assumptions still rule out post-treatment confounding.
For independent data, several methods were recently introduced to estimate indirect effects in the presence of post-treatment confounding \citep[e.g.,][]{Hong2023, Rudolph2023,xia2023identification, Bae2024}. However, these methods assume no interference nor multilevel data structure and hence cannot be applied directly in CRTs to investigate finer mechanisms.

\subsection{Motivating application}
\label{sec:motivating_example}

Child undernutrition remains a major public health challenge, increasing early-life morbidity and mortality and constraining cognitive and economic potential across the life course. The conditional cash transfer (CCT) programs aim to relax liquidity constraints while incentivizing verifiable health behaviors that promote investments in child health and development. 
The \emph{Red de Protecci\'on Social} (RPS) study is a landmark CRT evaluating the efficacy of CCT programs, conducted in $42$ rural \emph{comarcas}, administrative regions (unit of randomization), in Nicaragua and randomized one to one to treatment and control. Baseline interviews were conducted in August-September $2000$, with follow-ups in October $2001$ and $2002$. The sample comprises $N=449$ children aged $6$-$35$ months measured in $2002$.

\begin{table}[htbp]
\centering
\caption{Summary statistics of baseline characteristics, mean (standard deviation). }
\label{tab:baseline_cov}
\begin{adjustbox}{width=14.2cm}
\begin{tabular}{llcc}
\toprule
\textbf{Characteristic} &  & Control ($n=253$ households) & Treatment ($n=196$ households) \\
\midrule
\multicolumn{4}{l}{\textbf{Cluster-level covariates}}\\
& Cluster size (number of households) & $12.048 \;(4.955)$ & $9.333 \;(4.328)$  \\
\multicolumn{4}{l}{\textbf{Individual-level covariates}}\\
& Mother's literacy & $1.897 \;(0.987)$ & $1.913 \;(0.981)$  \\
& Mother's formal education & $1.972 \;(1.570)$ & $1.923 \;(1.488)$ \\
& Highest educational attainment & $2.870 \;(1.156)$ & $2.857 \;(0.829)$ \\
\multicolumn{4}{l}{\textbf{Post-treatment confounder}}\\
& Health-checkup condition & $0.854 \;(0.354)$ & $0.954 \;(0.210)$ \\
\multicolumn{4}{l}{\textbf{Mediator}}\\
& Household dietary diversity score & $7.794 \;(1.844)$ & $9.230 \;(1.732)$ \\
\multicolumn{4}{l}{\textbf{Outcome}}\\
& Height-for-age z-score   & $-1.865 \;(1.213)$ & $-1.540 \;(1.147)$ \\
\bottomrule
\end{tabular}
\end{adjustbox}
\end{table}

The CCT program is randomized at the comarcas level, and provided demand-side transfers and funded supply-side services to the households in treated comarcas.
Baseline covariates that may confound mediator-outcome relations include: mother's formal education, mother's literacy, and the household's highest educational attainment.
The primary outcome was height-for-age $z$-score. All anthropometrics were measured by trained interviewers and standardized to international reference distributions. 
Household dietary diversity score is  measured and constructed from $12$ food-group categories using detailed records of foods acquired by any household member in the prior $15$ days. 
Additionally, preventive health services use is captured by a binary indicator for whether the child was taken to a routine health check-up in the previous $6$ months. One potential hypothesis is that, by the temporal ordering, health check-ups precede dietary diversity. This is because health check-up behavior is directly incentivized by RPS conditionalities, and it can influence both subsequent household food choices (via counseling and growth-monitoring feedback) and child growth; hence, health check-up behavior serves as a post-treatment confounder. Table \ref{tab:baseline_cov} reports the summary statistics of baseline covariates, post-treatment confounders, mediators, and outcomes for both treatment and control groups.

Prior work on the RPS trial reported improvements in child nutrition and examined mediation from several angles.  \citet{Charters2023} modeled routine pediatric check-ups and household dietary diversity as independent, sequential mediators; however, their analysis assumed no interference, a strong restriction in CRTs that rules out within-cluster spillovers, and it relied on parametric models, raising concerns about model misspecification. \citet{cheng2024} developed doubly robust estimators for mediation and spillover effects in CRTs with a single mediator; however, they focused exclusively on dietary diversity and did not adjust for the post-treatment confounder (health check-ups). Finally, \citet{ohnishi2025} treated these two as unstructured mediators, ignoring their likely causal ordering, with check-ups plausibly preceding changes in diet. Consequently, both \citet{cheng2024} and \citet{ohnishi2025} leave open the possibility of post-treatment mediator-outcome confounding through check-ups, when investigating the causal role of household dietary diversity score. In addition, these analyses invoke cross-world independence for the check-up variable within a cluster, an untestable condition that warrants principled sensitivity analysis of the counterfactual dependence structure.

\subsection{Contributions}
Motivated by the RPS study, our work makes two primary methodological contributions to causal mediation analysis. First, we develop a framework for causal mediation in CRTs with post-treatment confounders, together with a sensitivity analysis procedure for cross-world dependence among those confounders. We consider two causally ordered intermediate outcomes, and treat the first as a post-treatment confounder. 
In this setting, we define pathway-specific causal estimands for the cluster-level intervention, including spillover and individual mediator effects. We propose a new assumption that replaces cross-world independence with a Gaussian copula sensitivity model indexed by a single, interpretable dependence parameter \citep[e.g.,][]{Hong2023,Bae2024}. Under this assumption, we show that the proposed estimands are identified from the observed-data law once a sensitivity model is specified for the unidentified cross-world law of the post-treatment confounder. This formulation provides a transparent framework for sensitivity analysis to examine how causal conclusions change as the untestable cross-world assumption is relaxed.
Placing a prior on this sensitivity parameter enables transparent reporting as assumptions are relaxed or tightened, with automatic uncertainty quantification.

Second, we develop a novel BNP prior, termed the nested Common Atoms Enriched Dirichlet Process (CA-EDP), for causal modeling in CRTs. The importance of flexible models in CRTs has been emphasized in prior work \citep{Ho2013} to mitigate bias due to model misspecification. 
In causal mediation analysis, one requires a coherent model for the joint distribution of covariates and responses. The recently proposed nDDP model \citep{ohnishi2025}, while highly flexible, does not jointly model covariates and responses; instead, it specifies the conditional response distribution directly given confounders, replacing the covariate distributions with its empirical counterpart, which can under-propagate uncertainty in the confounder distribution. Consequently, when the covariate space is high-dimensional, estimating this functional dependence becomes difficult, making the approach more vulnerable to the curse of dimensionality. In addition, by modeling the joint distribution of covariates, we can automatically accommodate ignorable missingness (in covariates).
Another challenge is that the nDDP may suffer from a \emph{degeneracy} issue inherent to its nested random-partition structure \citep{Camerlenghi2019_2}: if two distributions share even a single atom, they are automatically assigned to the same latent group. 
Our proposed CA-EDP combines two BNP priors: the Common Atoms Model (CAM) \citep{Denti2023} and the Enriched Dirichlet Process (EDP) \citep{Wade2011}. CAM enables sharing a common set of atoms to capture between- and within-cluster heterogeneity while avoiding degeneracy, and EDP supplies robust covariate adjustment at both the cluster and individual levels through joint modeling. We provide theoretical underpinnings for the CA-EDP, including analyses of its distributional properties (tie probabilities and correlation structures), discuss a random partition property showing how the model avoids degeneracy, and establish convergence guarantees for the truncated CA-EDP.
Our extensive simulation results confirm that our methods outperform the parametric model in accuracy and provide conservative uncertainty quantification relative to the nDDP across a range of data-generating processes. We then discuss how the CA-EDP can complement the nDDP.

Our empirical analysis using the RPS data yields several key insights. First, even after adjusting for the post-treatment confounder, the CCT program's effect on child nutritional status is largely mediated through household dietary diversity with pronounced spillovers within administrative regions, findings that strengthen prior results \citep{cheng2024,ohnishi2025}. However, adjusting for post-treatment confounders reduces the magnitude of the effect mediated through household dietary diversity and increases the direct effect operating independently of the mediators.
Using the proposed copula-based sensitivity analysis, we further assess how sensitive our conclusions are to the cross-world dependence structure of post-treatment confounders and find that both the sign and the magnitude of the diet-mediated effects remain stable.

\section{Structural Assumptions, Causal Estimands, and Nonparametric Identification}
\label{sec:setup}
\subsection{Notation, data structure, and causal estimands}
\label{sec:notation}

We consider a CRT with $ I $ clusters. For cluster $ i \in \{1, \ldots, I\} $, we denote $ N_i $ as the cluster size, $ A_i \in \{0, 1\} $ as the cluster-level treatment assignment, with $A_i = 1$ if it is assigned treatment (CCT program) and $A_i = 0$ otherwise, and $ \mathbf{V}_i \in \mathcal{V}=\mathbb{R}^{q \times 1} $ as a vector of cluster-level baseline covariates. 
The total number of individuals in the study is denoted by $ N = \sum_{i=1}^{I} N_i$.
For individual $ j \in \{1, \ldots, N_i\} $ in cluster $ i $, we observe a vector of individual-level baseline covariates $ \mathbf{X}_{ij} \in \mathcal{X}=\mathbb{R}^{p \times 1} $, and write $ \mathbf{X}_i = [\mathbf{X}_{i1}, \ldots, \mathbf{X}_{iN_i}]^\top \in \mathbb{R}^{N_i \times p} $. 
Let $ \mathbf{C}_i = \{\mathbf{V}_i, \mathbf{X}_i\} $ represent all baseline covariates in cluster $ i $ and $ \mathbf{C}_{ij} = \{\mathbf{V}_i, \mathbf{X}_{ij}\} $ represent baseline covariates of individual $j$ in cluster $ i $. 
We also observe an individual-level outcome $ Y_{ij} \in \mathbb{R} $ (i.e., child nutritional outcome z-score), and consider an individual-level post-treatment confounder $D_{ij} \in \mathbb{R}$  (i.e., child health check-up condition), measured before observing the outcome but after assignment. Additionally, we observe an individual-level mediator $M_{ij} \in \mathbb{R}$ (i.e., household dietary diversity score) that lies between $D_{ij}$ and $Y_{ij}$. We let $ \mathbf{Y}_i = [Y_{i1}, \ldots, Y_{iN_i}]^\top \in \mathbb{R}^{N_i \times 1} $, $ \mathbf{D}_i = [D_{i1}, \ldots, D_{iN_i}]^\top \in \mathbb{R}^{N_i \times 1} $, $ \mathbf{M}_i = [M_{i1}, \ldots, M_{iN_i}]^\top \in \mathbb{R}^{N_i \times 1} $, and $ \mathbf{D}_{i(-j)}, \mathbf{M}_{i(-j)} \in \mathbb{R}^{(N_i - 1) \times 1} $ as the vector of post-treatment confounders and mediators from cluster $ i $ excluding individual $ j $. 
Finally, we let $ \mathbf{A} $, $ \mathbf{M} $, $ \mathbf{D} $, and $ \mathbf{Y} $ be the ($ I \times 1 $)-dimensional vector of treatment assignments and the ($ N \times 1 $)-dimensional vectors of the mediators, the post-treatment confounders and outcomes, respectively. Figure \ref{fig:DAG_mediators} provides a graphical representation of the causal structure between the observed variables. We adopt the potential outcomes framework and define $D_{ij}(\mathbf{A})$, $M_{ij}(\mathbf{A}, \mathbf{D})$ as the potential mediator variables under assignment vector $ \mathbf{A} $, and $Y_{ij}(\mathbf{A}, \mathbf{D}, \mathbf{M})$ as the potential outcomes for unit $ j $ in cluster $ i $ when $ \mathbf{A}, \mathbf{D}, \mathbf{M} $ were the vectors of assignments, post-treatment confounders, and mediators in the whole study population.
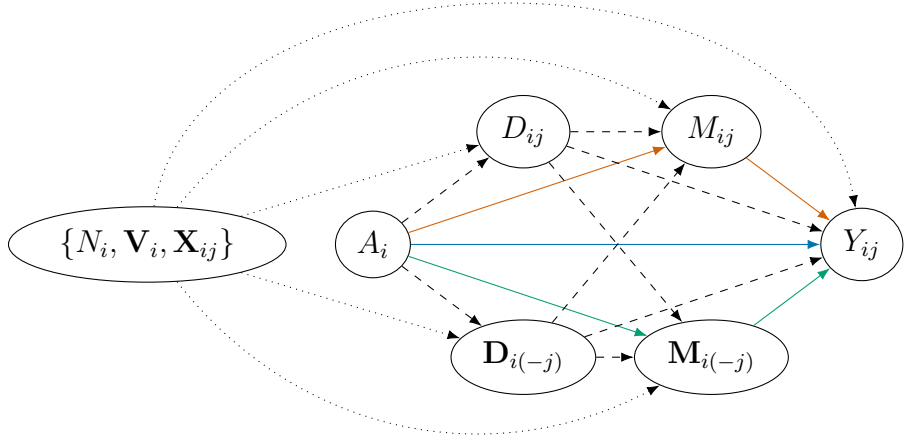
\begin{figure}[htbp]
\centering 
 \begin{tikzpicture}
 
    \node[state] (a) at (-0.5,0) {$A_i$};
    \node[state] (d1) at (1.5,1.5) {$D_{ij}$};
    \node[state] (m1) at (4, 1.5)  {$M_{ij}$};
    \node[state] (y) at (6,0) {$Y_{ij}$};
    \node[state] (d2) at (1.5,-1.5) {$\mathbf{D}_{i(-j)}$};
    \node[state] (m2) at (4, -1.5) {$\mathbf{M}_{i(-j)}$};
    \node[state] (c) at (-3.5,0) {$\left\{N_i,\mathbf{V}_i,\mathbf{X}_{ij}\right\}$};

    \path[dashed] (a) edge (d1);
    \path[edgeorange]         (a) edge (m1);
    \path[dashed] (a) edge (d2);
    \path[edgeteal]         (a) edge (m2);
    \path[edgeblue]      (a) edge (y);
    
    \path[dashed]  (d1) edge (y);
    \path[edgeorange]          (m1) edge (y);
    \path[dashed]  (d2) edge (y);
    \path[edgeteal]          (m2) edge (y);
    
    \draw[dashed]  (d1) edge (m1);
    \draw[dashed]  (d1) edge (m2);
    \draw[dashed]  (d2) edge (m1);
    \draw[dashed]  (d2) edge (m2);

    \path[dotted] (c) edge[bend left=80] (y);
    \path[dotted] (c) edge (d1);
    \path[dotted] (c) edge (d2);
    \path[dotted] (c) edge[bend left=40] (m1);
    \path[dotted] (c) edge[bend left=-40] (m2);
    
\end{tikzpicture}
\caption{ Mediation directed acyclic graph for a CRT with a post-treatment confounder. Here, $\mathbf{X}_i$ and $N_i$ are baseline covariates and cluster size, $A_i$ is cluster-level treatment assignment, $D_{ij}$, $M_{ij}$ and $Y_{ij}$ are post-treatment confounder, mediator and outcome of individual $j$ in cluster $i$, and $\mathbf{D}_{i(-j)}$ and $\mathbf{M}_{i(-j)}$ are the vectors of post-treatment confounders and mediators excluding individual $j$. The blue, orange, and green arrows represent the natural direct effect, individual mediator effect, and spillover mediator effect, respectively.}
\label{fig:DAG_mediators}
\end{figure}

\begin{assumption}[Cluster-level SUTVA]
\label{asmp:sutva}
    Cluster-level stable unit treatment value assumption (SUTVA) for the cluster-randomized experiment consists of two parts:
    \begin{enumerate}
        \item \textit{(No interference between clusters).} An individual's potential mediators and outcome do not vary with treatments assigned to clusters other than their cluster, that is, $ D_{ij}(\mathbf{A}) = D_{ij}(A_i)$, $ M_{ij}(\mathbf{A}, \mathbf{D}) = M_{ij}(A_i, \mathbf{D}_{i})$, and $Y_{ij}(\mathbf{A}, \mathbf{D}, \mathbf{M}) = Y_{ij}(A_i, \mathbf{D}_{i}, \mathbf{M}_{i}) $ 
        \item \textit{(No different versions of treatment).} There are no different versions of each treatment level, that is, if $ A_i = A_i' $ then $ D_{ij}(A_i) = D_{ij}(A_i')$ and if $ A_i = A_i' $, $\mathbf{D}_{i}= \mathbf{D}_{i}' $ and $\mathbf{M}_{i}= \mathbf{M}_{i}' $ then $ Y_{ij}(A_i, \mathbf{D}_{i}, \mathbf{M}_{i}) = Y_{ij}(A_i', \mathbf{D}_{i}', \mathbf{M}_{i}')$.
    \end{enumerate}
\end{assumption}

The first part of the cluster-level SUTVA does not rule out the possibility of spillover effects of the mediators on the outcomes within the same cluster. That is, the outcome for unit $ j $ in cluster $ i $ can still be affected by mediators of other units in the same cluster $i$. The second part of Assumption \ref{asmp:sutva} ensures a well-defined treatment. This is a standard assumption in causal inference and is considered in the RPS study since CCT follows a standardized procedure. By Assumption \ref{asmp:sutva}, we define $ D_{i j}(a) $ as the potential intermediate confounder under condition $ a \in \{0, 1\} $, and $ \mathbf{D}_{i}(a) = [D_{i 1}(a), \ldots, D_{i N_i}(a)]^\top $ as the vector of potential confounder variables for all individuals in cluster $ i $.
Define $ M_{ij}(a, \mathbf{D}_{i}(a)) $ as the potential mediator for unit $j$ in cluster $i$ under condition $ a \in \{0, 1\} $ and $ \mathbf{M}_{i}(a) = [M_{i1}(a, \mathbf{D}_{i}(a)), \ldots, M_{i N_i}(a, \mathbf{D}_{i}(a))]^\top $ as the vector of potential mediator variables for all individuals in cluster $ i $.
Finally, $ Y_{i j}(a, \mathbf{D}_{i}(a), \mathbf{M}_{i}(a)) $ represents the potential outcome.

To simplify notation and focus on evaluating the mediation effect, we succinctly write $M_{ij}(a)$ and $ Y_{i j}(a, \mathbf{M}_{i}(a')) $ for $M_{ij}(a, \mathbf{D}_{i}(a))$ and $ Y_{i j}(a, \mathbf{D}_{i}(a), \mathbf{M}_{i}(a')) $, respectively.
Also, notice that one can equivalently represent $ Y_{i j}(a, \mathbf{m}_i) = Y_{i j}(a, m_{i j}, \mathbf{m}_{i(-j)}) $ with $ \mathbf{m}_i = \{m_{i j}, \mathbf{m}_{i(-j)} \} $; this notation explicitly distinguishes an individual's own mediator from the mediators of the remaining cluster members. 
The only possibly observable potential outcome is the one where, if $ A_i $ were set to $ a $, the mediators of all the units in cluster $i$  were set to the value they would have taken under condition $a$. Throughout we use the following notation for potential outcomes of this type: 
$ Y_{i j}(a) = Y_{i j}(a, \mathbf{M}_{i}(a)) = Y_{i j}(a, M_{i j}(a), \mathbf{M}_{i(-j)}(a) ) $.

\begin{assumption}[Cluster randomization]
    \label{asmp:randomization}
    The treatment assignment for each cluster is an independent realization from a Bernoulli distribution with $ p(A_i =1 ) = \pi \in (0,1) $.
\end{assumption}

\begin{assumption}[Super-population framework]
\label{asmp:superpopulation}
    (a) The cluster size $N_i$ follows an unknown distribution $ \mathcal{P}^{N}$ over $ \mathbb{N}^{+} $. (b) Conditional on $ N_i $, all observed variables are realization from the joint distribution $\mathcal{P}^{Y, M, D, C, A \mid N} $. Furthermore, positivity holds such that the conditional density $f_{M, D \mid \mathbf{C}, A, N} (\mathbf{m},\mathbf{d} \mid \mathbf{c}, a, n) >0$ for any $\{ \mathbf{m},\mathbf{d}, \mathbf{c}, a, n \}$ over their valid support.
\end{assumption}

\noindent
Assumption \ref{asmp:randomization} eliminates unmeasured confounding for both the treatment-mediator and the treatment-outcome relationships, and is guaranteed by cluster randomization. Assumption \ref{asmp:superpopulation} extends \citet{Wang2024} and conceptualizes a super-population of clusters with a finite size of individuals within each cluster.

Without ruling out the potential for informative cluster size \citep{kahan2024demystifying}, we focus on the cluster-average treatment effect $\mathrm{TE_C} = \E \left[ \frac{1}{N_i} \sum_{j=1}^{N_i} \left\{ Y_{ij}(1) - Y_{ij}(0) \right\}   \right]$.
The cluster-average treatment effect can be decomposed into two parts: the natural direct effect (NDE) and the natural indirect effect (NIE). That is, $ \mathrm{TE_C} = \mathrm{NIE_C} + \mathrm{NDE_C} $, where 
\begin{align*}
    \mathrm{NIE_C} &= \E \left[ \frac{1}{N_i} \sum_{j=1}^{N_i} \left\{ Y_{ij}(1, \mathbf{M}_i(1)) - Y_{ij}(1, \mathbf{M}_i(0)) \right\}   \right], \\
    \mathrm{NDE_C} &= \E \left[ \frac{1}{N_i} \sum_{j=1}^{N_i} \left\{ Y_{ij}(1, \mathbf{M}_i(0)) - Y_{ij}(0, \mathbf{M}_i(0)) \right\}   \right].
\end{align*}
Due to within-cluster interference, the NIE can be further decomposed into the individual mediator effect (IME) and spillover mediator effect (SME) as
$ \mathrm{NIE}_{\mathrm{C}} = \mathrm{IME}_{\mathrm{C}} + \mathrm{SME}_{\mathrm{C}}$, where 
\begin{align*}
    \mathrm{IME}_{\mathrm{C}} 
    &= \E \left[ \frac{1}{N_i}\sum_{j=1}^{N_i} \left \{ Y_{ij}(1, M_{ij}(1), \mathbf{M}_{i(-j)}(0)) - Y_{ij}(1, M_{ij}(0), \mathbf{M}_{i(-j)}(0)) \right \} \right], \\
    \mathrm{SME}_{\mathrm{C}} 
    &= \E \left[ \frac{1}{N_i}\sum_{j=1}^{N_i} \left \{ Y_{ij}(1, M_{ij}(1), \mathbf{M}_{i(-j)}(1)) - Y_{ij}(1, M_{ij}(1), \mathbf{M}_{i(-j)}(0)) \right \} \right].
\end{align*}
These estimands have been studied by \citet{cheng2024}, who estimated causal mediation effects in the absence of a post-treatment confounder. 
In this article, we consider a setting with an exposure-induced mediator-outcome confounder $D$, that is, a post-treatment variable affected by treatment that acts as a common cause of the mediator $M$ and the outcome $Y$ (see Figure \ref{fig:DAG_mediators}).
As noted by \citet{Chao2025}, there are at least two complementary perspectives on the role of a post-treatment confounder $D$ in causal mediation analysis: (i) treating $D$ as a post-treatment confounder that precedes the second mediator $M$, and (ii) regarding $D$ as an additional mediator on the pathway from $A$ to $Y$. While \citet{ohnishi2025} considered perspective (ii) by studying two causally unstructured mediators, here we adopt perspective (i). As a context, prior work by \citet{ohnishi2025} indicates that the second intermediate outcome, household dietary diversity, serves as a key mediator, and \citet{cheng2024} focused on this mediator but ignored health check-ups as the post-treatment confounder. Hence, it remains unanswered what the pure mediation effect of household dietary diversity would be when appropriately accounting for post-treatment confounding. More generally, it remains unclear how to properly identify $\mathrm{NIE_C}$ and $\mathrm{SME}_{\mathrm{C}}$ in CRTs when post-treatment confounding is present, thus motivating the proposed methodological framework.

Figure \ref{fig:DAG_mediators} provides a graphical representation of each estimand. The NDE (blue arrow) quantifies the effect of the treatment on the outcome by setting the mediators $\mathbf{M}_i$ to their natural values (the potential values in the absence of the intervention). The IME (orange arrow) quantifies the effect of the treatment on the outcome of unit $j$ in cluster $i$ that operates through the unit's own mediator, $M_{ij}$. The SME (green arrow) quantifies the effect of the treatment on the outcome of unit $j$ in cluster $i$ that operates independently of the unit's own mediator $M_{ij}$ but through the peers' mediators within the same cluster $i$,  $\mathbf{M}_{i(-j)}$. The composition of the orange arrows and green arrows constitutes the NIE, and the composition of all colored arrows constitutes the TE under CRTs.

\subsection{Nonparametric identification}
\label{sec:identification_cwi}
To identify the proposed causal mediation estimands, we introduce additional identification assumptions and then provide the nonparametric identification results. 
\begin{assumption}[Sequential ignorability]\label{asmp:si}  $\mathbf{M}_i(a) \perp \mathbf{D}_i(a') \mid A_i=a, \mathbf{D}_i(a)=\mathbf{d}, \mathbf{C}_i=\mathbf{c}, N_i=n$ and $Y_{ij}(a,\mathbf{m}) \perp \mathbf{M}_i(a') \mid A_i=a, \mathbf{D}_i(a)=\mathbf{d}, \mathbf{C}_i=\mathbf{c}, N_i=n $
for all $ i, j$, $ a,a' \in \{0,1\}$, $\mathbf{d}$, and $\mathbf{m}$ over their valid support. 
\end{assumption}

\begin{assumption}[Cross-world inter-individual independence]
\label{asmp:cond_crossworld_indep_mediators}
Conditional on cluster size, baseline covariates, and post-treatment confounders, we have $M_{ij}(1)\perp\mathbf{M}_{i(-j)}(0)\mid \mathbf{D}_i(1)=\mathbf{d},\ \mathbf{C}_i=\mathbf{c},\ N_i=n$.
\end{assumption}

In the RPS setting, $A_i$ was randomized at the cluster level, which eliminates any confounding associated with treatment by design. Assumption \ref{asmp:si} further extends sequential ignorability \citep{Imai2013} to clustered data with a post-treatment confounder, ruling out unmeasured confounding of the $\mathbf{D} \to \mathbf{M}$ and $\mathbf{M} \to Y$ relations after conditioning on $(A_i,\mathbf{C}_i,N_i)$. 
In particular, the post-treatment confounder $\mathbf{D}_i(a)$, routine health check-ups mandated and verified by program conditionalities, precedes the mediator window and is closely monitored through the management information system. Given this design and measurement timing, conditioning on $(\mathbf{C}_i,N_i,\mathbf{D}_i(a))$ plausibly captures the main determinants of both dietary diversity and child growth that could otherwise confound mediation analyses. 
Assumption \ref{asmp:cond_crossworld_indep_mediators} posits that, given $N_i$, $\mathbf{C}_i$, and $\mathbf{D}_i(1)$, an individual's mediator under treatment is independent of the cluster peers' mediators under control. In RPS, several features make this cross-world independence plausible: (i) cluster-level randomization and program implementation generate a common treated environment, with health-service use captured by $\mathbf{D}_i(1)$; (ii) key cluster-level determinants of food access and norms are likely explained by $\mathbf{C}_i$ and, indirectly, by $N_i$, which can be a proxy for the intensity of within-cluster interactions; and (iii) the control-state mediators for peers are not jointly realized when the cluster is treated, limiting direct mechanisms that link $M_{ij}(1)$ to $\mathbf{M}_{i(-j)}(0)$ beyond what is already conditioned on. Of note, we do not rule out dependence among mediators under the same treatment status and between a unit's own cross-world mediators, and Assumption \ref{asmp:cond_crossworld_indep_mediators} is required to identify SME and IME, but not NIE or NDE.

\begin{theorem}
\label{thm:identification}
        Under Assumption \ref{asmp:sutva}--\ref{asmp:si}, for $a, a' \in \{0,1\}$, $\E\left[\frac{1}{N} \sum_{j=1}^{N}  Y_{\cdot j}(a, \mathbf{M}(a')) \right] $ are partially identified up to the conditional distribution of $\mathbf{D}(a') \mid \mathbf{D}(a), A, \mathbf{C},N $ as follows: 
        \begin{align*}
             &\E\left[\frac{1}{N} \sum_{j=1}^{N}  Y_{\cdot j}(a, \mathbf{M}(a')) \right] \\
             &= \; \E_{\mathbf{C},N} \left[\frac{1}{N} \sum_{j=1}^{N}  \int_{\mathbb{R}^{N}}  \int_{\mathbb{R}^{N}} \int_{\mathbb{R}^{N}} \E\left[ Y_{\cdot j} \mid \mathbf{M}=\mathbf{m}, \mathbf{D}=\mathbf{d}, A=a, \mathbf{C},N \right] \right.\\
            &\left. \quad\quad\quad\quad\quad\quad  dF_{\mathbf{M} \mid \mathbf{D} = \mathbf{d}',  A=a', \mathbf{C},N}(\mathbf{m}) dF_{ \mathbf{D}(a') \mid  \mathbf{D}(a) = \mathbf{d}, A=a, \mathbf{C},N}(\mathbf{d}') dF_{\mathbf{D} \mid A=a, \mathbf{C},N}(\mathbf{d}) \right].
        \end{align*}
        Additionally, under Assumption \ref{asmp:sutva}--\ref{asmp:cond_crossworld_indep_mediators}, $\E \left[ \frac{1}{N}\sum_{j=1}^{(n)}  Y_{\cdot j}(1, M_{\cdot j}(1), \mathbf{M}_{(-j)}(0)) \right] $ is partially identified as follows:
        \begin{align*}
            &\E \left[ \frac{1}{N}\sum_{j=1}^{(n)}  Y_{\cdot j}(1, M_{\cdot j}(1), \mathbf{M}_{(-j)}(0)) \right] \\
            &= \;  
             \E_{\mathbf{C},N} \left [\frac{1}{N} \sum_{j=1}^{N}  \int_{\mathbb{R}^{N}}\int_{\mathbb{R}^{N}}\int_{\mathbb{R}^{N-1}}\int_{\mathbb{R}}   \E\left[ Y_{\cdot j}  \mid A=1, M_{\cdot j}=m, \mathbf{M}_{(-j)}=\mathbf{m}_{N-1}, \mathbf{D}=\mathbf{d}, \mathbf{C},N \right] \right. \\
            &\left. \quad\quad\quad\quad\quad\quad\quad\quad\quad\quad\quad\quad \quad\quad 
            dF_{M_{\cdot j} \mid A=1, \mathbf{D}=\mathbf{d}, \mathbf{C},N}(m) 
            dF_{\mathbf{M}_{(-j)} \mid \mathbf{D} = \mathbf{d}',  A=a', \mathbf{C},N}(\mathbf{m}_{N-1})  \right. \\
            &\left. \quad\quad\quad\quad\quad\quad\quad\quad\quad\quad\quad\quad\quad\quad\quad \quad\quad
            dF_{ \mathbf{D}(a') \mid  \mathbf{D}(a) = \mathbf{d}, A=a, \mathbf{C},N}(\mathbf{d}') 
            dF_{\mathbf{D} \mid A=1, \mathbf{C},N}(\mathbf{d}) \right]. &
        \end{align*}
\end{theorem}

\noindent
This identification result generalizes the result of \citet{cheng2024} in the presence of a post-treatment confounder. 
This theorem shows that the estimands are only partially identified because the joint distribution of $\mathbf{D}(a) $ and $ \mathbf{D}(a'))$ is not point identified from the observed data, as these two potential quantities are never jointly observed. In particular, the correlation parameter linking them cannot be estimated from the data.
Notably, this identification result does not require the cross-world independence assumption about the post-treatment confounders as in \citet{ohnishi2025}. Instead, the identification formula leaves the cross-world conditional law $F_{\mathbf{D}(a') \mid \mathbf{D}(a)=\mathbf{d},A=a,\mathbf{C},N}$ unrestricted, thereby relaxing the conditional cross-world independence assumption of \citet{ohnishi2025} and isolating the sole unidentifiable component to this cross-world dependence structure. In our implementation, we model this component through a Gaussian copula indexed by a sensitivity parameter $\rho$, which enables a transparent sensitivity analysis by showing how causal conclusions vary across scientifically plausible cross-world dependence structure of the post-treatment confounders; see Section \ref{sec:sensitivity_analysis_introduction} for a detailed discussion.
Except for the cross-world conditional law $F_{ \mathbf{D}(a') \mid  \mathbf{D}(a) = \mathbf{d}, A=a, \mathbf{C},N}$, all other components can be expressed in terms of the observed data. Therefore, this identification formula motivates a g-computation approach based on models for the observed quantities, such as the outcome, mediator, and post-treatment confounder. Details are provided in Supplementary Material D.

\section{Bayesian Nonparametric Causal Mediation Analysis in CRTs}
\label{sec:BNP}
\subsection{Nested common atoms enriched Dirichlet process}
\label{sec:CA-EDP}
Although multilevel parametric models \citep{VanderWeele2013_CRT, Forastiere2016} are common for hierarchical data,
they remain highly susceptible to bias whenever the functional form is misspecified.
The nDDP \citep{ohnishi2025} extends the classical nDP \citep{Rodriguez2008} by allowing the weights and atoms of the nDP to depend on cluster- and individual-level covariates, while maintaining a hierarchical grouping mechanism that simultaneously quantifies uncertainty and facilitates information-sharing across multiple levels of grouped data. 
Despite its flexibility, the nDDP faces two inferential challenges. First, it is prone to the \emph{degeneracy} problem \citep{Camerlenghi2019_2}: if two distributions share even a single atom, they are assigned to the same mixture component. This issue, inherent to nDP priors, can constrain their flexibility in applications to CRTs where clusters often share common structure but differ in sparse, site-specific ways. Degeneracy also leads to over-pooling and reduces the effective number of independent mixture components, which in turn underestimate between-cluster variability and attenuate treatment contrasts.

In addition, estimating the causal mediation and spillover effects in Section \ref{sec:identification_cwi} requires a  model for the joint distribution $\pr( \mathbf{Y}_i, \mathbf{M}_i, \mathbf{D}_i, \mathbf{X}_i,\mathbf{V}_i, A_i,N_i)$. By contrast, the nDDP models the conditional outcome distribution $\pr(\mathbf{Y}_i, \mathbf{M}_i, \mathbf{D}_i \mid \mathbf{X}_i, \mathbf{V}_i, N_i)$ and estimates $\pr(\mathbf{X}_i,\mathbf{V}_i, N_i)$ with its empirical distribution. Consequently, the inferential target shifts from population-level treatment effects to its approximation, conditional mixed treatment effects \citep{FanLi2023}. 
This practice may fail to fully propagate uncertainty in the covariate distribution with small samples and a limited number of clusters, as is common in CRTs, the empirical distribution can be a suboptimal approximation to the population distribution. We discuss these trade-offs in the Supplementary Material Section C.

To allow for information sharing across clusters without the degeneracy issue and avoiding the dominance of potentially high-dimensional covariates, we introduce a new BNP prior for CRTs, the Nested Common Atom Enriched Dirichlet Process (CA-EDP). The CA-EDP retains the hierarchical strengths of the nDP, while avoiding the degeneracy issue by explicitly modeling the commonality of atoms between subgroups, mirroring the structure of the Common Atom Model (CAM) \citep{Denti2023}. 
Within the CAM structure, the CA-EDP incorporates the enrichment strategy of the Enriched Dirichlet Process (EDP) \citep{Wade2011} for within-group distribution, where the distribution of covariate parameters $\boldsymbol{\phi}_{ij}$ is nested within the groups of the response parameters $\boldsymbol{\theta}_{ij}$. This strategy stabilizes inference for the joint distributions of the outcome and potentially high-dimensional covariates by avoiding the unfavorable random partition, 
dominated by the high-dimensional covariates \citep{Wade2014}.
Specifically, we consider the following hierarchical model:
\begin{equation}
\label{eq:model}
    \begin{split}
        Y_{ij} \mid \mathbf{M}_{i}, \mathbf{D}_{i},\mathbf{X}_{ij}, \mathbf{V}_i, N_i, A_i; \boldsymbol{\theta}_{ij}^{(y)} &\sim f(y_{ij} \mid  \mathbf{m}_{i}, \mathbf{d}_{i}, \mathbf{x}_{ij}, \mathbf{v}_i, n_i, a_i; \boldsymbol{\theta}_{ij}^{(y)}), \\
        M_{ij} \mid \mathbf{D}_{i},\mathbf{X}_{ij}, \mathbf{V}_i, N_i, A_i; \boldsymbol{\theta}_{ij}^{(m)} &\sim f(m_{ij} \mid  \mathbf{d}_{i}, \mathbf{x}_{ij}, \mathbf{v}_i, n_i, a_i; \boldsymbol{\theta}_{ij}^{(m)}), \\
        D_{ij} \mid \mathbf{X}_{ij}, \mathbf{V}_i, N_i, A_i; \boldsymbol{\theta}_{ij}^{(d)} &\sim f( d_{ij} \mid \mathbf{x}_{ij}, \mathbf{v}_i, n_i, a_i; \boldsymbol{\theta}_{ij}^{(d)}), \\
        \mathbf{X}_{ij} \mid \mathbf{V}_i, N_i; \boldsymbol{\phi}_{ij} &\sim f( \mathbf{x}_{ij} \mid \mathbf{v}_i, n_i; \boldsymbol{\phi}_{ij}),  \\
        (\boldsymbol{\theta}_{ij}, \boldsymbol{\phi}_{ij}) \mid F_i &\sim F_i, \\
        \mathbf{V}_{i} \mid N_i; \boldsymbol{\eta}_i^{(v)} &\sim f( \mathbf{v}_i \mid n_i; \boldsymbol{\eta}_{i}^{(v)}), \\
        N_i \mid \boldsymbol{\eta}_{i}^{(n)} &\sim f(  n_i \mid \boldsymbol{\eta}_{i}^{(n)}),\\
        (F_i, \boldsymbol{\eta}_i) &\sim \sum_{k=1}^{\infty} \pi^{*}_k  \delta_{(F^{*}_{k}, \boldsymbol{\eta}_k^*)},
    \end{split}
\end{equation}
where $\boldsymbol{\theta}_{ij}=(\boldsymbol{\theta}_{ij}^{(y)},\boldsymbol{\theta}_{ij}^{(m)},\boldsymbol{\theta}_{ij}^{(d)})$, 
$\boldsymbol{\eta}_i = (\boldsymbol{\eta}_{i}^{(n)},\boldsymbol{\eta}_{i}^{(v)})$, $\pi^{*}_k = s^{*}_k \prod_{j=1}^{k-1}(1- s^{*}_j)$, $s^{*}_j \sim \mathrm{Beta}(1,\astar)$, $\boldsymbol{\eta}_k^* \sim H_{\eta}$ for a base measure $H_{\eta}$ defined over $\mathcal{H}$. 

An important feature of the proposed model is that the within-cluster distribution $F_{i}$ and parameters for the cluster-level covariates $\boldsymbol{\eta}_i$ constitute the atoms for the first layer stick-breaking representation, enabling joint modelling of  cluster-level variables with individual-level variables. 
Another important feature is that each atomic distribution $F^{*}_{k}$ is constructed using an EDP-style formulation, which creates a mixture over separate sets of atoms for the outcome parameters and the covariate parameters:
$F^{*}_{k}(\cdot) = \sum_{l=1}^{\infty}\sum_{m=1}^{\infty} w^{\theta}_{kl}w^{\phi}_{klm} \delta_{\boldsymbol{\theta}^{*}_{l}}(\cdot)\delta_{\boldsymbol{\phi}^{*}_{m}}(\cdot)$ 
with  
$w_{kl}^{\theta} = v_{kl}^{\theta} \prod_{j<l}(1-v_{kj}^{\theta})$, 
$v_{kl}^{\theta} \sim \mathrm{Be}(1,\atheta)$, $w_{k1}^{\theta}=v_{k1}^{\theta}$ for each $k$,   
$w_{klm}^{\phi} = v_{klm}^{\phi} \prod_{j<m}(1-v_{klj}^{\phi})$, 
$v_{klm}^{\phi} \sim \mathrm{Be}(1,\aphi)$, $w_{kl1}^{\phi}=v_{kl1}^{\phi}$ for each $k,l$.
A critical feature of this formulation, and  the EDP, is that the mixture weights for the covariate parameters, $w^{\phi}_{klm}$, has subpartitions indexed by $l$ nested under the partition $k$. This dependence nests the covariate-parameter mixture for $\boldsymbol{\phi}$ within the response groups defined by $\boldsymbol{\theta}$, allowing the covariate distribution to adapt to (rather than dominate) the response model. 
Additionally, to avoid degeneracy of the nested prior, we adopt the CAM strategy within the EDP-induced random partition; in particular, we draw the common atoms from shared base measures $G_{\theta}$ and $G_{\phi}$: $\boldsymbol{\theta}^{*}_{l}\sim G_{\theta}$, $\boldsymbol{\phi}^{*}_{m}\sim G_{\phi}$.
The EDP-style formulation induces a prior for the random joint distribution $G_0$ through the joint law of the marginal and conditionals and the mapping $(G_{\theta}, G_{\phi}) \to \int G_{\phi}(\cdot \mid \theta ) dG_{\theta}$. 
Then, the prior is parameterized by the base measure $G_0$, defined by $G_0(A \times B) = \int_A G_{\phi}(B \mid \theta ) dG_{\theta}$.

The CA-EDP has three concentration parameters $\astar$, $\atheta$ and $\aphi$, while the nDDP, CAM and EDP only have two concentration parameters. The number of groups for cluster-level variables ($N_i$ and $\mathbf{V}_i$) depends on $\astar$, the number of groups for outcomes, mediators and post-treatment confounders ($Y_{ij}$, $M_{ij}$ and $D_{ij}$) depends on $\atheta$, and the number of groups for individual-level covariates ($\mathbf{X}_{ij}$) depends on $\aphi$. In what follows, we use $v$-group, $y$-group, and $x$-group to represent different groups induced by CA-EDP for simplicity. 
The two nested levels induced by three concentration parameters serve a distinct purpose. At the first level, the random partition generated by $\pi^{*}_k$ (inspired by the CAM construction) captures cluster heterogeneity, while allowing for information sharing across clusters. Each atom of this Dirichlet process is a pair $(F_i, \boldsymbol{\eta}_i)$, where $\boldsymbol{\eta}_i$ is the parameter for the cluster‑level covariate distribution and $F_i$ represents the data‑generating distribution of individual observations within cluster $i$.
Then, the second level of nested structure is introduced into $F_i$ (inspired by the EDP construction),  which avoids the random partition within a subgroup being dominantly determined by the high-dimensional covariates \citep{Wade2014}, while explicitly modeling the commonality of atoms between subgroups (inspired by the CAM).

To operationalize this prior for causal mediation analysis of CRTs, within each of the above groups, we assume (simple) generalized linear models (GLM) for $Y_{ij}$, $M_{ij}$, $D_{ij}$, $\mathbf{X}_{ij}$, and $\mathbf{V}_{i}$.
For example, if $Y_{ij}$ is continuous, we specify $Y_{ij} \mid \mathbb{C}_{ij}^{(y)}; \boldsymbol{\beta}^{(y)}, \sigma^{(y)} \sim \mathrm{N}(\mathbb{C}_{ij}^{(y)}\boldsymbol{\beta}^{(y)}, (\sigma^{(y)})^2)$, where $\boldsymbol{\theta}^{(y)} = (\boldsymbol{\beta}^{(y)}, \sigma^{(y)})$ and $\mathbb{C}_{ij}^{(y)}=(g^{(m)}(\mathbf{M}_{i}), g^{(d)}(\mathbf{D}_{i}), \mathbf{X}_{ij}, \mathbf{V}_i, N_i, A_i)$ with $g^{(m)}(\cdot)$ and $g^{(d)}(\cdot)$ being prespecified summary functions of within-cluster mediators and post-treatment confounders. Since the dimensions of $\mathbf{M}_i$ and $\mathbf{X}_i$ can vary across clusters in
CRTs, we consider adjusting for summary functions with fixed dimensions in the models. Specifically, we consider a bivariate summary function
$g^{(m)}(\mathbf{M}_{i})=\left\{ M_{ij}, \frac{1}{ |\mathbf{M}_{i}|- 1} \sum_{\substack{k=1 \\ k \ne j}}^{|\mathbf{M}_{i}|} M_{ik} \right\}$ and 
$g^{(d)}(\mathbf{D}_{i})=\left\{ D_{ij}, \frac{1}{ |\mathbf{D}_{i}|- 1} \sum_{\substack{k=1 \\ k \ne j}}^{|\mathbf{D}_{i}|} D_{ik} \right\}$ \citep[e.g., ][]{Ogburn2024,cheng2024, Ohnishi2025_DOI}. The models for $Y_{ij}$ incorporate the summaries of the full vectors $ \mathbf{M}_{i}$ and $\mathbf{D}_{i}$ to capture within-cluster peer effects.
Similarly, we assume GLMs for $M_{ij}$ and $D_{ij}$ with designed vectors 
$\mathbb{C}_{ij}^{(m)}=(g^{(d)}(\mathbf{D}_{i}), \mathbf{X}_{ij}, \mathbf{V}_i, N_i, A_i)$ and 
$\mathbb{C}_{ij}^{(d)}=(\mathbf{X}_{ij}, \mathbf{V}_i, N_i, A_i)$. Additionally, we assume $N_i \mid \lambda_{i}^{(n)} \sim \mathrm{Pois}(\lambda_{i}^{(n)})$, where $\boldsymbol{\eta}_{i}^{(n)} = \lambda_{i}^{(n)}$.
For the base measures, we consider the conjugate prior specification. 
In particular, let 
\begin{equation}
\label{eq:base_measure_hyperparameters}
\begin{split}
    G_{\theta} =& \mathrm{MVN}(\boldsymbol{\mu}_{\boldsymbol{\beta}^{(y)}},\boldsymbol{\Sigma}_{\boldsymbol{\beta}^{(y)}})\mathrm{IG}(a_{\sigma^{(y)}},b_{\sigma^{(y)}}) 
    \mathrm{MVN}(\boldsymbol{\mu}_{\boldsymbol{\beta}^{(m)}},\boldsymbol{\Sigma}_{\boldsymbol{\beta}^{(m)}})\mathrm{IG}(a_{\sigma^{(m)}},b_{\sigma^{(m)}}) \\
    & \mathrm{MVN}(\boldsymbol{\mu}_{\boldsymbol{\beta}^{(d)}},\boldsymbol{\Sigma}_{\boldsymbol{\beta}^{(d)}})\mathrm{IG}(a_{\sigma^{(d)}},b_{\sigma^{(d)}}), \text{ and } G_{\phi} = \mathrm{N}(\mu_{x},\boldsymbol{\Sigma}_{x})\mathrm{IG}(a_{\sigma^{x}},b_{\sigma^{x}}).
\end{split}
\end{equation}

\noindent
We discuss the choice of each hyperparameter in simulation studies and the data analysis.

\subsection{Distributional properties}
\label{sec:properties}

This section provides distributional properties for the CA-EDP model. All proofs are deferred to Supplementary Material A.2. 
First, the following properties hold for $i \neq i'$: $\pr(F_{i} = F_{i'})  = \pr(\boldsymbol{\eta}_{i} = \boldsymbol{\eta}_{i'}) = \frac{1}{1+\astar}$. 
This is the probability that clusters $i$ and $i'$ share the same random measure and cluster-level atom. It follows from the top-level stick-breaking representation. The probabilities of a tie between atoms from the  measures for two cluster $i$ and $i'$, $F_{i}, F_{i'}$, are 
\begin{align*}
        q_{\theta} &:= \pr(\boldsymbol\theta_{ij} = \boldsymbol\theta_{i'j'})  = \frac{1}{1+\astar}\qty(\frac{1}{1+\atheta} + \frac{\astar}{1+2\atheta}), \\  
        q_{\phi} &:= \pr(\boldsymbol\phi_{ij} = \boldsymbol\phi_{i'j'})  = \frac{1}{(1+\astar)(1+\atheta)}\qty(\frac{1}{1+\aphi} + \frac{\astar + \atheta +\astar \atheta}{1+2\aphi}),\\  
        q_{\theta, \phi} &:= \pr(\boldsymbol\theta_{ij} = \boldsymbol\theta_{i'j'}, \boldsymbol\phi_{ij} = \boldsymbol\phi_{i'j'}) = \frac{1}{1+\astar}\qty(\frac{1}{(1+\atheta)(1+\aphi)} + \frac{\astar}{(1+2\atheta)(1+2\aphi)}). 
\end{align*}
Compared with $\pr(\boldsymbol\theta_{ij} = \boldsymbol\theta_{i'j'})  = \frac{1}{(1+\astar)(1+\atheta)}$ of the nDP \citep{Rodriguez2008}, the second term of $q_{\theta}$ shows the allowance for the information borrowing across clusters. The tie probabilities $q_{\phi}$ and $q_{\theta, \phi}$ further indicate that the prior clusters $x$-group data and allows for information sharing across clusters. 

Next, we further investigate the correlation between two EDP measures, $F_{i}$ and $F_{i'}$, on different Borel sets $(A_{\theta},A_{\phi}),(B_{\theta},B_{\phi}) \in \Theta \times \Phi$. The covariance and correlation are useful quantities to investigate the dependence across random probability measures and their suitability for practical applications. 
Let 
$\Delta_{\theta}(A_{\theta}, B_{\theta}) = G_{\theta}(A_{\theta} \cap B_{\theta}) - G_{\theta}(A_{\theta})G_{\theta}(B_{\theta})$ 
and 
$\Delta_{\phi}(A_{\phi}, B_{\phi}) = G_{\phi}(A_{\phi} \cap B_{\phi}) - G_{\phi}(A_{\phi})G_{\phi}(B_{\phi})$. The covariance of $F_{i}$ and $F_{i'}$ on different Borel
sets can be characterized as: 
\begin{align*}
    \cov\left( F_{i}\left( A_{\theta}, A_{\phi} \right),F_{i'} \left( B_{\theta}, B_{\phi} \right) \right) 
        &= q_{\theta}  G_{\phi}(A_{\phi})G_{\phi}(B_{\phi}) \Delta_{\theta}(A_{\theta}, B_{\theta})  \\
        + q_{\phi} &  G_{\theta}(A_{\theta})G_{\theta}(B_{\theta}) \Delta_{\phi}(A_{\phi}, B_{\phi}) 
        + q_{\theta, \phi}  \Delta_{\theta}(A_{\theta}, B_{\theta}) \Delta_{\phi}(A_{\phi}, B_{\phi}).
\end{align*}
This covariance structure is composed of three terms, corresponding to the scenarios where heterogeneity arises from the $\boldsymbol\theta$ parameters, the $\boldsymbol\phi$ parameters, or both, reflecting the hierarchical nature of the model. For $A_{\phi}=B_{\phi}=\Phi$, we have $\cov \left(F_{i} \left(A_{\theta}, \Phi \right),F_{i'} \left( B_{\theta}, \Phi  \right) \right) = q_{\theta}  \Delta_{\theta}(A_{\theta}, B_{\theta})$, 
with the correlation on the same set $A_{\theta}$ given by 
\begin{equation*}
    \varrho_{\theta} = \corr \left(F_{i} ( A_{\theta}, \Phi ),F_{i'} ( A_{\theta}, \Phi ) \right) = 1 - \frac{\astar}{1+\astar}\frac{\atheta}{1+2\atheta}.
\end{equation*}
These covariance and correlation expressions reduce to Equation (8) and (9) in \citet{Denti2023}, respectively. Similarly, for $A_{\theta}=B_{\theta}=\Theta$, we have $\cov \left(F_{i} \left( \Theta, A_{\phi} \right),F_{i'} \left( \Theta, B_{\phi} \right) \right) = q_{\phi}  \Delta_{\phi}(A_{\phi}, B_{\phi})$, and 
the correlation on the same set $A_{\phi}$ is given by 
\begin{equation*}
    \varrho_{\phi} = \corr \left(F_{i} ( \Theta, A_{\phi} ),F_{i'} ( \Theta, A_{\phi} ) \right) = 1 - \frac{\astar}{1+\astar}\frac{\aphi}{1+\atheta + \atheta \aphi + 2\aphi}.
\end{equation*}
We can see both $\varrho_{\theta}, \varrho_{\phi} \in (1/2,1)$, due to the sharing of the atoms. 
For the nDP and nDDP, the correlation does not depend on $\atheta$, as they assume independence between atoms in separate distributions.

Finally, we examine the random partition structure induced by our model. The partially exchangeable partition probability function (pEPPF; \citet{Camerlenghi2019}) provides a convenient formalism. 
Let $\widetilde{F}_i := F_i \otimes \delta_{\eta_i}$ be a random product measure on $\Theta \times \Phi \times \mathcal{H}$ induced by the CA-EDP model.
For simplicity, set $\boldsymbol\xi_{ij}=(\boldsymbol\theta_{ij},\boldsymbol\phi_{ij},\boldsymbol\eta_{i})$. 
Let $\boldsymbol\xi_1^{*},\ldots,\boldsymbol\xi_s^{*}$ denote the $s$ distinct values observed in $\{ \boldsymbol\xi_{ij} \}$, and for each group $i$ let $\mathbf{n}_i=(n_{i1},\ldots,n_{is})$ be the vector of their within-group frequencies. Following \citet{Camerlenghi2019}, the pEPPF is the probability of this specific allocation of observations to clusters, defined as $\Pi^{(s)} (\mathbf{n}_1,\ldots,\mathbf{n}_I) = \E \int_{\Xi} \prod_{i=1}^{I}\prod_{j=1}^{s} d\widetilde{F}_{i}^{n_{ij}}(d\boldsymbol\xi_{j}^{*}),$
where $\E$ represents the marginal expectation with respect to the random measure $\widetilde{F}_i$. When $I=1$, this reduces to the usual exchangeable partition probability function (EPPF) for an individual sample \citep{Pitman1995}, denoted here by $\Psi^{(s)}(\cdot)$. 
For our model, the pEPPF can be decomposed based on whether groups share the same
underlying atom-generating distribution $F_{k}^{*}$.
\begin{proposition}
\label{prop:peppf}
For $I=2$, the pEPPF of the proposed model is given by:
    \begin{equation}
    \label{eq:peppf_CA-EDP}
        \Pi^{(s)} (\mathbf{n}_1,\mathbf{n}_2) = q_1 \Psi^{(s)}(\mathbf{n}_1 + \mathbf{n}_2) + (1-q_1) \int_{\Xi} \E \prod_{i=1}^{2}\prod_{j=1}^{s} \widetilde{F}_{i}^{n_{ij}}(d\boldsymbol\xi_{j}^{*}),
    \end{equation}
    where $q_1 = (1+\astar)^{-1}$ is the probability that both groups share the same atom $(F_{k}^{*}, \boldsymbol\eta_{k}^{*})$. The integral in Equation \ref{eq:peppf_CA-EDP} is strictly positive.
\end{proposition}
\noindent
Proposition \ref{prop:peppf} is proved by conditioning on the top-level grouping structure, following \citet{Denti2023}. The first term corresponds to the event that groups 1 and 2 are assigned to the same top-level group, under which all observations are fully exchangeable. The second term, weighted by the probability $1-q_1$ that the groups fall in different groups, is what enforces partial exchangeability. In particular, the strict positivity of the integral term keeps the pEPPF from collapsing to the EPPF, thereby preventing degeneracy to full exchangeability.
The nested EDP-type specification within the CAM makes the partition probabilities depend on group labels, extending the CAM partial-exchangeability result of \citet{Denti2023}. Because the model does not collapse to full exchangeability, it can flexibly capture both within-group and between-group heterogeneity as intended.

\subsection{Posterior inference}
\label{sec:posterior_inference}
To generate posterior samples, we use a blocked Gibbs sampler \citep{Ishwaran2000} based on a three-level truncation approximation for the CA-EDP, following the truncation approximation for an EDP considered by \citet{Burns2023}. The blocked Gibbs sampler offers potential improvement in mixing and simple implementation.
To implement the algorithm, we first choose conservative upper bounds $K$, $L$, and $M$ for the numbers of $v$-groups, $y$-groups, and $x$-groups, respectively.
The algorithm then uses data augmentation, iterating between sampling model parameters given nested group membership and sampling group memberships given the model parameters. 
For each posterior sample, we compute the corresponding posterior sample of the estimands using  g-computation based on Theorem \ref{thm:identification}.
The full details of the algorithm are provided in the Supplementary Material Section D.

The finite approximation allows a tractable Gibbs sampler and posterior samples of the distributions needed for g-computation. However, it is important to quantify the error introduced by approximating the infinite-dimensional process. Let $\widetilde{F}_i^{(K,L,M)}$ be the finite-dimensional approximation of $\widetilde{F}_i$, defined as the random measure corresponding to the truncated generative process with truncation levels $K$,
$L$, and $M$. The following theorem bounds the distance between these two measures in expectation.

\begin{theorem} 
\label{thm:finite_approximation}
The expected total variation distance between the true random measure $\widetilde{F}_i$ and its finite approximation $\widetilde{F}_i^{(K,L,M)}$ is bounded as follows:
    \begin{equation}
        \E\qty[TV(\widetilde{F}_i, \widetilde{F}_i^{(K,L,M)})] \leq \qty(\frac{\astar}{1+\astar})^K + \qty(\frac{\atheta}{1+\atheta})^L + \qty(\frac{\aphi}{1+\aphi})^M.
    \end{equation}
\end{theorem}
\noindent
The proof is provided in Supplementary Material Section A.2.
This result shows that the error in approximating the random measure itself can be made arbitrarily small by choosing sufficiently large truncation levels $K$, $L$, and $M$.  
Although Theorem \ref{thm:finite_approximation} confirms that our approximate prior converges to the true one, it is also helpful to investigate the quality of this approximation to the marginal distribution of the observed data. A small distance between the prior measures is a prerequisite, but does not automatically guarantee that the resulting data distributions will also be close. 

For clarity,  let the pair $(N_i, \mathbf V_i, \mathbf X_{ij} ,Y_{ij} , M_{ij}, D_{ij} )$ represent the cluster size, cluster-level and individual-level covariates, outcomes, mediators and post-treatment confounders. 
For notational simplicity, we use here $y_{ij}$ to denote the joint outcome vector $(Y_{ij} , M_{ij}, D_{ij} )$, $x_{ij}$ to denote the individual covariates $\mathbf X_{ij}$, and $v_i$ to denote the joint cluster covariate vector $(\mathbf V_i,N_i)$ .
The marginal densities under the true and approximate models are defined as:
\begin{align*}
    m(V,X,Y) &= \E\qty[\prod_{i=1}^{I}\prod_{j=1}^{N_i} \int_{\Theta \times \Phi \times \mathcal{H}} f(y_{ij}, x_{ij},  v_i \mid \boldsymbol\theta, \boldsymbol\phi, \boldsymbol\eta) \widetilde{F}_i(d\boldsymbol\theta, d\boldsymbol\phi, d\boldsymbol\eta)] \\
    m^{(K,L,M)}(V,X,Y) &= \E\qty[\prod_{i=1}^{I}\prod_{j=1}^{N_i} \int_{\Theta \times \Phi \times \mathcal{H}} f(y_{ij}, x_{ij},  v_i \mid \boldsymbol\theta, \boldsymbol\phi, \boldsymbol\eta) \widetilde{F}_i^{(K,L,M)}(d\boldsymbol\theta, d\boldsymbol\phi, d\boldsymbol\eta)] 
\end{align*}
Here, $V$, $X$ and $Y$ represent the collections of all cluster sizes and cluster-level covariates, individual-level covariates, and outcome data (outcomes, mediators and post-treatment confounders) across all clusters and indviduals, respectively.

\begin{theorem} 
\label{thm:finite_approximation_marginal}
The total variation distance between the marginal densities of the data under the true and approximate models is bounded by:
    \begin{equation}
        TV(m(V,X,Y), m^{(K,L,M)}(V,X,Y)) \leq N \qty(\qty(\frac{\astar}{1+\astar})^K + \qty(\frac{\atheta}{1+\atheta})^L + \qty(\frac{\aphi}{1+\aphi})^M),
    \end{equation}
    where $N=\sum_{i=1}^{I}N_i$ is the total number of observations.
\end{theorem}

Theorem \ref{thm:finite_approximation_marginal} provides an explicit error bound for the marginal data distribution. The bound depends on the same tail probabilities of the stick-breaking processes as in Theorem \ref{thm:finite_approximation}, but is scaled by the total sample size $N$. Note that in \citet{Burns2023}, the bound consists of the second and third terms (for an EDP model); in \citet{Denti2023}, the bound involves the first and second terms (for a CAM model). Our result naturally combines all three, reflecting our built-in three-level hierarchy of our prior formulation.

\subsection{Sensitivity analysis for cross-world dependence}
\label{sec:sensitivity_analysis_introduction}
To implement the g-computation algorithm, we need to construct the unidentified conditional $\mathbf{D}_i(1-a)\mid \mathbf{D}_i(a),\mathbf{C}_i,N_i$ in Theorem \ref{thm:identification}. Here, we introduce a working multivariate Gaussian copula \citep{Masarotto2012} on the joint $(\mathbf{D}_i(1),\mathbf{D}_i(0))\mid (\mathbf{C}_i,N_i)$. 
In particular, for any $N_i$-dimensional vectors $\mathbf{d}=(d_1,\ldots, d_{N_i})^\top$ and $\mathbf{d}'=(d_1',\ldots, d_{N_i}')^\top$,
    \begin{equation}
    \label{eq:gaussian_copula}
    \begin{split}
        &F_{\mathbf{D}_{i}(1),\mathbf{D}_{i}(0)}(\mathbf{d},\mathbf{d}' \mid \mathbf{C}_{i},N_i) \\
        &= \; \Phi_{2N_i} \Bigl[  \Phi_1^{-1}\qty{F_{D_{i1}(1)}(d_1) \mid \mathbf{C}_{i1},N_i) }, 
        \Phi_1^{-1}\qty{F_{D_{i1}(0)}(d_1') \mid \mathbf{C}_{i1},N_i) }, \ldots,  \\
        &  \quad\quad\quad\quad\quad
        \Phi_1^{-1}\qty{F_{D_{iN_i}(1)}(d_{N_i}) \mid \mathbf{C}_{iN_i},N_i) }, 
        \Phi_1^{-1}\qty{F_{D_{iN_i}(0)}(d_{N_i}') \mid \mathbf{C}_{iN_i},N_i) }, \,  \boldsymbol{\Omega} \Bigr],
    \end{split}
    \end{equation}
    where $\Phi_1$ is the univariate standard normal cumulative distribution function (CDF) and $\Phi_{2N_i}$ is the multivariate normal CDF with a  $(2N_i \times 2N_i)$ correlation matrix $\boldsymbol{\Omega}$. 
    The correlation matrix $\boldsymbol{\Omega}$ is block-partitioned as a $2 \times 2$ block matrix, where each block is of size $N_i \times N_i$, $\boldsymbol{\Omega} = 
    \begin{pmatrix} 
    \mathbf{C}_{11} & \mathbf{C}_{10} \\ 
    \mathbf{C}_{10}^\top & \mathbf{C}_{00} 
    \end{pmatrix},$
    where 
    \begin{equation*}
    \begin{split}
        \mathbf{C}_{11} = 
        \begin{pmatrix} 
            1 & \gamma_1 & \cdots & \gamma_1 \\ 
            \gamma_1 & 1 & \cdots & \gamma_1 \\
            \vdots & \vdots & \ddots & \vdots \\
            \gamma_1 & \cdots & \cdots & 1 \\ 
        \end{pmatrix},
        \mathbf{C}_{10} = 
        \begin{pmatrix} 
            \rho & \rho^{*} & \cdots & \rho^{*} \\ 
            \rho^{*} & \rho & \cdots & \rho^{*} \\
            \vdots & \vdots & \ddots & \vdots \\
            \rho^{*} & \cdots & \cdots & \rho \\ 
        \end{pmatrix},
        \mathbf{C}_{00} = 
        \begin{pmatrix} 
        1 & \gamma_0 & \cdots & \gamma_0 \\ 
        \gamma_0 & 1 & \cdots & \gamma_0 \\
        \vdots & \vdots & \ddots & \vdots \\
        \gamma_0 & \cdots & \cdots & 1 \\ 
        \end{pmatrix}.
    \end{split}
    \end{equation*}

\noindent
Note that we do not assume that this copula represents the true joint law; it is a working model used solely to derive the conditional distribution of $\mathbf{D}_i(1-a)$ given $\mathbf{D}_i(a)$, $\mathbf{C}_i$, and $N_i$.
Equation \eqref{eq:gaussian_copula} specifies the dependence structure among post-treatment confounders. Imposing within-cluster exchangeability, we assume the same dependence structure for any pair of individuals in a cluster, encoded by equicorrelation block matrices.  
The partial identifiability comes from the observed data not identifying the parameters in $\mathbf{C}_{10}$, $\rho$ and $\rho^{*}$, because $D_{ij}(a)$ and $D_{ij}(a')$ are never jointly observed.

We adopt a sensitivity analysis approach inspired by \citet{Roy2024}. The idea is to keep the BNP inference unchanged and introduce copula-based dependence via \eqref{eq:gaussian_copula} in post-processing by estimating the same-world correlation parameters $\gamma_0$ and $\gamma_1$ and sampling the cross-world correlation parameter $\rho$ from its prior.
For the cross-world inter-individual correlation $\rho^{*}$, we adopt parameterizations from prior work \citep[e.g.,][]{Zigler2012,Kim2019,Roy2024} tailored to the inter-individual correlation in CRTs, and set $\rho^{*} = \frac{\gamma_0+\gamma_1}{2}\times \rho$.
This parameterization implies that: (i) for the same individual $(j=j')$, the cross-world correlation equals $\rho$; and (ii) for different individuals $(j\neq j')$, the cross-world correlation equals an attenuated average of the corresponding same-world correlations.

To ensure $\boldsymbol{\Omega}$ is positive definite under this parameterization, the following condition needs to be met.
\begin{equation}
\label{eq:cond_posdef}
    \begin{split}
    -\frac{1}{N-1} < \gamma_a < 1, ~~
    \rho^2 < \min \qty{\frac{4 (1-\gamma_1)(1-\gamma_0)}{(2-\gamma_1-\gamma_0)^2},
    \frac{4\left(1 + (N-1)\gamma_1 \right)\left(1+ (N-1)\gamma_0 \right)}{\left(2+ (N-1)(\gamma_1+\gamma_0) \right)^2}},
    \end{split}
\end{equation}
This condition is derived in Supplementary Material Section A.4. 
The bounds in \eqref{eq:cond_posdef} can inform the prior specification for the correlation parameters, but one may use tighter bounds in practice based on substantive knowledge. For example, one can specify $0<\gamma_a<1$ and  
\begin{equation}
\label{eq:prior_rho}
    \rho \sim \mathrm{Unif}\qty(0, \min \qty{\frac{2 \sqrt{(1-\gamma_1)(1-\gamma_0)}}{(2-\gamma_1-\gamma_0)},
    \frac{2\sqrt{\qty(1 + (N-1)\gamma_1)\qty(1+ (N-1)\gamma_0)}}{\qty(2+ (N-1)(\gamma_1+\gamma_0))}}).
\end{equation}
These are justified in our application, in which within-cluster checkup conditions are positively correlated due to shared environments, providers, and norms, and individuals' underlying health propensities persist across treatment states in a cash-transfer program. 
Note that $\gamma_a$ is estimated from the observed marginals in the post-processing step before simulating counterfactual outcomes (i.e., g-computation), because the observed data are informative about same-world correlation. By contrast, $\rho$ is treated as a sensitivity parameter because it is not identifiable from observed data. During g-computation, $\rho$ is drawn from \eqref{eq:prior_rho}, whose upper bound depends on $N$, the synthetic cluster size generated as part of the g-computation procedure. Finally, by comparing with a fixed $\rho=0$, we can observe how sensitive the causal effects are to the unidentifiable cross-world component. The detailed sensitivity analysis procedure is given in Supplementary Material Section D.2.

\begin{remark}
    \emph{The copula model for $(\mathbf D_i(1),\mathbf D_i(0))$ is introduced to construct the unidentified conditional $\mathbf{D}_i(1-a)\mid \mathbf{D}_i(a),\mathbf{C}_i,N_i$ needed for g-computation.  The CA-EDP in Section \ref{sec:CA-EDP} provides the marginals $F_{\mathbf{D}_i(a)}$. Copula parameters (e.g., same-world equicorrelation in $\mathbf{C}_{00}$ and $\mathbf{C}_{11}$,  or cross-world parameters in  $\mathbf{C}_{10}$) are estimated or treated as sensitivity parameters, and positive-definiteness of $\boldsymbol{\Omega}$ is enforced by the  constraints given above.}
\end{remark}

\section{Simulation Studies}
\label{sec:simulation}

We assess the frequentist performance of the CA-EDP for the mediation estimands, NIE and SME, against two benchmarks: the nDDP of \citet{ohnishi2025} and a parametric Bayesian model, the latter of which is common practice.  
For the data-generating process in Scenario S1-S6, we use a hierarchical mixture mechanism at both the cluster and individual levels for the outcome and mediator. Each mixture component contains non-linear and interaction terms with non-Gaussian error terms. Scenario S7 uses a simple linear data-generating process, for which the paramteric model is correctly specified. See Supplementary Material Section E for full details. We let the individual level covariate distributions and the cluster-size distributions vary across seven scenarios below:
\begin{itemize}[nosep]
    \item \textbf{Scenario S1 (baseline).}  
    The cluster size is discrete-uniform, $N_i\sim\mathrm{Unif}\{20,\ldots, 40\}$.  
    The covariates are generated as $(X_{1},X_{2},X_{3})^\top \sim\mathrm{MVN}\qty( (0.5,0.0,-0.5)^\top
    , \qty(\begin{smallmatrix} 
    1.0 &0.2&0.2  \\ 
    0.2 &1.0&0.2  \\ 
    0.2 &0.2&1.0  \\ 
    \end{smallmatrix}))$. The number of clusters is $I=40$.

    \item \textbf{Scenario S2 (mixture of clusters and covariates).}  
    With probability $0.8$, $N_i\sim\mathrm{Pois}(15)$, otherwise $N_i\sim\mathrm{Pois}(30)$.  Additionally, with probability $0.8$, generate $(X_{1},X_{2},X_{3})^\top \sim\mathrm{MVN}\qty( (-1.0,-1.5,-0.5)^\top
    , \qty(\begin{smallmatrix} 
    1.0 &0.2&0.2  \\ 
    0.2 &1.0&0.2  \\ 
    0.2 &0.2&1.0  \\ 
    \end{smallmatrix}))$, and otherwise, generate $(X_{1},X_{2},X_{3})^\top \sim\mathrm{MVN}\qty( (1.5,1.0,0.5)^\top
    , \qty(\begin{smallmatrix} 
    1.0 &0.4&0.4  \\ 
    0.4 &1.0&0.4  \\ 
    0.4 &0.4&1.0  \\ 
    \end{smallmatrix}))$,
    giving a bimodal, cluster-level latent-group structure. $I=40$.

    \item \textbf{Scenario S3 (moderate dimension, some irrelevant covariates).}  
    We return to the baseline cluster size and the original three covariates of S1.  
    Five additional noise covariates are added: $X_{p}\sim\mathrm{N}(0,1.0^{2})$ for $p=4,\dots,8$.  
    These extra variables have no effects on the outcome and mediators, allowing us to probe each model's robustness to irrelevant predictors. $I=40$.

    \item \textbf{Scenario S4 (high dimension, many irrelevant covariates).}    
    Similar to Scenario S3, more noise covariates are added: $X_{p}\sim\mathrm{N}(0,1.0^{2})$ for $p=4,\dots,15$.  $I=40$.

    \item \textbf{Scenario S5 (fewer clusters).}  The same distributions as in S1 with $I=20$.

    \item \textbf{Scenario S6 (more clusters).}   The same distributions as in S1 with $I=80$.

    \item \textbf{Scenario S7 (correctly specified parametric models).}  Parametetric models are correctly specified. The covariate distributions are the same as in S1.
\end{itemize}
\noindent
These scenarios separate the challenges of (i) unimodal versus multimodal covariate structures, (ii) low- versus high-dimensional covariates, and (iii) varying number of clusters.

We use  a $g$-prior \citep{Zellner1986} with the Empirical Bayes choice of $g$ \citep{Liang2008} for each prior specification of \eqref{eq:base_measure_hyperparameters}.
Let $\mathbb{C}^{(y)}$ denote the design matrix obtained by vertically stacking the unit-level design vectors $\mathbb{C}^{(y)}_{ij}$ (see Section \ref{sec:CA-EDP}) across clusters $i=1,\ldots,I$ and units $j=1,\ldots,N_i$. In particular,
$\mathbb{C}^{(y)} =\big(\mathbb{C}^{(y)}_{11},\ldots,\mathbb{C}^{(y)}_{IN_I}\big)^\top.
$
For $Y$, we regress $Y$ on $\mathbb{C}^{(y)}$ and set $\boldsymbol{\mu}_{\boldsymbol{\beta}^{(y)}} = \widehat{\boldsymbol{\beta}}_{\mathrm{OLS}}$, and $\boldsymbol{\Sigma}_{\boldsymbol{\beta}^{(y)}} = g_y\,\hat{\sigma}_{(y)}^{2}\,
\bigl(\mathbb{C}^{(y)\top}\mathbb{C}^{(y)}\bigr)^{-1}$.
We take $\hat{\sigma}_{(y)}^{2}=\mathrm{SSE}/(N-d_y)$ from that regression and estimate $g_y$ by the maximum likelihood under the $g$-prior, which yields $g_y = \max\left\{0, \frac{R_y^{2}}{1-R_y^{2}}\cdot\frac{N-d_y-1}{d_y}\right\},$ where $R_y^{2}$ is the OLS coefficient of determination (see \citet{George2000,Liang2008} for details). 
We specify $M$ and $D$ analogously, using their respective design matrices $\mathbb{C}^{(m)}$ and $\mathbb{C}^{(d)}$.

The initial parameter values were randomly drawn from the prior distributions, and the posterior samples were obtained by running a chain for $5000$ MCMC iterations after an initial $10000$ burn-in iteration. 
We simulate $100$ datasets and evaluate the bias and root mean square error (RMSE) of the point estimator of causal estimands, as well as  coverage probability and interval length of the credible interval estimator.
We compute the true estimands using a Monte Carlo approximation by averaging them from $500,000$ clusters.

\begin{table*}
  \centering
  \caption{Bias and root mean squared error (RMSE) of point estimates and average length (AL) and coverage probability (CP) of $95\%$ credible intervals of causal mediation estimands, NIE and SME. S1 is the baseline scenario. S2 assumes multimodal distributions for $X_{ij}$ and $N_i$. S3 and S4 consider moderate to high dimensional covariates, where most of the coordinates are irrelevant to the outcome and mediator generation.  S5 and S6 consider varying number of clusters. In S7, the parametric model is correctly specified.}
  \begin{adjustbox}{width=14.2cm}
  \begin{tabular}{c c rrrr rrrr rrrr}
    \toprule
    & & \multicolumn{4}{c}{CA-EDP} & \multicolumn{4}{c}{nDDP} & \multicolumn{4}{c}{Parametric} \\
    \cmidrule(lr){3-6} \cmidrule(lr){7-10} \cmidrule(lr){11-14}
    \textbf{Scenario} & \textbf{Estimand} & Bias & RMSE & AL & CP & Bias & RMSE & AL & CP & Bias & RMSE & AL & CP \\
    \midrule
    \multirow{2}{*}{S1}  & SME & -0.09 & 1.24 & 5.99 & 0.98 & -0.22 & 1.27 & 4.55 & 0.95 & 0.16 & 3.41 & 11.99 & 0.93  \\
    & NIE & -0.49 & 1.38 & 6.21 & 0.98 & -0.79 & 1.53 & 4.63 & 0.91 & -0.14 & 3.57 & 11.86 & 0.93  \\
   \midrule
       \multirow{2}{*}{S2}   & SME & 0.18 & 1.73 & 6.50 & 0.93 & 0.03 & 1.68 & 5.62 & 0.91 & 0.32 & 3.20 & 10.59 & 0.92  \\
 & NIE & -0.17 & 1.91 & 6.81 & 0.93 & -0.37 & 1.81 & 5.87 & 0.88 & 0.10 & 3.63 & 10.71 & 0.90  \\
   \midrule
   \multirow{2}{*}{S3}  & SME & -0.28 & 1.78 & 6.77 & 0.92 & -0.21 & 1.43 & 5.10 & 0.93 & -0.41 & 3.07 & 10.84 & 0.93  \\
 & NIE & -0.41 & 1.81 & 6.89 & 0.95 & -0.46 & 1.57 & 5.14 & 0.91 & -0.48 & 3.07 & 10.73 & 0.93  \\
    \midrule
           \multirow{2}{*}{S4}  & SME & -0.12 & 1.84 & 7.05 & 0.91 & -0.24 & 1.80 & 5.67 & 0.88 & -0.31 & 4.16 & 11.27 & 0.90  \\
 & NIE & -0.18 & 1.84 & 7.13 & 0.91 & -0.48 & 1.90 & 5.64 & 0.88 & -0.47 & 4.08 & 11.22 & 0.91  \\
 \midrule
          \multirow{2}{*}{S5}  & SME & -0.37 & 2.80 & 12.00 & 0.96 & -0.53 & 2.69 & 10.24 & 0.97 & -0.44 & 5.60 & 17.39 & 0.88  \\
& NIE & -0.39 & 2.89 & 12.19 & 0.97 & -0.73 & 2.87 & 10.43 & 0.94 & -0.75 & 5.52 & 17.46 & 0.91  \\
 \midrule
\multirow{2}{*}{S6}  & SME & -0.11 & 0.86 & 4.06 & 0.99 & -0.24 & 0.73 & 2.84 & 0.96 & 0.27 & 2.75 & 8.87 & 0.91  \\
& NIE & -0.42 & 0.95 & 4.33 & 0.97 & -0.72 & 1.01 & 2.94 & 0.80 & 0.01 & 2.81 & 8.75 & 0.93  \\
\midrule
\multirow{2}{*}{S7}  & SME & 0.00 & 0.39 & 1.96 & 0.97 & 0.01 & 0.39 & 1.32 & 0.92 & 0.03 & 0.35 & 1.24 & 0.95  \\
& NIE & 0.00 & 0.49 & 2.10 & 0.95 & 0.02 & 0.45 & 1.37 & 0.88 & 0.04 & 0.40 & 1.28 & 0.94  \\
    \bottomrule
  \end{tabular}
  \end{adjustbox}
  \label{tab:simulation1}
\end{table*}

Table \ref{tab:simulation1} reports evaluation metrics for both causal mediation estimands across the six scenarios. Overall, the two BNP priors behave as expected, showing minimal bias and substantially outperforming the fully parametric alternative in bias and RMSE. The parametric model fails to capture the complex structure in data-generating processes, leading to large RMSEs and wider interval lengths. The relative performance of two BNP models, however, varies with the data-generating mechanism.

In S1 and S2, the two BNP priors deliver nearly identical point-estimate accuracy, but their interval widths differ. The CA-EDP accounts for uncertainty in $(N,X)$ and thus produces consistently wider posterior intervals than the nDDP (which uses the empirical distribution for covariates), yielding higher coverage. This is more evident in S2, where the covariate and cluster-size distributions follow mixtures that are difficult to capture with the fixed empirical distributions used by the nDDP. The CA-EDP attains better (closer to $95\%$) coverage by propagating uncertainty in both covariate and cluster-size distributions.

When irrelevant predictors inflate dimensionality in S3 and S4, the nDDP attains smaller RMSE than CA-EDP due to the  cost of learning the full joint law outweighs any shrinkage benefit from modeling the covariate and cluster size distributions. Specifically, in S3, nDDP achieves the lowest RMSE and narrower intervals, but the  CA-EDP still maintains near-nominal coverage. Scenarios S5 and S6 examine sensitivity to sample size. With more clusters in S6, CA-EDP achieves low RMSE and bias while maintaining nominal coverage, whereas the nDDP approach for NIE exhibits undercoverage. Finally, in Scenario S7, the correct parametric model is expected to perform well. Even in this simplest scenario, the two BNP models perform comparably to the parametric model, with no notable loss of efficiency (but nDDP still shows undercoverage for NIE).

Taken together, these results suggest the two competing BNP methods are complementary. CA-EDP generally delivers well-calibrated coverage across scenarios through flexible modeling of covariate and cluster-size distributions, especially when the empirical distribution struggles to capture complex structure (e.g., the mixture setting in S2). By contrast, nDDP is a strong choice when covariate cardinality is low or when $X$ is high-dimensional, cases in which a fixed empirical distribution keeps intervals tight and additional covariate modeling is unlikely to help.
The two BNP models safeguard performance by flexibly encompassing the parametric case, matching its results without efficiency loss.

\section{Application of the Bayesian Nonparametric Causal Mediation Method to the RPS Cluster-Randomized Trial}
\label{sec:empirical_analysis}

We analyze the RPS CRT introduced in Section \ref{sec:motivating_example} by applying our BNP mediation framework, fitting both the CA-EDP and nDDP models.
Our analysis focuses on household dietary diversity as the mediator of substantive interest.  By adjusting for a post‐treatment health‐service confounder (along with baseline covariates), we estimate (i) the direct effect of RPS transfers not explained by dietary changes (NDE) and (ii) the indirect effect transmitted through enhanced dietary diversity (NIE), which is further decomposed into (ii.a) the individual portion of each household (IME) and (ii.b) the spillover portion from other households (SME). Our analysis also includes a principled sensitivity analysis to examine the unidentified dependence structure among counterfactual post-treatment confounders, which has not been previously addressed \citep{cheng2024, ohnishi2025}.

Table \ref{tab:posterior_estimates} presents the posterior estimates of the causal mediation estimands with the sensitivity parameter $\rho$ generated from the prior in \eqref{eq:prior_rho}.
The results are generally consistent in direction with prior analyses, but the level of certainty varies across the two BNP priors.
For both BNP priors, the NIE is positive; under the CA-EDP the posterior mean is $0.165$ with a $95\%$ credible interval (CI) $(-0.044,0.380)$ and $96.2\%$ posterior probability (PP) that the estimand is greater than zero, whereas the nDDP yields a slightly smaller mean, $0.162$ with $95\%$ CI $(0.020,0.358)$ and $\mathrm{PP}=98.0\%$.
These estimates suggest that household dietary diversity is a plausible mediator of the total intervention effect.
For the decomposition of the NIE into the SME and IME both have CIs that cross zero under both BNP priors, indicating greater uncertainty about their standalone impact for CA-EDP (SME/IME: $/89.2\%/84.2\%$) and for nDDP (SME/IME: $/91.4\%/87.4\%$).
Overall, the two BNP priors lead to similar qualitative conclusions.  The mediation through household dietary diversity appears positive but modest. The wider intervals for CA-EDP relative to nDDP reflect propagation of uncertainty through the covariate and cluster size distributions, consistent with the simulations.
For comparison, a parametric Bayesian model (same specification as in Section \ref{sec:simulation}) yielded estimates that also support a positive mediated pathway and total effect, while its NDE and SME/IME remain positive. 

Finally, we evaluate the predictive performance of the conditional models using the log pseudo marginal likelihood (LPML; \citealp{Geisser1979}). 
For the conditional model of CA-EDP, the LPML is computed from the likelihood of the observed data $(D,M,Y)$ given the cluster- and individual-level covariates $(N,V,X)$ and model parameters in each MCMC iteration.
The LPML is a Bayesian model-fit criterion derived from leave-one-out (LOO) predictive assessments of the data. 
For the Bayesian models considered in the simulations (Parametric, nDDP, and CA-EDP), the LPML values are 18.83, 29.60, and 33.34, respectively. Because a higher LPML indicates better LOO predictive fit, CA-EDP is the best-performing model among the three. We focus on CA-EDP and nDDP in the subsequent comparisons.
\begin{table}[ht]
    \centering
    \caption{Posterior estimates of causal estimands. ``Est'', ``$95\%$ CI'', and ``PP'' represent the posterior mean, $95\%$ credible interval, and the posterior probability that the estimand is greater than zero, respectively.}
    \label{tab:posterior_estimates}
    \begin{adjustbox}{width=14.2cm}
    \begin{tabular}{lcccccccccc}
        \toprule
        & \multicolumn{3}{c}{CA-EDP} & \multicolumn{3}{c}{nDDP} & \multicolumn{3}{c}{Parametric} \\
        \cmidrule(lr){2-4}\cmidrule(lr){5-7}\cmidrule(lr){8-10}
        Estimand & Est & $95\%$ CI & PP (\%) & Est & $95\%$ CI & PP (\%) & Est & $95\%$ CI & PP (\%) \\
        \midrule
        $\mathrm{TE}$   & $0.330$ & $(0.052, 0.590)$ & $98.6$ & $0.319$ & $(0.065, 0.554)$ & $99.0$ & $0.259$ & $(0.031, 0.481)$ & $99.2$ \\
        NIE             & $0.165$ & $(-0.044, 0.380)$ & $96.2$ & $0.162$ & $(0.020, 0.358)$ & $98.0$ & $0.143$ & $(0.031, 0.296)$ & $99.4$ \\
        NDE             & $0.165$ & $(-0.127, 0.418)$ & $89.4$ & $0.157$ & $(-0.136, 0.420)$ & $88.4$ & $0.115$ & $(-0.124, 0.366)$ & $82.0$ \\
        $\mathrm{SME}$  & $0.111$ & $(-0.078, 0.313)$ & $89.2$ & $0.110$ & $(-0.046, 0.298)$ & $91.4$ & $0.091$ & $(-0.024, 0.233)$ & $93.2$ \\
        $\mathrm{IME}$  & $0.054$ & $(-0.062, 0.176)$ & $84.2$ & $0.052$ & $(-0.058, 0.161)$ & $87.4$ & $0.053$ & $(-0.018, 0.133)$ & $93.0$ \\
        \bottomrule
    \end{tabular}
    \end{adjustbox}
\end{table}

\subsection{Sensitivity analysis on cross-world dependence}
\label{sec:sensitivity_analysis}
Previous literature \citep{ohnishi2025} assumed conditional independence between cross-world, inter-individual mediators (post-treatment confounder in our context), which cannot be verified with the observed data. Our sensitivity analysis framework utilizes copula modeling to capture dependence between cross-world post-treatment confounders within clusters, as described in Section \ref{sec:sensitivity_analysis_introduction}. Figure \ref{fig:analysis_result_rho} presents the sensitivity analysis results for both CA-EDP and nDDP. We compare posterior estimates under a fixed sensitivity parameter $\rho=0$ (no cross-world dependence) and under a prior on $\rho$ as given in \eqref{eq:prior_rho}. In both cases, the causal effect estimates are stable. Figure \ref{fig:analysis_result_various_rho} presents additional sensitivity analyses examining how posterior estimates of the causal estimands vary over a range of plausible fixed values, $\rho \in \{0.1, 0.3, 0.5, 0.7, 0.9\}$, rather than placing a prior on $\rho$.  We find that the causal conclusions remain stable across these fixed values of $\rho$, suggesting that the conclusions are robust to this sensitivity parameter.

In the RPS application, the limited sensitivity to the cross-world dependence assumption appears to be data-specific rather than a generic feature of CRT mediation problems. In particular, the post-treatment confounder, routine health check-up status, is binary and highly prevalent in both study arms (Table \ref{tab:baseline_cov}), so the arm-specific marginal distributions of $\mathbf{D}(0)$ and $\mathbf{D}(1)$ are already tightly informed by the observed data, leaving relatively limited room for alternative cross-world couplings to materially alter the g-computation. Moreover, once baseline covariates and the observed post-treatment confounder are accounted for, the remaining fitted mediator and outcome regressions do not suggest that modest changes in the latent cross-world association of $\mathbf{D}(0)$ and $\mathbf{D}(1)$ would induce large shifts in the mediation estimands. Thus, in this dataset, the sensitivity parameter $\rho$ mainly perturbs an unidentifiable dependence structure among counterfactual check-up statuses without substantially changing the observed-data features that drive posterior inference.
Overall, the qualitative conclusions are robust to plausible departures from the identifying assumptions encoded by $\rho$.

\begin{figure*}
    \centering
    \includegraphics[width=14.2cm]{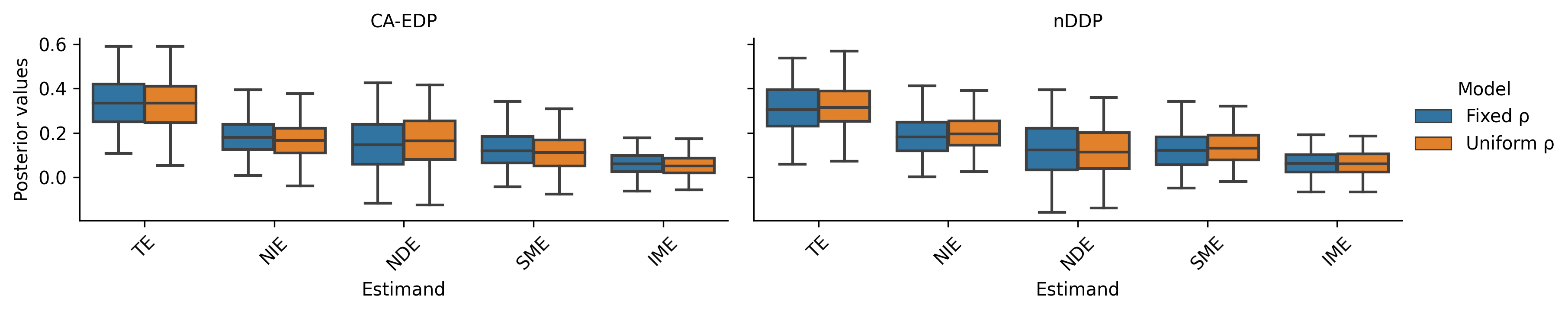}
    \caption{Boxplots of posterior samples for all estimands under the two BNP priors (CA-EDP, nDDP), with a fixed $\rho=0$ (no cross-world dependence) and a prior on $\rho$ in \eqref{eq:prior_rho}.}  
    \label{fig:analysis_result_rho}
\end{figure*}

\begin{figure*}
    \centering
    \includegraphics[width=14.2cm]{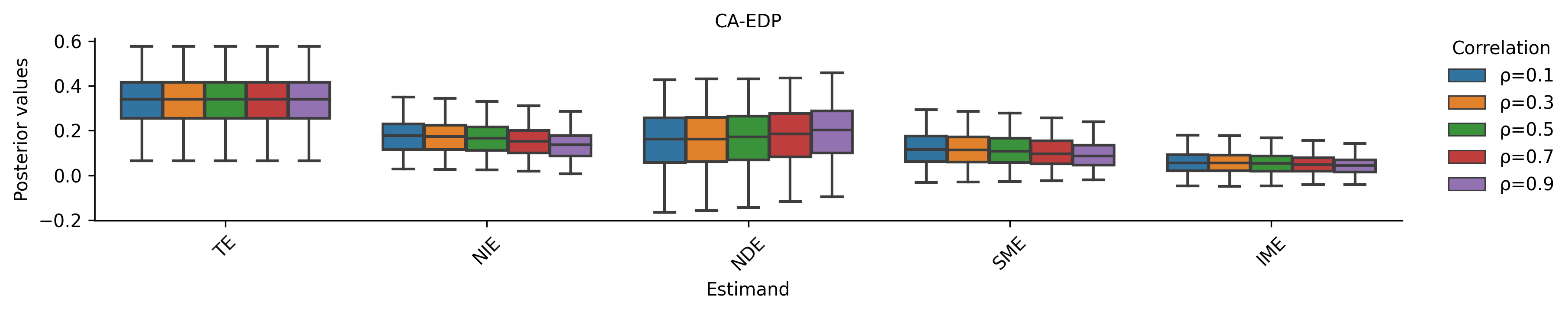}
    \caption{Boxplots of posterior samples for all estimands under the CA-EDP, with various values of $\rho \in \{0.1,0.3,0.5,0.7,0.9\}$.}  
    \label{fig:analysis_result_various_rho}
\end{figure*}

\subsection{Effects of post-treatment confounders}
\label{sec:post_confouders_effect}

\citet{cheng2024} assumed no post-treatment confounding in the RPS CRT.  However, household health check-up status is plausibly affected by treatment and related to both the mediator and the outcome, violating that assumption. Figure \ref{fig:analysis_result_PC} compares posterior distributions of all estimands under CA-EDP and nDDP, with and without adjustment for this post-treatment confounder.

The TE is stable across the two BNP priors, but the adjustment reduces the mediated pathway. In nDDP, NIE decreases and SME/IME move toward zero under both BNP priors. This is consistent with positive mediator-outcome confounding by health check-ups; without adjustment, some of the effect attributable to the post-treatment confounder is misattributed to the mediator.
Consequently, the direct pathway strengthens with adjustment. NDE increases correspondingly, aligning with the decomposition TE = NDE + NIE and implying that, once the confounder is controlled, a larger share of the effect operates through direct pathways not captured by the mediator.

Overall, adjusting for the post-treatment confounder reduces NIE (and SME/IME) and increases NDE while leaving TE essentially unchanged. Mediation remains positive but more modest, highlighting the importance of explicitly modeling post-treatment confounding in CRTs when interpreting pathway-specific effects.

\begin{figure*}
    \centering
    \includegraphics[width=14.2cm]{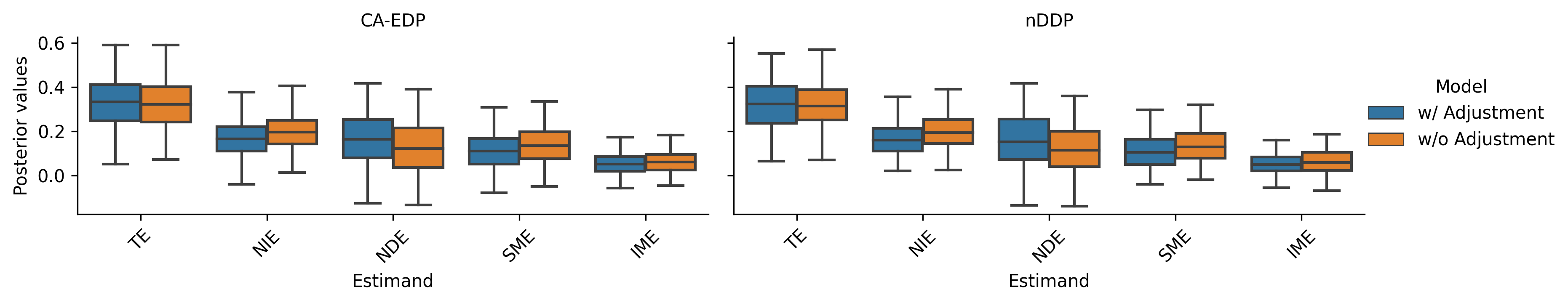}
    \caption{Boxplots of posterior samples for all estimands under two BNP priors (CA-EDP, nDDP), with and without adjustment for the post-treatment confounder.}
    \label{fig:analysis_result_PC}
\end{figure*}

\section{Concluding Remarks}
\label{sec:conlusion}
We proposed a new BNP prior, CA-EDP, for flexible modeling of joint distributions in multilevel (e.g., cluster/individual) data structures.  Although our empirical focus was causal mediation analysis in CRTs with a post-treatment confounder, the proposed prior is broadly applicable to other settings involving hierarchically structured data.  The CA-EDP prior complements the nDDP of \citet{ohnishi2025}.  Simulation results reveal preferred scenarios for each prior: the nDDP performs well when covariates are high‑dimensional and the study comprises fewer clusters with smaller cluster sizes, whereas the CA-EDP outperforms the nDDP in uncertainty quantification and is particularly advantageous when the sample is large enough to properly inform the underlying covariate distribution. 
Overall, the simulations provide practical guidance on model choice for specific CRT scenarios. 
Another key feature of CA-EDP is its joint modeling of covariates at both the cluster and individual levels, which yields principled uncertainty quantification by propagating uncertainty in the covariate distribution. Although bootstrap-based methods can be used to propagate uncertainty in the covariate distribution, they are typically implemented as plug-in or post hoc corrections rather than through a fully joint Bayesian model that yields coherent posterior uncertainty quantification. This joint modeling framework also naturally accommodates ignorable covariate missingness. Additional discussion comparing the CA-EDP and nDDP is provided in Supplementary Material C.

Another contribution of this work is the development of a three-level truncation approximation for the proposed CA-EDP prior, which both facilitates posterior computation through a blocked Gibbs sampler and broadens the practical usability of the proposed model. An interesting direction for future work is to replace the Dirichlet process component with other tractable stick-breaking priors, such as the Pitman-Yor process \citep{Ishwaran2001}, 
which may be incorporated with little additional theoretical or computational burden while providing greater flexibility in the induced grouping behavior.

We also propose a copula‑based sensitivity analysis framework to examine the dependence structure of counterfactual post‑treatment confounders in CRTs. This copula formulation enables a principled sensitivity analysis: by indexing unidentifiable cross-world dependence with a sensitivity parameter $\rho$, investigators can assess how posterior inferences vary over scientifically plausible ranges and thereby gauge the robustness of their conclusions. 
In principle, plausible values of $\rho$ could also be benchmarked by calibrating the sensitivity parameter to observed-data patterns, following the general strategy for sensitivity-parameter calibration discussed by \citet{Daniels_Linero_Roy_2023}. Adapting such a calibration strategy to the present multilevel mediation setting is left for future research.

\section*{Acknowledgement}
Research in this article was supported by the Patient-Centered Outcomes Research Institute\textsuperscript{\textregistered} (PCORI\textsuperscript{\textregistered} Award ME-2023C1-31350) and National Institutes of Health (R01 HL166324).

\bibliographystyle{Chicago}
\bibliography{literature}

\clearpage

\appendix
\setcounter{page}{1}

\section{Technical proofs}
\subsection{Proof of Theorem \ref{thm:identification}}
\label{sec:proof_identification}
\begin{proof}
    Let $\theta_{\mathrm{C}}(a,a') = \E\qty[\frac{1}{N_i} \sum_{j=1}^{N_i}  Y_{ij}(a, \mathbf{M}_i(a')) ]$.
    For simplicity, we suppress the cluster subscript $i$ throughout the proof.
    \begin{equation}
    \label{eq:identification_proof_1}
    \begin{split}
        \E\qty[\frac{1}{N} \sum_{j=1}^{(n)}  Y_{\cdot j}(a, \mathbf{M}(a')) ] &= \E_{\mathbf{C},N} \qty [ \E\qty[\frac{1}{N} \sum_{j=1}^{(n)}  Y_{\cdot j}(a, \mathbf{M}(a')) \mid \mathbf{C},N ] ] \\
        &= \E_{\mathbf{C},N} \qty [\frac{1}{N} \sum_{j=1}^{(n)}  \int_{\mathbb{R}^{(n)}}  \E\qty[ Y_{\cdot j}(a, \mathbf{M}(a')) \mid \mathbf{D}(a)=\mathbf{d}, \mathbf{C},N ] dF_{\mathbf{D}(a) \mid \mathbf{C},N}(\mathbf{d})] \\
        &= \E_{\mathbf{C},N} \qty [\frac{1}{N} \sum_{j=1}^{(n)}  \int_{\mathbb{R}^{(n)}}  \E\qty[ Y_{\cdot j}(a, \mathbf{M}(a')) \mid \mathbf{D}(a)=\mathbf{d}, \mathbf{C},N ] dF_{\mathbf{D} \mid A=a, \mathbf{C},N}(\mathbf{d})],
    \end{split}
    \end{equation}
    where the first and second equality are by the law of iterated expectations (LIE), and the third equality is by Assumptions \ref{asmp:sutva} and \ref{asmp:randomization}. 
    Now consider the identification of the inner expectation of \eqref{eq:identification_proof_1}. Note that
    \begin{align}
    \label{eq:identification_proof_2}
        &\E\qty[ Y_{\cdot j}(a, \mathbf{M}(a')) \mid \mathbf{D}(a)=\mathbf{d}, \mathbf{C},N ] \nonumber\\
        &= \int_{\mathbb{R}^{(n)}}  \E\qty[ Y_{\cdot j}(a, \mathbf{m}) \mid \mathbf{M}(a')=\mathbf{m}, \mathbf{D}(a)=\mathbf{d}, \mathbf{C},N ] dF_{\mathbf{M}(a') \mid \mathbf{D}(a) = \mathbf{d}, \mathbf{C},N}(\mathbf{m}) \nonumber\\
        &= \int_{\mathbb{R}^{(n)}}  \E\qty[ Y_{\cdot j}(a, \mathbf{m}) \mid \mathbf{M}(a')=\mathbf{m}, \mathbf{D}(a)=\mathbf{d}, A=a, \mathbf{C},N ] dF_{\mathbf{M}(a') \mid \mathbf{D}(a) = \mathbf{d}, A=a, \mathbf{C},N}(\mathbf{m}) \nonumber\\
        &= \int_{\mathbb{R}^{(n)}}  \E\qty[ Y_{\cdot j}(a, \mathbf{m}) \mid \mathbf{M}(a)=\mathbf{m}, \mathbf{D}(a)=\mathbf{d}, A=a, \mathbf{C},N ] dF_{\mathbf{M}(a') \mid \mathbf{D}(a) = \mathbf{d}, A=a, \mathbf{C},N}(\mathbf{m}) \nonumber\\
        &= \int_{\mathbb{R}^{(n)}}  \E\qty[ Y_{\cdot j} \mid \mathbf{M}=\mathbf{m}, \mathbf{D}=\mathbf{d}, A=a, \mathbf{C},N ] dF_{\mathbf{M}(a') \mid \mathbf{D}(a) = \mathbf{d}, A=a, \mathbf{C},N}(\mathbf{m}),
    \end{align}
    where the first equality follows from the LIE, the second equality follows from Assumption \ref{asmp:randomization}, the third equality follows from Assumption \ref{asmp:si}, and the fourth equality follows from Assumption \ref{asmp:sutva}. Now, the remaining task is to identify the conditional distribution of $\mathbf{M}(a')$. Note that we have
    \begin{align}
    \label{eq:identification_proof_3}
        & \pr\qty(\mathbf{M}(a') = \mathbf{m} \mid \mathbf{D}(a) = \mathbf{d}, A=a, \mathbf{C},N) \nonumber\\
        &= \int_{\mathbb{R}^{(n)}} \pr\qty(\mathbf{M}(a') = \mathbf{m}, \mathbf{D}(a') = \mathbf{z}' \mid \mathbf{D}(a) = \mathbf{d}, A=a, \mathbf{C},N) d\mathbf{z}' \nonumber\\
        &= \int_{\mathbb{R}^{(n)}} \pr\qty(\mathbf{M}(a') = \mathbf{m} \mid \mathbf{D}(a') = \mathbf{z}',  \mathbf{D}(a) = \mathbf{d}, A=a, \mathbf{C},N) \pr\qty( \mathbf{D}(a') = \mathbf{z}' \mid  \mathbf{D}(a) = \mathbf{d}, A=a, \mathbf{C},N) d\mathbf{z}' \nonumber\\
        &= \int_{\mathbb{R}^{(n)}} \pr\qty(\mathbf{M}(a') = \mathbf{m} \mid \mathbf{D}(a') = \mathbf{z}',  \mathbf{D}(a) = \mathbf{d}, A=a', \mathbf{C},N) \pr\qty( \mathbf{D}(a') = \mathbf{z}' \mid  \mathbf{D}(a) = \mathbf{d}, A=a, \mathbf{C},N) d\mathbf{z}' \nonumber\\
        &= \int_{\mathbb{R}^{(n)}} \pr\qty(\mathbf{M}(a') = \mathbf{m} \mid \mathbf{D}(a') = \mathbf{z}',  A=a', \mathbf{C},N) \pr\qty( \mathbf{D}(a') = \mathbf{z}' \mid  \mathbf{D}(a) = \mathbf{d}, A=a, \mathbf{C},N) d\mathbf{z}' \nonumber\\
        &= \int_{\mathbb{R}^{(n)}} \pr\qty(\mathbf{M} = \mathbf{m} \mid \mathbf{D} = \mathbf{z}',  A=a', \mathbf{C},N) \pr\qty( \mathbf{D}(a') = \mathbf{z}' \mid  \mathbf{D}(a) = \mathbf{d}, A=a, \mathbf{C},N) d\mathbf{z}',
    \end{align}
    where the third equality follows from Assumption \ref{asmp:randomization}, the fourth equality follows from Assumption \ref{asmp:si}, and the fifth equality follows from Assumption \ref{asmp:sutva}. Note that the conditional distribution, $\mathbf{D}(a') \mid  \mathbf{D}(a) = \mathbf{d}, A=a, \mathbf{C},N$, is partially identified by the Equation \eqref{eq:gaussian_copula}.
    Now, by Equations \eqref{eq:identification_proof_1} -- \eqref{eq:identification_proof_3}, we have:
    \begin{align*}
        &\E\left[\frac{1}{N} \sum_{j=1}^{(n)}  Y_{\cdot j}(a, \mathbf{M}(a')) \right] \\
        &= \E_{\mathbf{C},N} \left[\frac{1}{N} \sum_{j=1}^{(n)}  \int_{\mathbb{R}^{(n)}}  \int_{\mathbb{R}^{(n)}}  \E\left[ Y_{\cdot j} \mid \mathbf{M}=\mathbf{m}, \mathbf{D}=\mathbf{d}, A=a, \mathbf{C},N \right] \right.\\
        &\left. \quad\quad\quad\quad\quad\quad dF_{\mathbf{M}(a') \mid \mathbf{D}(a) = \mathbf{d}, A=a, \mathbf{C},N}(\mathbf{m}) dF_{\mathbf{D} \mid A=a, \mathbf{C},N}(\mathbf{d}) \right]\\
        &= \E_{\mathbf{C},N} \left[\frac{1}{N} \sum_{j=1}^{(n)}  \int_{\mathbb{R}^{(n)}}  \int_{\mathbb{R}^{(n)}}  \E\left[ Y_{\cdot j} \mid \mathbf{M}=\mathbf{m}, \mathbf{D}=\mathbf{d}, A=a, \mathbf{C},N \right] \right.\\
        &\left. \quad\quad\quad\quad\quad\quad \int_{\mathbb{R}^{(n)}} dF_{\mathbf{M} \mid \mathbf{D} = \mathbf{z}',  A=a', \mathbf{C},N}(\mathbf{m}) dF_{ \mathbf{D}(a') \mid  \mathbf{D}(a) = \mathbf{d}, A=a, \mathbf{C},N}(\mathbf{z}') dF_{\mathbf{D} \mid A=a, \mathbf{C},N}(\mathbf{d}) \right] \\
        &= \E_{\mathbf{C},N} \left[\frac{1}{N} \sum_{j=1}^{(n)}  \int_{\mathbb{R}^{(n)}}  \int_{\mathbb{R}^{(n)}} \int_{\mathbb{R}^{(n)}} \E\left[ Y_{\cdot j} \mid \mathbf{M}=\mathbf{m}, \mathbf{D}=\mathbf{d}, A=a, \mathbf{C},N \right] \right.\\
        &\left. \quad\quad\quad\quad\quad\quad  dF_{\mathbf{M} \mid \mathbf{D} = \mathbf{z}',  A=a', \mathbf{C},N}(\mathbf{m}) dF_{ \mathbf{D}(a') \mid  \mathbf{D}(a) = \mathbf{d}, A=a, \mathbf{C},N}(\mathbf{z}') dF_{\mathbf{D} \mid A=a, \mathbf{C},N}(\mathbf{d}) \right].
    \end{align*}

    Next, we consider the identification of $\E \left[ \frac{1}{N}\sum_{j=1}^{(n)}  Y_{\cdot j}(1, M_{\cdot j}(1), \mathbf{M}_{(-j)}(0)) \right] $.
    \begin{align*}
        &\E \left[ \frac{1}{N}\sum_{j=1}^{(n)}  Y_{\cdot j}(1, M_{\cdot j}(1), \mathbf{M}_{(-j)}(0)) \right] \\
        &= \E_{\mathbf{C},N} \qty [\frac{1}{N} \sum_{j=1}^{(n)}  \int_{\mathbb{R}^{(n)}}  \E\qty[ Y_{\cdot j}(1, M_{\cdot j}(1), \mathbf{M}_{(-j)}(0)) \mid \mathbf{D}(1)=\mathbf{d}, \mathbf{C},N ] dF_{\mathbf{D}(1) \mid \mathbf{C},N}(\mathbf{d})] \\
        &= \E_{\mathbf{C},N} \left [\frac{1}{N} \sum_{j=1}^{(n)}  \int_{\mathbb{R}^{(n)}}\int_{\mathbb{R}^{(n)}}  \E\left[ Y_{\cdot j}(1, m, \mathbf{m}_{N-1}  )  \mid  M_{\cdot j}(1)=m, \mathbf{M}_{(-j)}(0)=\mathbf{m}_{N-1}, \mathbf{D}(1)=\mathbf{d}, \mathbf{C},N \right] \right. \\
        &\left. \quad\quad\quad\quad\quad\quad dF_{M_{\cdot j}(1)\mathbf{M}_{(-j)}(0) \mid \mathbf{D}(1)=\mathbf{d}, \mathbf{C},N}(m, \mathbf{m}_{N-1})dF_{\mathbf{D}(1) \mid \mathbf{C},N}(\mathbf{d}) \right] \\
        &= \E_{\mathbf{C},N} \left [\frac{1}{N} \sum_{j=1}^{(n)}  \int_{\mathbb{R}^{(n)}}\int_{\mathbb{R}^{(n)}}  \E\left[ Y_{\cdot j}(1, m, \mathbf{m}_{N-1}  )  \mid  M_{\cdot j}(1)=m, \mathbf{M}_{(-j)}(0)=\mathbf{m}_{N-1}, \mathbf{D}(1)=\mathbf{d}, \mathbf{C},N \right] \right. \\
        &\left. \quad\quad\quad\quad\quad\quad dF_{M_{\cdot j}(1) \mid \mathbf{D}(1)=\mathbf{d}, \mathbf{C},N}(m)dF_{\mathbf{M}_{(-j)}(0) \mid \mathbf{D}(1)=\mathbf{d}, \mathbf{C},N}(\mathbf{m}_{N-1})dF_{\mathbf{D}(1) \mid \mathbf{C},N}(\mathbf{d}) \right] \\
        &= \E_{\mathbf{C},N} \left [\frac{1}{N} \sum_{j=1}^{(n)}  \int_{\mathbb{R}^{(n)}}\int_{\mathbb{R}^{(n)}}  \E\left[ Y_{\cdot j}(1, m, \mathbf{m}_{N-1}  )  \mid A=1, M_{\cdot j}(1)=m, \mathbf{M}_{(-j)}(0)=\mathbf{m}_{N-1}, \mathbf{D}(1)=\mathbf{d}, \mathbf{C},N \right] \right. \\
        &\left. \quad\quad\quad\quad\quad\quad dF_{M_{\cdot j}(1) \mid A=1, \mathbf{D}(1)=\mathbf{d}, \mathbf{C},N}(m)dF_{\mathbf{M}_{(-j)}(0) \mid A=1, \mathbf{D}(1)=\mathbf{d}, \mathbf{C},N}(\mathbf{m}_{N-1})dF_{\mathbf{D}(1) \mid A=1, \mathbf{C},N}(\mathbf{d}) \right] \\
        &= \E_{\mathbf{C},N} \left [\frac{1}{N} \sum_{j=1}^{(n)}  \int_{\mathbb{R}^{(n)}}\int_{\mathbb{R}^{(n)}}  \E\left[ Y_{\cdot j}(1, m, \mathbf{m}_{N-1}  )  \mid A=1, M_{\cdot j}(1)=m, \mathbf{M}_{(-j)}(1)=\mathbf{m}_{N-1}, \mathbf{D}(1)=\mathbf{d}, \mathbf{C},N \right] \right. \\
        &\left. \quad\quad\quad\quad\quad\quad dF_{M_{\cdot j}(1) \mid A=1, \mathbf{D}(1)=\mathbf{d}, \mathbf{C},N}(m)dF_{\mathbf{M}_{(-j)}(0) \mid A=1, \mathbf{D}(1)=\mathbf{d}, \mathbf{C},N}(\mathbf{m}_{N-1})dF_{\mathbf{D}(1) \mid A=1, \mathbf{C},N}(\mathbf{d}) \right] \\
        &= \E_{\mathbf{C},N} \left [\frac{1}{N} \sum_{j=1}^{(n)}  \int_{\mathbb{R}^{(n)}}\int_{\mathbb{R}^{(n)}}  \E\left[ Y_{\cdot j}  \mid A=1, M_{\cdot j}=m, \mathbf{M}_{(-j)}=\mathbf{m}_{N-1}, \mathbf{D}=\mathbf{d}, \mathbf{C},N \right] \right. \\
        &\left. \quad\quad\quad\quad\quad\quad dF_{M_{\cdot j} \mid A=1, \mathbf{D}=\mathbf{d}, \mathbf{C},N}(m)dF_{\mathbf{M}_{(-j)}(0) \mid A=1, \mathbf{D}(1)=\mathbf{d}, \mathbf{C},N}(\mathbf{m}_{N-1})dF_{\mathbf{D} \mid A=1, \mathbf{C},N}(\mathbf{d}) \right],
    \end{align*}
    where the third equality follows from Assumption \ref{asmp:cond_crossworld_indep_mediators}, the fourth equality follows from Assumption \ref{asmp:randomization}, the fifth equality follows from Assumption \ref{asmp:si}, and the sixth equality follows from Assumption \ref{asmp:sutva}. It remains necessary to identify the conditional distribution, $\mathbf{M}_{(-j)}(0) \mid A=1, \mathbf{D}(1)=\mathbf{d}, \mathbf{C},N$, which can be partially identified under Equation \eqref{eq:gaussian_copula} as shown in \eqref{eq:identification_proof_3}. Therefore, we have 
    \begin{align*}
        &\E \left[ \frac{1}{N}\sum_{j=1}^{(n)}  Y_{\cdot j}(1, M_{\cdot j}(1), \mathbf{M}_{(-j)}(0)) \right] \\
        &= \E_{\mathbf{C},N} \left [\frac{1}{N} \sum_{j=1}^{(n)}  \int_{\mathbb{R}^{(n)}}\int_{\mathbb{R}^{(n)}}  \E\left[ Y_{\cdot j}  \mid A=1, M_{\cdot j}=m, \mathbf{M}_{(-j)}=\mathbf{m}_{N-1}, \mathbf{D}=\mathbf{d}, \mathbf{C},N \right] \right. \\
        &\left. \quad\quad\quad dF_{M_{\cdot j} \mid A=1, \mathbf{D}=\mathbf{d}, \mathbf{C},N}(m)
        \int_{\mathbb{R}^{(n)}} dF_{\mathbf{M}_{(-j)} \mid \mathbf{D} = \mathbf{z}',  A=a', \mathbf{C},N}(\mathbf{m}_{N-1}) dF_{ \mathbf{D}(a') \mid  \mathbf{D}(a) = \mathbf{d}, A=a, \mathbf{C},N}(\mathbf{z}')
        dF_{\mathbf{D} \mid A=1, \mathbf{C},N}(\mathbf{d}) \right],\\
        &= \E_{\mathbf{C},N} \left [\frac{1}{N} \sum_{j=1}^{(n)}  \int_{\mathbb{R}^{(n)}}\int_{\mathbb{R}^{(n)}}\int_{\mathbb{R}^{(n)}}   \E\left[ Y_{\cdot j}  \mid A=1, M_{\cdot j}=m, \mathbf{M}_{(-j)}=\mathbf{m}_{N-1}, \mathbf{D}=\mathbf{d}, \mathbf{C},N \right] \right. \\
        &\left. \quad\quad\quad dF_{M_{\cdot j} \mid A=1, \mathbf{D}=\mathbf{d}, \mathbf{C},N}(m)
        dF_{\mathbf{M}_{(-j)} \mid \mathbf{D} = \mathbf{z}',  A=a', \mathbf{C},N}(\mathbf{m}_{N-1}) dF_{ \mathbf{D}(a') \mid  \mathbf{D}(a) = \mathbf{d}, A=a, \mathbf{C},N}(\mathbf{z}')
        dF_{\mathbf{D} \mid A=1, \mathbf{C},N}(\mathbf{d}) \right],
    \end{align*}
    which proves the desired result.
\end{proof}

\subsection{Proofs of distributional properties}
\label{sec:proofs_dist_prop}

In what follows, let GEM$(\cdot)$ denote the Griffiths-Engen-McCloskey (stick-breaking) prior distribution \citep{Sethuraman1994}. Additionally, we repeatedly use the following stick-breaking facts \citep[see e.g., ][]{Denti2023, ohnishi2025}: if $\{\pi_k\}$ follow a  GEM$(\alpha)$ prior, then for $k \neq k'$,
\[
\E\qty[ \sum_{k=1}^\infty \pi_k^2 ]=\frac{1}{1+\alpha},
\qquad
\E\qty[ \sum_{k=1}^\infty \pi_k \pi_{k'} ] = \sum_{k=1}^\infty \E[\pi_k]^2=\frac{1}{1+2\alpha}.
\]
The first identity follows from the recursion
$\mu=\E[\sum_k\pi_k^2]=\E[V_1^2]+\E[(1-V_1)^2]\mu$ with $V_1\sim\mathrm{Beta}(1,\alpha)$.

\subsubsection{Proof of $\pr(F_{i} = F_{i'})  = \pr(\boldsymbol{\eta}_{i} = \boldsymbol{\eta}_{i'})   = \frac{1}{1+\astar}$}
Let $Q=\sum_{k=1}^\infty \pi_k^{*}\delta_{(F_k^{*},\boldsymbol{\eta}_k^{*})}$ be the top-level GEM$(\astar)$ random measure on group-specific base distributions.
Assuming the marginal base laws for $F_k^{*}$ and $\boldsymbol{\eta}_k^{*}$ are atomic (there are no duplicate atoms $(F^{*}_{k}, \boldsymbol{\eta}_k^*)$ across different $k$) so that $\pr(F_{i}^{*}=F_{i'}^{*})=0$  and $\pr(\boldsymbol{\eta}_i^{*}=\boldsymbol{\eta}_{i'}^{*})=0$ for $i\neq i'$,
\[
\pr(F_i=F_{i'})= \E\qty[\pr(F_i=F_{i'} \mid Q) ]  = \E\qty[\sum_{k=1}^\infty (\pi_k^*)^2]=\frac{1}{1+\astar},
\]
which proves this equation.

\subsubsection{Proof of tie probabilities}
We omit superscripts on atom probabilities for notational simplicity (i.e., $w_{kl}=w^\theta_{kl}$). For $j\neq j'$,
\begin{align*}
    \pr(\boldsymbol{\theta}_{ij}=\boldsymbol{\theta}_{i'j'}) &= \E\{\pr(\boldsymbol{\theta}_{ij}=\boldsymbol{\theta}_{i'j'}\mid F_i,F_{i'})\} \\
    &= \E \qty[ \frac{1}{1+\astar}\pr(\boldsymbol{\theta}_{ij}=\boldsymbol{\theta}_{i'j'}\mid F_i=F_{i'})
    + \frac{\astar}{1+\astar}\pr(\boldsymbol{\theta}_{ij}=\boldsymbol{\theta}_{i'j'}\mid F_i \neq F_{i'})] \\
    &= \frac{1}{1+\astar}\E\qty(\sum_l w_{kl}^2)
    + \frac{\astar}{1+\astar}\E\qty(\sum_l w_{kl}w_{k'l}).
\end{align*}

For the cross term, common atom structure, independence across groups ($j\neq j'$) and identical distribution of the stick-breaking weights, $w_{kl} \sim \mathrm{GEM}(\atheta)$,  imply
\[
\mathbb{E}\Big(\sum_{k=1}^\infty w_{kl}w_{k'l}\Big)
= \sum_{l=1}^\infty \mathbb{E}(w_{kl})^2
= \sum_{l=1}^\infty \left[\frac{1}{1+\atheta}\left(\frac{\atheta}{1+\atheta}\right)^{l-1}\right]^2
= \frac{1}{1+2\atheta}. 
\]
Threfore,
\begin{equation}
    \pr(\boldsymbol{\theta}_{ij}=\boldsymbol{\theta}_{i'j'})
    = \frac{1}{1+\astar}\left(\frac{1}{1+\atheta}
    + \frac{\astar}{1+2\atheta}\right). 
\end{equation}

We omit superscripts on atom probabilities for simplicity (i.e., $w_{klm}=w^{\phi}_{klm}$). For $i\neq i'$,
\begin{align*}
    \pr(\boldsymbol{\phi}_{ij}=\boldsymbol{\phi}_{i'j'})&=\E\big[\pr(\boldsymbol{\phi}_{ij}=\boldsymbol{\phi}_{i'j'}\mid F_i,F_{i'})\big] \\
    & =\frac{1}{1+\astar}\pr(\boldsymbol{\phi}_{ij}=\boldsymbol{\phi}_{i'j'}\mid F_i=F_{i'}) +\frac{\astar}{1+\astar}\pr(\boldsymbol{\phi}_{ij}=\boldsymbol{\phi}_{i'j'}\mid F_i\neq F_{i'}).
\end{align*}
When $F_i=F_{i'}$, we can decompose based on whether the $\theta$ atoms coincide:
\begin{align*}
    \pr(\boldsymbol{\phi}_{ij}=\boldsymbol{\phi}_{i'j'}\mid F_i=F_{i'})
    &=\pr(\boldsymbol{\phi}_{ij}=\boldsymbol{\phi}_{i'j'}\mid F_i=F_{i'},\boldsymbol{\theta}_{ij}=\boldsymbol{\theta}_{i'j'})\pr(\boldsymbol{\theta}_{ij}=\boldsymbol{\theta}_{i'j'}\mid F_i=F_{i'}) \\
    &+\pr(\boldsymbol{\phi}_{ij}=\boldsymbol{\phi}_{i'j'}\mid F_i=F_{i'},\boldsymbol{\theta}_{ij}\neq\boldsymbol{\theta}_{i'j'})\pr(\boldsymbol{\theta}_{ij}\neq\boldsymbol{\theta}_{i'j'}\mid F_i=F_{i'}).
\end{align*}
Using the stick-breaking expectations at the $\theta$- and $\phi$-levels, we have
\begin{align*}
    &\pr(\boldsymbol{\theta}_{ij}=\boldsymbol{\theta}_{i'j'}\mid F_i=F_{i'})=\E\Big[\sum_{l} w_{kl}^2\Big]=\frac{1}{1+\atheta},\\
    &\pr(\boldsymbol{\theta}_{ij}\neq\boldsymbol{\theta}_{i'j'}\mid F_i=F_{i'})=\frac{\atheta}{1+\atheta},  \\
    &\pr(\boldsymbol{\phi}_{ij}=\boldsymbol{\phi}_{i'j'}\mid F_i=F_{i'},\boldsymbol{\theta}_{ij}=\boldsymbol{\theta}_{i'j'})=\E\Big[\sum_{m} w_{klm}^2\Big]=\frac{1}{1+\aphi}, \\
    &\pr(\boldsymbol{\phi}_{ij}=\boldsymbol{\phi}_{i'j'}\mid F_i=F_{i'},\boldsymbol{\theta}_{ij}\neq\boldsymbol{\theta}_{i'j'})=\E\Big[\sum_{m} w_{klm}w_{kl'm}\Big]=\frac{1}{1+2\aphi}.
\end{align*}
When $F_i\neq F_{i'}$,
\begin{align*}
    \pr(\boldsymbol{\phi}_{ij}=\boldsymbol{\phi}_{i'j'}\mid F_i\neq F_{i'})=\E\Big[\sum_{m} w_{klm}w_{k'l'm}\Big]=\frac{1}{1+2\aphi}.
\end{align*}
Putting the pieces together,
\begin{align*}
    \pr(\boldsymbol{\phi}_{ij}=\boldsymbol{\phi}_{i'j'})
    &=\frac{1}{1+\astar}\Big(\frac{1}{1+\atheta}\cdot\frac{1}{1+\aphi}+\frac{\atheta}{1+\atheta}\cdot\frac{1}{1+2\aphi}\Big)+\frac{\astar}{1+\astar}\cdot\frac{1}{1+2\aphi} \\
    &=\frac{1}{(1+\astar)(1+\atheta)}\cdot\frac{1}{1+\aphi}
    +\Big(1-\frac{1}{(1+\astar)(1+\atheta)}\Big)\cdot\frac{1}{1+2\aphi},
\end{align*}
which reduces to $q_{\phi}$.

For a joint tie, we decompose on the top level:
\begin{align*}
    \pr(\boldsymbol{\theta}_{ij}=\boldsymbol{\theta}_{i'j'},\boldsymbol{\phi}_{ij}=\boldsymbol{\phi}_{i'j'}) 
    &= \pr(F_i=F_{i'})\pr(\boldsymbol{\theta}_{ij}=\boldsymbol{\theta}_{i'j'},\boldsymbol{\phi}_{ij}=\boldsymbol{\phi}_{i'j'} \mid F_i=F_{i'}) \\
    &+ \pr(F_i \neq F_{i'})\pr(\boldsymbol{\theta}_{ij}=\boldsymbol{\theta}_{i'j'},\boldsymbol{\phi}_{ij}=\boldsymbol{\phi}_{i'j'} \mid F_i \neq F_{i'}) \\
    &= \pr(F_i=F_{i'})\frac{1}{(1+\atheta)(1+\aphi)} +\pr(F_i\neq F_{i'})\frac{1}{(1+2\atheta)(1+2\aphi)}.
\end{align*}
Using $\pr(F_i=F_{i'})=1/(1+\astar)$ gives
\[
q_{\theta,\phi}
=\frac{1}{1+\astar}\Big(\frac{1}{(1+\atheta)(1+\aphi)}+\frac{\astar}{(1+2\atheta)(1+2\aphi)}\Big),
\]
which is $q_{\theta,\phi}$.

\subsubsection{Proof of the identity on $\cov\left( F_{i}\left( A_{\theta}, A_{\phi} \right),F_{i'} \left( B_{\theta}, B_{\phi} \right) \right)$}
Our goal is to prove
\begin{equation}
    \label{eq:corr_AB_new}
    \begin{split}
        \cov\left( F_{i}\left( A_{\theta}, A_{\phi} \right),F_{i'} \left( B_{\theta}, B_{\phi} \right) \right) 
        &= q_{\theta}  G_{\phi}(A_{\phi})G_{\phi}(B_{\phi}) \Delta_{\theta}(A_{\theta}, B_{\theta})  \\
        &+ q_{\phi}   G_{\theta}(A_{\theta})G_{\theta}(B_{\theta}) \Delta_{\phi}(A_{\phi}, B_{\phi}) \\
        &+ q_{\theta, \phi}  \Delta_{\theta}(A_{\theta}, B_{\theta}) \Delta_{\phi}(A_{\phi}, B_{\phi}).
    \end{split}
\end{equation}
Let $A_\theta,B_\theta\subset\Theta$, $A_\phi,B_\phi\subset\Phi$ be Borel sets and define
$\Delta_\theta(A_\theta,B_\theta)=G_\theta(A_\theta\cap B_\theta)-G_\theta(A_\theta)G_\theta(B_\theta)$, $\Delta_\phi(A_\phi,B_\phi)=G_\phi(A_\phi\cap B_\phi)-G_\phi(A_\phi)G_\phi(B_\phi)$.
Write $A=A_\theta\times A_\phi$ and $B=B_\theta\times B_\phi$.
Without loss of generality, assume $i=1$ and $i'=2$. For simplicity, we write $\pi^{*}_k$ as $\pi_k$. Then
\begin{align*}
    \E\qty[F_1(A)F_2(B)] 
    &= \E\qty[ \sum_i \pi_{i}^2 F^{*}_i(A)F^{*}_i(B) + \sum_{i \neq i'} \pi_{i}\pi_{i'} F^{*}_i(A)F^{*}_{i'}(B)]\\
    &= \frac{1}{1+\astar} \E\qty[F^{*}_1(A)F^{*}_1(B)]
    + \frac{\astar}{1+\astar} \E\qty[F^{*}_1(A)F^{*}_2(B)].
\end{align*}

\paragraph{Calculation of $\mathbb{E}[F^{*}_1(A)F^{*}_1(B)]$}
\begin{align}
    \E\qty[F^{*}_1(A)F^{*}_1(B)]  
    &=\E\qty[ \qty(\sum_l \sum_m w_{1l}w_{1lm}\delta_{(\theta_l,\phi_m)}(A))\qty(\sum_l \sum_m w_{1l}w_{1lm}\delta_{(\theta_l,\phi_m)}(B))] \nonumber\\
    &= \E\qty[\sum_l \sum_m w_{1l}^2w_{1lm}^2\delta_{(\theta_l,\phi_m)}(A\cap B)] \label{eq:term1}\\
    &+ \E\qty[\sum_l \sum_{m_1\neq m_2} w_{1l}^2w_{1lm_1}w_{1lm_2}\delta_{(\theta_l,\phi_{m_1})}(A)\delta_{(\theta_l,\phi_{m_2})}(B)] \label{eq:term2}\\
    &+ \E\qty[\sum_{l_1\neq l_2}\ \sum_m w_{1l_1}w_{1l_2}w_{1 l_1 m}w_{1 l_2 m} \delta_{(\theta_{l_1},\phi_m)}(A)\delta_{(\theta_{l_2},\phi_m)}(B) ] \label{eq:term3}\\
    &+ \E\qty[ \sum_{l_1\neq l_2}\ \sum_{m_1\neq m_2} w_{1l_1}w_{1l_2}w_{1 l_1 m_1}w_{1 l_2 m_2} \delta_{(\theta_{l_1},\phi_{m_1})} (A)\delta_{(\theta_{l_2},\phi_{m_2})}(B)]. \label{eq:term4}
\end{align}

\noindent
The first term \eqref{eq:term1} simplifies to:
\begin{align*}
    \E\qty[\sum_l \sum_m w_{1l}^2w_{1lm}^2\delta_{(\theta_l,\phi_m)}(A\cap B)] 
    &= \E\qty[\sum_l \sum_m w_{1l}^2w_{1lm}^2] G_\theta(A_\theta\cap B_\theta) G_\phi(A_\phi\cap B_\phi)\\
    &= \frac{1}{1+\atheta}\frac{1}{1+\aphi}G_\theta(A_\theta\cap B_\theta) G_\phi(A_\phi\cap B_\phi)
\end{align*}
\noindent
The second term \eqref{eq:term2} simplifies to:
\begin{align*}
    &\E\qty[\sum_l \sum_{m_1\neq m_2} w_{1l}^2w_{1lm_1}w_{1lm_2}\delta_{(\theta_l,\phi_{m_1})}(A)\delta_{(\theta_l,\phi_{m_2})}(B)] \\
    =& \E\qty[\sum_l \sum_{m_1\neq m_2} w_{1l}^2w_{1lm_1}w_{1lm_2}] G_\theta(A_\theta\cap B_\theta) G_\phi(A_\phi)G_\phi(B_\phi) \\
    =& \E\qty[\sum_l  w_{1l}^2] 
    \E\qty[\sum_{m_1\neq m_2} w_{1lm_1}w_{1lm_2}] 
    G_\theta(A_\theta\cap B_\theta) G_\phi(A_\phi)G_\phi(B_\phi) \\
    =& \E\qty[\sum_l  w_{1l}^2] 
    \E\qty[1 - \sum_{m} w_{1lm}^2] 
    G_\theta(A_\theta\cap B_\theta) G_\phi(A_\phi)G_\phi(B_\phi) \\
    =& \frac{1}{1+\atheta}\frac{\aphi}{1+\aphi} G_\theta(A_\theta\cap B_\theta) G_\phi(A_\phi)G_\phi(B_\phi).
\end{align*}
\noindent
The third term \eqref{eq:term3} simplifies to:
\begin{align*}
    &\E\qty[\sum_{l_1\neq l_2}\ \sum_m w_{1l_1}w_{1l_2}w_{1 l_1 m}w_{1 l_2 m} \delta_{(\theta_{l_1},\phi_m)}(A)\delta_{(\theta_{l_2},\phi_m)}(B) ] \\
    =& \E\qty[\sum_{l_1\neq l_2}\ \sum_m w_{1l_1}w_{1l_2}w_{1 l_1 m}w_{1 l_2 m}] G_\theta(A_\theta) G_\theta(B_\theta) G_\phi(A_\phi \cap B_\phi) \\
    =& \E\qty[\sum_{l_1\neq l_2} w_{1l_1} w_{1l_2} ] \E \qty[ \sum_m w_{1 l_1 m}w_{1 l_2 m}] G_\theta(A_\theta) G_\theta(B_\theta) G_\phi(A_\phi \cap B_\phi) \\
    =& \E\qty[1 - \sum_{l} w_{1l}^2 ] \E \qty[ \sum_m w_{1 l_1 m}w_{1 l_2 m}] G_\theta(A_\theta) G_\theta(B_\theta) G_\phi(A_\phi \cap B_\phi) \\
    =& \E\qty[1 - \sum_{l} w_{1l}^2 ]  \sum_m \qty{ \E\qty[w_{1 l_1 m}]}^2 G_\theta(A_\theta) G_\theta(B_\theta) G_\phi(A_\phi \cap B_\phi) \\
    =& \frac{\atheta}{1+\atheta}\cdot \frac{1}{1+2\aphi} G_\theta(A_\theta) G_\theta(B_\theta) G_\phi(A_\phi \cap B_\phi).
\end{align*}
\noindent
The fourth term \eqref{eq:term4} simplifies to:
\begin{align*}
    &\E\qty[ \sum_{l_1\neq l_2}\ \sum_{m_1\neq m_2} w_{1l_1}w_{1l_2}w_{1 l_1 m_1}w_{1 l_2 m_2} \delta_{(\theta_{l_1},\phi_{m_1})} (A)\delta_{(\theta_{l_2},\phi_{m_2})}(B)] \\
    =& \E\qty[ \sum_{l_1\neq l_2}\ \sum_{m_1\neq m_2} w_{1l_1}w_{1l_2}w_{1 l_1 m_1}w_{1 l_2 m_2}] G_\theta(A_\theta) G_\theta(B_\theta) G_\phi(A_\phi)G_\phi(B_\phi) \\
    =& \E\qty[1 - \sum_{l} w_{1l}^2 ] \E\qty[1 - \sum_{m} w_{1lm}^2] G_\theta(A_\theta) G_\theta(B_\theta) G_\phi(A_\phi)G_\phi(B_\phi) \\
    =& \frac{\atheta}{1+\atheta}\cdot \frac{2\aphi}{1+2\aphi}G_\theta(A_\theta) G_\theta(B_\theta) G_\phi(A_\phi)G_\phi(B_\phi).
\end{align*}

\paragraph{Calculation of $\mathbb{E}[F^{*}_1(A)F^{*}_2(B)]$}

\begin{align}
    \E\qty[F^{*}_1(A)F^{*}_2(B)]  
    &=\E\qty[ \qty(\sum_l \sum_mw_{1l}w_{1lm}\delta_{(\theta_l,\phi_m)}(A))\qty(\sum_l \sum_m w_{2l}w_{2lm}\delta_{(\theta_l,\phi_m)}(B))] \nonumber\\
    &= \E\qty[\sum_l \sum_m w_{1l}w_{2l}w_{1 l m}w_{2 l m}\delta_{(\theta_l,\phi_m)}(A\cap B)] \label{eq:term1_cross}\\
    &+ \E\qty[\sum_l \sum_{m_1\neq m_2} w_{1l}w_{2l}w_{1lm_1}w_{2lm_2}\delta_{(\theta_l,\phi_{m_1})}(A)\delta_{(\theta_l,\phi_{m_2})}(B)] \label{eq:term2_cross}\\
    &+ \E\qty[\sum_{l_1\neq l_2}\ \sum_m w_{1l_1}w_{2l_2}w_{1 l_1 m}w_{2 l_2 m} \delta_{(\theta_{l_1},\phi_m)}(A)\delta_{(\theta_{l_2},\phi_m)}(B) ] \label{eq:term3_cross}\\
    &+ \E\qty[ \sum_{l_1\neq l_2}\ \sum_{m_1\neq m_2} w_{1l_1}w_{2l_2}w_{1 l_1 m_1}w_{2 l_2 m_2} \delta_{(\theta_{l_1},\phi_{m_1})} (A)\delta_{(\theta_{l_2},\phi_{m_2})}(B)]. \label{eq:term4_cross}
\end{align}

\noindent
The first term \eqref{eq:term1_cross} simplifies to:
\begin{align*}
    &\E\qty[\sum_l \sum_m w_{1l}w_{2l}w_{1 l m}w_{2 l m}\delta_{(\theta_l,\phi_m)}(A\cap B)]\\
    =& \E\qty[\sum_l w_{1l}w_{2l}] \E\qty[ \sum_m w_{1 l m}w_{2 l m}] G_\theta(A_\theta\cap B_\theta) G_\phi(A_\phi\cap B_\phi)\\
    =& \frac{1}{1+2\atheta}\frac{1}{1+2\aphi}G_\theta(A_\theta\cap B_\theta) G_\phi(A_\phi\cap B_\phi)
\end{align*}
\noindent
The second term \eqref{eq:term2_cross} simplifies to:
\begin{align*}
    &\E\qty[\sum_l \sum_{m_1\neq m_2} w_{1l}w_{2l}w_{1lm_1}w_{2lm_2}\delta_{(\theta_l,\phi_{m_1})}(A)\delta_{(\theta_l,\phi_{m_2})}(B)] \\
    =& \E\qty[\sum_l \sum_{m_1\neq m_2} w_{1l}w_{2l}w_{1lm_1}w_{2lm_2}] G_\theta(A_\theta\cap B_\theta) G_\phi(A_\phi)G_\phi(B_\phi) \\
    =& \E\qty[\sum_l  w_{1l}w_{2l}] 
    \E\qty[\sum_{m_1\neq m_2} w_{1lm_1}w_{2lm_2}] 
    G_\theta(A_\theta\cap B_\theta) G_\phi(A_\phi)G_\phi(B_\phi) \\
    =& \E\qty[\sum_l  w_{1l}w_{2l}] 
    \E\qty[1 - \sum_{m} w_{1lm}w_{2lm}] 
    G_\theta(A_\theta\cap B_\theta) G_\phi(A_\phi)G_\phi(B_\phi) \\
    =& \frac{1}{1+2\atheta}\frac{2\aphi}{1+2\aphi} G_\theta(A_\theta\cap B_\theta) G_\phi(A_\phi)G_\phi(B_\phi).
\end{align*}
\noindent
The third term \eqref{eq:term3_cross} simplifies to:
\begin{align*}
    &\E\qty[\sum_{l_1\neq l_2}\ \sum_m w_{1l_1}w_{2 l_2}w_{1 l_1 m}w_{2 l_2 m} \delta_{(\theta_{l_1},\phi_m)}(A)\delta_{(\theta_{l_2},\phi_m)}(B) ] \\
    =& \E\qty[\sum_{l_1 \neq l_2}\ \sum_m w_{1l_1}w_{2l_2}w_{1 l_1 m}w_{2 l_2 m}] G_\theta(A_\theta) G_\theta(B_\theta) G_\phi(A_\phi \cap B_\phi) \\
    =& \E\qty[\sum_{l_1\neq l_2} w_{1l_1} w_{2l_2} ] \E \qty[ \sum_m w_{1 l_1 m}w_{2 l_2 m}] G_\theta(A_\theta) G_\theta(B_\theta) G_\phi(A_\phi \cap B_\phi) \\
    =& \E\qty[1 - \sum_{l}  w_{1l_1} w_{2l_2} ] \E \qty[ \sum_m w_{1 l_1 m}w_{2 l_2 m}] G_\theta(A_\theta) G_\theta(B_\theta) G_\phi(A_\phi \cap B_\phi) \\
    =& \E\qty[1 - \sum_{l}  w_{1l} w_{2l} ]  \E \qty[ \sum_m w_{1 l_1 m}w_{2 l_2 m}] G_\theta(A_\theta) G_\theta(B_\theta) G_\phi(A_\phi \cap B_\phi) \\
    =& \frac{2\atheta}{1+2\atheta}\cdot \frac{1}{1+2\aphi} G_\theta(A_\theta) G_\theta(B_\theta) G_\phi(A_\phi \cap B_\phi).
\end{align*}
\noindent
The fourth term \eqref{eq:term4_cross} simplifies to:
\begin{align*}
    &\E\qty[ \sum_{l_1\neq l_2}\ \sum_{m_1\neq m_2} w_{1l_1}w_{2l_2}w_{1 l_1 m_1}w_{2 l_2 m_2} \delta_{(\theta_{l_1},\phi_{m_1})} (A)\delta_{(\theta_{l_2},\phi_{m_2})}(B)] \\
    =& \E\qty[ \sum_{l_1\neq l_2}\ \sum_{m_1\neq m_2} w_{1l_1}w_{2l_2}w_{1 l_1 m_1}w_{2 l_2 m_2}] G_\theta(A_\theta) G_\theta(B_\theta) G_\phi(A_\phi)G_\phi(B_\phi) \\
    =& \E\qty[1 - \sum_{l} w_{1l}w_{2l} ] \E\qty[1 - \sum_{m} w_{1 l_1 m}w_{2 l_2 m}] G_\theta(A_\theta) G_\theta(B_\theta) G_\phi(A_\phi)G_\phi(B_\phi) \\
    =& \frac{2\atheta}{1+2\atheta}\cdot \frac{2\aphi}{1+2\aphi}G_\theta(A_\theta) G_\theta(B_\theta) G_\phi(A_\phi)G_\phi(B_\phi).
\end{align*}

Finally, since $\E\qty[F_1(A)]=\E\qty[\E\qty[F_1(A)\mid Q]]=\E\qty[F^*_1(A)]=H_\theta(A_\theta)H_\phi(A_\phi)$, we have
\begin{align*}
    \E\qty[F_1(A)]\E\qty[F_1(B)] = G_\theta(A_\theta) G_\theta(B_\theta) G_\phi(A_\phi)G_\phi(B_\phi).
\end{align*}
Putting all these together, Equation \eqref{eq:corr_AB_new} follows.

\subsubsection{Proof of additional covariance and correlation identities in Section 3.2}
Note that $\Var\qty[F_1(A)] = \E\qty[F_1(A)^2] + \qty( \E\qty[F_1(A)] )^2.$
From the proof of \eqref{eq:corr_AB_new}, $\E\qty[F_1(A)^2]$ is easy to derive by letting $B=A$.
\begin{align*}
    \Var\qty{F_1(A)} &= \frac{G_\theta(A_\theta)G_\phi(A_\phi)}{(1+\atheta)(1+\aphi)}+
\frac{\aphi G_\theta(A_\theta)G_\phi(A_\phi)^2}{(1+\atheta)(1+\aphi)} \\
&+ \frac{\atheta G_\theta(A_\theta)^2G_\phi(A_\phi)}{(1+\atheta)(1+2\aphi)} +
\qty{G_\theta(A_\theta)G_\phi(A_\phi)}^2\qty( \frac{2\atheta\aphi}{(1+\atheta)(1+2\aphi)} - 1 ).
\end{align*}
\noindent
Let $A_\phi=B_\phi=\Phi$, then $G_\phi(A_\phi)=G_\phi(A_\phi)=G_\phi(\Phi)=1$. 
\begin{align*}
    \Var\qty{F_1(A_{\theta}, \Phi)} = \frac{1}{1+\atheta} G_\theta(A_\theta) (1-G_\theta(A_\theta))
\end{align*}
Additionally, from Equation \eqref{eq:corr_AB_new}, we have
\begin{align*}
    \cov \left(F_{i} \left(A_{\theta}, \Phi \right),F_{i'} \left( A_{\theta}, \Phi  \right) \right) = q_{\theta} G_\theta(A_\theta) (1-G_\theta(A_\theta)).
\end{align*}
Putting these together, we have
\begin{align*}
    \varrho_\theta = \corr \left(F_{i} \left(A_{\theta}, \Phi \right),F_{i'} \left( A_{\theta}, \Phi  \right) \right) = 1-\frac{\astar}{1+\astar}\frac{\atheta}{1+2\atheta}.
\end{align*}
Similarly, taking $A_\theta=B_\theta=\Theta$, we have
\begin{equation*}
    \cov \left(F_{i} \left( \Theta, A_{\phi} \right),F_{i'} \left( \Theta, B_{\phi} \right) \right) = q_{\phi}  \Delta_{\phi}(A_{\phi}, B_{\phi}).
\end{equation*}and the correlation on the same set $A_{\phi}$ is given by:
\begin{equation*}
    \varrho_{\phi} = \corr \left(F_{i} ( \Theta, A_{\phi} ),F_{i'} ( \Theta, A_{\phi} ) \right) = 1 - \frac{\astar}{1+\astar}\frac{\aphi}{1+\atheta + \atheta \aphi + 2\aphi}.
\end{equation*}

\subsection{Proof of Proposition \ref{prop:peppf}}
The result directly follows from the proof of Theorem 1 and Proposition 2 in \citet{Denti2023}.

\subsection{Proof of \eqref{eq:cond_posdef}}
\label{sec:proof_condition}

Recall that the cross-world dependence structure is characterized by the multivariate Gaussian copula \eqref{eq:gaussian_copula}.
This structure allows $\boldsymbol{\Omega}$ to be written as a $2 \times 2$ block matrix, where each block is of size $N_i \times N_i$:
\begin{align*}
    \boldsymbol{\Omega} = 
    \begin{pmatrix} 
    \mathbf{C}_{11} & \mathbf{C}_{10} \\ 
    \mathbf{C}_{10}^\top & \mathbf{C}_{00} 
    \end{pmatrix},
\end{align*}
where 
\begin{align*}
    \mathbf{C}_{11} = 
    \begin{pmatrix} 
    1 & \gamma_1 & \cdots & \gamma_1 \\ 
    \gamma_1 & 1 & \cdots & \gamma_1 \\
    \vdots & \vdots & \ddots & \vdots \\
    \gamma_1 & \cdots & \cdots & 1 \\ 
    \end{pmatrix},
    \mathbf{C}_{10} = 
    \begin{pmatrix} 
    \rho & \rho^{*} & \cdots & \rho^{*} \\ 
    \rho^{*} & \rho & \cdots & \rho^{*} \\
    \vdots & \vdots & \ddots & \vdots \\
    \rho^{*} & \cdots & \cdots & \rho \\ 
    \end{pmatrix},
    \mathbf{C}_{00} = 
    \begin{pmatrix} 
    1 & \gamma_0 & \cdots & \gamma_0 \\ 
    \gamma_0 & 1 & \cdots & \gamma_0 \\
    \vdots & \vdots & \ddots & \vdots \\
    \gamma_0 & \cdots & \cdots & 1 \\ 
    \end{pmatrix} \\
\end{align*}

The blocks have a compound symmetric structure:
\begin{equation*}
\begin{split}
    \mathbf{C}_{11} &= (1-\gamma_1)\mathbf{I}_{N_i} + \gamma_1 \mathbf{J}_{N_i} \\
    \mathbf{C}_{00} &= (1-\gamma_0)\mathbf{I}_{N_i} + \gamma_0 \mathbf{J}_{N_i} \\
    \mathbf{C}_{10} &= (\rho-\rho^{*})\mathbf{I}_{N_i} + \rho^{*} \mathbf{J}_{N_i},
\end{split}
\end{equation*}
where $\mathbf{I}_{N_i}$ is the identity matrix and $\mathbf{J}_{N_i}$ is the matrix of all ones.

For $\boldsymbol{\Omega}$ to be positive definite, all its eigenvalues must be positive. The blocks $\mathbf{C}_{11}, \mathbf{C}_{00}, \mathbf{C}_{10}$ share a common set of eigenvectors. This property simplifies the analysis. The eigenvectors of any matrix of the form $a\mathbf{I} + b\mathbf{J}$ are:
\begin{enumerate}[nosep]
    \item The vector of all ones, $\mathbf{1} \in \mathbb{R}^{N_i}$.
    \item Any vector $\mathbf{v} \in \mathbb{R}^{N_i}$ such that $\mathbf{v}$ is orthogonal to $\mathbf{1}$ (i.e., $\mathbf{v}^\top\mathbf{1} = 0$). There are $N_i-1$ such linearly independent vectors.
\end{enumerate}
This allows us to find the eigenvalues of $\boldsymbol{\Omega}$ by solving two separate $2 \times 2$ eigenvalue problems.

\paragraph{Case 1: Eigenvectors based on \texorpdfstring{$\mathbf{v} \perp \mathbf{1}$}{v orthogonal to 1}}
Let the eigenvector of $\boldsymbol{\Omega}$ be of the form $(\mathbf{v}, c\mathbf{v})^\top$ where $\mathbf{v}^\top\mathbf{1}=0$.
For such a vector $\mathbf{v}$, we have $\mathbf{J}\mathbf{v} = \mathbf{0}$. The multiplicatoin of the matrix blocks on $\mathbf{v}$ is as follows:
\begin{align*}
    \mathbf{C}_{11}\mathbf{v} &= ((1-\gamma_1)\mathbf{I} + \gamma_1 \mathbf{J})\mathbf{v} = (1-\gamma_1)\mathbf{v} \\
    \mathbf{C}_{00}\mathbf{v} &= ((1-\gamma_0)\mathbf{I} + \gamma_0 \mathbf{J})\mathbf{v} = (1-\gamma_0)\mathbf{v} \\
    \mathbf{C}_{10}\mathbf{v} &= ((\rho-\rho^{*})\mathbf{I} + \rho^{*} \mathbf{J})\mathbf{v} = (\rho-\rho^{*})\mathbf{v}.
\end{align*}
The eigenvalue problem $\boldsymbol{\Omega}(\mathbf{v}, c\mathbf{v})^\top = \lambda(\mathbf{v}, c\mathbf{v})^\top$ reduces to a $2 \times 2$ problem for the eigenvalues $\lambda$:
\begin{align*}
    \mathbf{M}_1 \begin{pmatrix} 1 \\ c \end{pmatrix} = 
    \begin{pmatrix} 
    1-\gamma_1 & \rho-\rho^{*} \\ 
    \rho-\rho^{*} & 1-\gamma_0 
    \end{pmatrix}
    \begin{pmatrix} 1 \\ c \end{pmatrix} = \lambda 
    \begin{pmatrix} 1 \\ c \end{pmatrix}.    
\end{align*}
The eigenvalues of $\boldsymbol{\Omega}$ associated with this case are the eigenvalues of $\mathbf{M}_1$. For $\boldsymbol{\Omega}$ to be positive definite, $\mathbf{M}_1$ must be positive definite. This yields the conditions:
\begin{align*}
    &1-\gamma_1 > 0 \implies \gamma_1 < 1, \\
    &1-\gamma_0 > 0 \implies \gamma_0 < 1, \\
    &\det(\mathbf{M}_1) = (1-\gamma_1)(1-\gamma_0) - (\rho-\rho^{*})^2 > 0.
\end{align*}

\paragraph{Case 2: Eigenvector based on the vector \texorpdfstring{$\mathbf{1}$}{1}}
Let the eigenvector of $\boldsymbol{\Omega}$ be of the form $(\mathbf{1}, c\mathbf{1})^\top$.
For the vector $\mathbf{1}$, we have $\mathbf{J}\mathbf{1} = N_i\mathbf{1}$. The multiplication of the matrix blocks on $\mathbf{1}$ is as follows:
\begin{align*}
    \mathbf{C}_{11}\mathbf{1} &= ((1-\gamma_1)\mathbf{I} + \gamma_1 \mathbf{J})\mathbf{1} = (1-\gamma_1 + N_i\gamma_1)\mathbf{1} = (1+(N_i-1)\gamma_1)\mathbf{1} \\
    \mathbf{C}_{00}\mathbf{1} &= ((1-\gamma_0)\mathbf{I} + \gamma_0 \mathbf{J})\mathbf{1} = (1+(N_i-1)\gamma_0)\mathbf{1} \\
    \mathbf{C}_{10}\mathbf{1} &= ((\rho-\rho^{*})\mathbf{I} + \rho^{*} \mathbf{J})\mathbf{1} = (\rho-\rho^{*}+N_i\rho^{*})\mathbf{1} = (\rho+(N_i-1)\rho^{*})\mathbf{1}.
\end{align*}
The eigenvalue problem reduces to:
\begin{align*}
    \mathbf{M}_2 \begin{pmatrix} 1 \\ c \end{pmatrix} = 
    \begin{pmatrix} 
    1+(N_i-1)\gamma_1 & \rho+(N_i-1)\rho^{*} \\ 
    \rho+(N_i-1)\rho^{*} & 1+(N_i-1)\gamma_0 
    \end{pmatrix}
    \begin{pmatrix} 1 \\ c \end{pmatrix} = \lambda 
    \begin{pmatrix} 1 \\ c \end{pmatrix}
\end{align*}
The eigenvalues of $\boldsymbol{\Omega}$ associated with this case are the eigenvalues of $\mathbf{M}_2$. For $\boldsymbol{\Omega}$ to be positive definite, $\mathbf{M}_2$ must be positive definite. This yields the conditions:
\begin{align*}
    &1+(N_i-1)\gamma_1 > 0 \\
    &1+(N_i-1)\gamma_0 > 0 \\
    &\det(\mathbf{M}_2) = \big(1+(N_i-1)\gamma_1\big)\big(1+(N_i-1)\gamma_0\big) - \big(\rho+(N_i-1)\rho^{*}\big)^2 > 0.
\end{align*}

Combining the conditions from both cases, the necessary and sufficient conditions for $\boldsymbol{\Omega}$ to be positive definite are:
\begin{enumerate}[nosep]
    \item $- \frac{1}{N_i-1} < \gamma_a < 1$ for $a \in \{0,1\}$
    \item $(1-\gamma_1)(1-\gamma_0) > (\rho-\rho^{*})^2$
    \item $\big(1+(N_i-1)\gamma_1\big)\big(1+(N_i-1)\gamma_0\big) > \big(\rho+(N_i-1)\rho^{*}\big)^2$.
\end{enumerate}
Setting $\rho^{*} = \frac{\gamma_1 + \gamma_0}{2}\times \rho$, we obtain the desired condition.


\subsection{Proofs of truncation error bounds}
In this section, we prove Theorem \ref{thm:finite_approximation} and \ref{thm:finite_approximation_marginal}.
Fix a top-level index $k \le K$, a $\theta$-index $ l \le L$, and a $\phi$-index $m \le M$, with corresponding (untruncated) stick-breaking weights $\pi_k$, $w_{kl}$, and $w_{klm}$. We omit the superscripts of the weights throughout. Let $S_i$ be the grouping process at group-level of the CA-EDP model, i.e., $S_i \sim \sum_k \pi_k \delta_k$. Conditional on $S_i=k$ (the top-level draw used by group $i$), the total variation distance between $F_i$ and $F_i^{(K,L,M)}$ satisfies the following identity.
\begin{lemma}
\label{lemma:F_indentity}
     When $S_i = k \leq K$, we have
    \begin{align*}
        TV(\widetilde{F}_i, \widetilde{F}_i^{(K,L,M)}) = \frac{1}{2}\qty(\sum_l \sum_m |w_{kl}^{K,L}w_{klm}^{K,L,M} - w_{kl}w_{klm}| )  = 1 - \sum_{l \leq L} \sum_{m \leq M}w_{kl}w_{klm}
    \end{align*}
\end{lemma}
\begin{proof}
We first have the following:
\begin{align*}
TV\qty( \widetilde{F}_{i}, \widetilde{F}^{(K,L,M)}_{i})
&= \frac{1}{2}\left(\sum_{l}\sum_{m}
\big|w^{(K,L)}_{kl}w^{(K,L,M)}_{klm}-w_{kl}w_{klm}\big|\right)\\[2pt]
&=\frac{1}{2}\left(\left(\sum_{l\le L-1}+\sum_{l=L}+\sum_{l\ge L+1}\right)
\left(\sum_{m\le M-1}+\sum_{m=M}+\sum_{m\ge M+1}\right)
\big|w^{(K,L)}_{kl}w^{(K,L,M)}_{klm}-w_{kl}w_{klm}\big|\right).
\end{align*}
We split these terms to derive the result.

\[
\sum_{l\le L-1}\sum_{m\le M-1}
\big|w^{(K,L)}_{kl}w^{(K,L,M)}_{klm}-w_{kl}w_{klm}\big|=0.
\]

\begin{align*}
    \sum_{l\le L-1}\sum_{m=M}
    \big|w^{(K,L)}_{kl}w^{(K,L,M)}_{klm}-w_{kl}w_{klm}\big|
    &=\sum_{l\le L-1}\sum_{m=M}w_{kl}\big|w^{(K,L,M)}_{klm}-w_{klm}\big|\\
    &=\sum_{l\le L-1}w_{kl}\big|w^{(K,L,M)}_{klM}-w_{klM}\big|\\
    &=\sum_{l\le L-1}w_{kl}\big|1-w_{kl1}-\cdots-w_{kl,M-1}-w_{klM}\big|\\
    &=\sum_{l\le L-1}\sum_{m\ge M+1}w_{kl}w_{klm} \tag{A}.\label{eq:A}
\end{align*}

Using the fact that $w^{(K,L,M)}_{klm}=0$ for $m\ge M+1$,
\[
\sum_{l\le L-1}\sum_{m\ge M+1}
\big|w^{(K,L)}_{kl}w^{(K,L,M)}_{klm}-w_{kl}w_{klm}\big|
=\sum_{l\le L-1}\sum_{m\ge M+1}w_{kl}w_{klm}. \tag{B}\label{eq:B}
\]

\begin{align*}
\sum_{l=L}\sum_{m\le M-1}
\big|w^{(K,L)}_{kl}w^{(K,L,M)}_{klm}-w_{kl}w_{klm}\big|
&= \sum_{m\le M-1}\big|w^{(K,L)}_{kL}-w_{kL}\big|w_{kLm}\\
&=\sum_{m\le M-1} \big|1-w_{k1}-\cdots-w_{k,L-1}-w_{kL}\big| w_{kLm}\\
&=\sum_{l\ge L+1}\sum_{m\le M-1}w_{kl}w_{kLm} \tag{C}\label{eq:C}
\end{align*}

\begin{align*}
    &\sum_{l=L}\sum_{m=M} \big|w^{(K,L)}_{kl}w^{(K,L,M)}_{klm}-w_{kl}w_{klm}\big| \\
    =& \big| (1 - \sum_{l\leq L-1}  w_{kl}) (1 - \sum_{m\leq M-1}  w_{kLm}) - w_{kL}w_{kLM}  \big| \\
    =& \big| 1 - \sum_{l\leq L-1}  w_{kl} - \sum_{m\leq M-1}  w_{kLm} + \sum_{l\leq L-1}\sum_{m\leq M-1} w_{kl}w_{kLm} - w_{kL}w_{kLM}  \big| \\
    =& \big| 1 - \sum_{l\leq L}  w_{kl} - \sum_{m\leq M}  w_{kLm}  + w_{kL} + w_{kLM} + \sum_{l\leq L}\sum_{m\leq M} w_{kl}w_{kLm} - \sum_{l\leq L}w_{kl}w_{kLM} - \sum_{m\leq M} w_{kL}w_{kLm}  \big| \\
    =& \big| w_{kLM} \qty( 1 - \sum_{l\leq L}w_{kl} ) +  w_{kL} \qty( 1 -  \sum_{m\leq M}  w_{kLm}) +  \sum_{l\leq L} (1 - w_{kl})  \sum_{m\leq M} (1 - w_{kLm})\big| \\
    =& \underbrace{w_{kLM} \sum_{l \geq L+1} w_{kl}}_{(D)} + \underbrace{w_{kL} \sum_{m \geq M+1} w_{kLm}}_{(E)} + \underbrace{\sum_{l \geq L+1} \sum_{m \geq M+1} w_{kl} w_{kLm} }_{(F)}.
\end{align*}

\begin{align*}
    \sum_{l=L}\sum_{m\geq M+1}
\big|w^{(K,L)}_{kl}w^{(K,L,M)}_{klm}-w_{kl}w_{klm}\big|
=\sum_{m\geq M+1}w_{kL}w_{kLm}. \tag{G}\label{eq:G}  
\end{align*}

\begin{align*}
    \sum_{l \geq L+1}\sum_{m=1}^{\infty} \big|w^{(K,L)}_{kl}w^{(K,L,M)}_{klm}-w_{kl}w_{klm}\big|
    = \sum_{l \geq L+1}\sum_{m=1}^{\infty} w_{kl}w_{klm} 
    = \sum_{l \geq L+1} w_{kl}.  \tag{H}\label{eq:H}
\end{align*}

Note that 
\begin{align*}
    TV\qty( \widetilde{F}_{i}, \widetilde{F}^{(K,L,M)}_{i}) = \frac{1}{2} ((A) + \ldots + (H)).
\end{align*}
\noindent
Additionally, 
\begin{align*}
    (A) + (E) &= \sum_{l \leq L} \sum_{m \geq M+1} w_{kl}w_{klm} \tag{AE} \\
    (C) + (D) &= \sum_{l \geq L+1} \sum_{m \leq M} w_{kl}w_{kLm} \tag{CD} \\
    (F) + (G) &= \sum_{l \le L} \sum_{m \ge M+1} w_{kl}w_{kLm} \tag{FG} \\
    (AE) + (B) &= \sum_{l \leq L-1} \sum_{m \geq M+1} w_{kl}w_{klm} + \sum_{l \leq L} \sum_{m \geq M+1} w_{kl}w_{klm}  \\
    &= 2 \sum_{l \leq L-1} \sum_{m \geq M+1} w_{kl}w_{klm} + \sum_{m \geq M+1} w_{kL}w_{kLm} \tag{AEB} \\
    (CD) + (FG) &= \sum_{l \geq L+1} \sum_{m =1}^{\infty} w_{kl}w_{kLm} + w_{kL} \sum_{m \ge M+1} w_{kLm}  \\
    &=\sum_{l \geq L+1}  w_{kl} + w_{kL} \sum_{m \ge M+1} w_{kLm}  \tag{CDFG}
\end{align*}

Finally, by adding up the above results, we have
\begin{align*}
    TV\qty( \widetilde{F}_{i}, \widetilde{F}^{(K,L,M)}_{i}) &= \frac{1}{2} ((AEB) + (CDFG) + (H)) \\
    &= \sum_{l \geq L+1} w_{kl} + w_{kL} \sum_{m \ge M+1} w_{kLm} + \sum_{l \leq L-1} \sum_{m \geq M+1} w_{kl}w_{klm}\\
    &= \sum_{l \geq L+1} w_{kl} +  \sum_{l \leq L} \sum_{m \geq M+1} w_{kl}w_{klm}\\
    &= \sum_{l \geq L+1} \sum_{m=1}^{\infty} w_{kl}w_{klm} +  \sum_{l \leq L} \sum_{m=1}^{\infty} w_{kl}w_{klm} -   \sum_{l \leq L} \sum_{m \leq M} w_{kl}w_{klm}\\
    &= 1-\sum_{l\le L}\sum_{m\le M}w_{kl}w_{klm}.
\end{align*}
\end{proof}

\subsubsection{Proof of Theorem \ref{thm:finite_approximation}}

\begin{proof}
Conditional on the top-level index $S_i=k$, we have:
\begin{align*}
    \E\qty[TV(\widetilde{F}_i^{(K,L,M)},\widetilde{F}_i)]
    &=\E\qty[\sum_{k = 1}^{K}\pi_k\E\qty[TV(\widetilde{F}_i^{(K,L,M)},\widetilde{F}_i)\mid S_i=k,Q]]\\
    &+\E\qty[\sum_{k = K+1}^{\infty}\pi_k\E\qty[TV(\widetilde{F}_i^{(K,L,M)},\widetilde{F}_i)\mid S_i=k,Q]].
\end{align*}
Applying Lemma \ref{lemma:F_indentity} to the first term and using the fact that $\E\qty[TV(\widetilde{F}_i^{(K,L,M)},\widetilde{F}_i)\mid S_i=k,Q] \leq 1$ for the second term, we have 
\begin{align*}
    & \E\qty[TV(\widetilde{F}_i^{(K,L,M)},\widetilde{F}_i)] \\
    &\leq \E\qty[\sum_{k \leq K}\pi_k\E\qty[ 1-\sum_{l\le L}\sum_{m\le M}w_{kl}w_{klm}]]  + \E\qty[\sum_{k \geq K+1}\pi_k] \\
    &= \E\qty[\sum_{k \leq K}\pi_k]\qty(1-\E\qty[ \sum_{l\le L}\sum_{m\le M}w_{kl}w_{klm} ] )]  + \E\qty[\sum_{k \geq K+1}\pi_k]\\
    &= \E\qty[\sum_{k \leq K}\pi_k]\qty(1-\E\qty[ \sum_{l\le L} w_{kl} \E\qty[\sum_{m\le M}w_{klm}] ] )]  + \E\qty[\sum_{k \geq K+1}\pi_k]\\
    &= \E\qty[\sum_{k \leq K}\pi_k]\qty(1-\E\qty[ \sum_{l\le L} w_{kl} ] \E\qty[\sum_{m\le M}w_{klm}] )]  + \E\qty[\sum_{k \geq K+1}\pi_k] \tag{$\because w_{kl} \perp w_{klm}$ }\\
    &= \qty(1 - \qty(\frac{\astar}{1+\astar})^K ) \qty[1 - \qty{ 1 - \qty(\frac{\atheta}{1+\atheta})^L}  \qty{ 1 - \qty(\frac{\aphi}{1+\aphi})^M}] + \qty(\frac{\astar}{1+\astar})^K \\
    &\leq \qty[1 - \qty{ 1 - \qty(\frac{\atheta}{1+\atheta})^L}  \qty{ 1 - \qty(\frac{\aphi}{1+\aphi})^M}] + \qty(\frac{\astar}{1+\astar})^K \\
    &\leq \qty(\frac{\astar}{1+\astar})^K + \qty(\frac{\atheta}{1+\atheta})^L + \qty(\frac{\aphi}{1+\aphi})^M.
\end{align*}
\end{proof}

\subsubsection{Proof of Theorem \ref{thm:finite_approximation_marginal}}
We use the following two lemma to prove Theorem \ref{thm:finite_approximation_marginal}.

\begin{lemma}
\label{lemma:TV_ineq}
    Given distribution $f(y_{ij}\mid\boldsymbol{\theta}_{ij})$, $F_{i}^{(1)}(\boldsymbol{\theta}_{ij})\mid Q$, $F_{i}^{(2)}(\boldsymbol{\theta}_{ij})\mid Q$ and Q, then the total variation of joint marginal distribution of $\{y_{ij}\}_{i,j}$ between $F^{(1)}$ and $F^{\prime(2)}$ is smaller than the total variation of joint distribution of $\{\boldsymbol{\theta}_{ij}\}_{i,j}$, i.e.,
    \begin{align*}
        TV\left(\E_{Q}\left[\prod_{i,j}\int_{\Theta}f(\cdot\mid\boldsymbol{\theta}_{ij})F_{i}^{(1)}(d\boldsymbol{\theta}_{ij})\right],\E_{Q}\left[\prod_{i,j}\int_{\Theta}f(\cdot\mid\boldsymbol{\theta}_{ij})F_{i}^{(2)}(d\boldsymbol{\theta}_{ij})\right]\right) \\    
        \le TV\left(\E_{Q}\left[\prod_{i,j}F_{i}^{(1)}(d\boldsymbol{\theta}_{ij})\right],\E_{Q}\left[\prod_{i,j}F_{i}^{(2)}(d\boldsymbol{\theta}_{ij})\right]\right).
    \end{align*}    
\end{lemma}

\begin{proof}
    Using Fubini's theorem,
    \begin{align*}
        & TV\left(\E_{Q}\left[\prod_{i,j}\int_{\Theta}f(y_{ij}\mid\boldsymbol{\theta}_{ij})F_{i}^{(1)}(d\boldsymbol{\theta}_{ij})\right],\E_{Q}\left[\prod_{i,j}\int_{\Theta}f(y_{ij}\mid\boldsymbol{\theta}_{ij})F_{i}^{(2)}(d\boldsymbol{\theta}_{ij})\right]\right) \\
        &= \frac{1}{2}\int_{y^{N}}\left\lvert\E_{Q}\left[\prod_{i,j}\int_{\Theta}f(y_{ij}\mid\boldsymbol{\theta}_{ij})F_{i}^{(1)}(d\boldsymbol{\theta}_{ij})-\prod_{i,j}\int_{\Theta}f(y_{ij}\mid\boldsymbol{\theta}_{ij})F_{i}^{(2)}(d\boldsymbol{\theta}_{ij})\right]\right\rvert\prod_{i,j}dy_{ij} \\
        &\stackrel{(*)}{=} \frac{1}{2}\int_{y^{N}}\left\lvert\E_{Q}\left[\int_{\Theta^{N}}\prod_{i,j}f(y_{ij}\mid\boldsymbol{\theta}_{ij})\prod_{i,j}F_{i}^{(1)}(d\boldsymbol{\theta}_{ij})-\int_{\Theta^{N}}\prod_{i,j}f(y_{ij}\mid\boldsymbol{\theta}_{ij})\prod_{i,j}F_{i}^{(2)}(d\boldsymbol{\theta}_{ij})\right]\right\rvert\prod_{i,j}dy_{ij} \\
        &= \frac{1}{2}\int_{y^{N}}\left\lvert\int_{\Theta^{N}}\prod_{i,j}f(y_{ij}\mid\boldsymbol{\theta}_{ij})\E_{Q}\left[\prod_{i,j}F_{i}^{(1)}(d\boldsymbol{\theta}_{ij})-\prod_{i,j}F_{i}^{(2)}(d\boldsymbol{\theta}_{ij})\right]\right\rvert\prod_{i,j}dy_{ij} \\
        &\le \frac{1}{2}\int_{y^{N}}\int_{\Theta^{N}}\prod_{i,j}f(y_{ij}\mid\boldsymbol{\theta}_{ij})\left\lvert\E_{Q}\left[\prod_{i,j}F_{i}^{(1)}(d\boldsymbol{\theta}_{ij})-\prod_{i,j}F_{i}^{(2)}(d\boldsymbol{\theta}_{ij})\right]\right\rvert\prod_{i,j}dy_{ij} \\
        &= \frac{1}{2}\int_{\Theta^{N}}\int_{Y^{N}}\prod_{i,j}f(y_{ij}\mid\boldsymbol{\theta}_{ij})\prod_{i,j}dy_{ij}\left\lvert\E_{Q}\left[\prod_{i,j}F_{i}^{(1)}(d\boldsymbol{\theta}_{ij})-\prod_{i,j}F_{i}^{(2)}(d\boldsymbol{\theta}_{ij})\right]\right\rvert \\
        &= \frac{1}{2}\int_{\Theta^{N}}\left\lvert\E_{Q}\left[\prod_{i,j}F_{i}^{(1)}(d\boldsymbol{\theta}_{ij})-\prod_{i,j}F_{i}^{(2)}(d\boldsymbol{\theta}_{ij})\right]\right\rvert \\
        &= TV\left(\E_{Q}\left[\prod_{i,j}F_{i}^{(1)}(d\boldsymbol{\theta}_{ij})\right],\E_{Q}\left[\prod_{i,j}F_{i}^{(2)}(d\boldsymbol{\theta}_{ij})\right]\right)
    \end{align*}
    Equation (*) follows from Fubini's theorem as below,
    \begin{align*}
        \prod_{i=1}^{N}\int_{\Theta}g(\boldsymbol{\theta}_{i})F(d\boldsymbol{\theta}_{i}) &=\prod_{i=1}^{N}\int_{\Theta}g(\boldsymbol{\theta}_{i})f(\boldsymbol{\theta}_{i})d\boldsymbol{\theta}_{i} \\
        &= \int_{\Theta^{N}}\prod_{i=1}^{N}[g(\boldsymbol{\theta}_{i})f(\boldsymbol{\theta}_{i})]d(\boldsymbol{\theta}_{1},...,\boldsymbol{\theta}_{N}) \\
        &= \int_{\Theta^{N}}\prod_{i=1}^{N}[g(\boldsymbol{\theta}_{i})]\prod_{i=1}^{N}[f(\boldsymbol{\theta}_{i})]d(\boldsymbol{\theta}_{1},...,\boldsymbol{\theta}_{N}) \\
        &= \int_{\Theta^{N}}\prod_{i=1}^{N}f(\boldsymbol{\theta}_{i})\prod_{i=1}^{N}[F(d\boldsymbol{\theta}_{i})]
    \end{align*}
\end{proof}

\begin{lemma}
\label{lemma:prod_ineq}
    Let $a_{1},...,a_{n}$, $b_1,..., b_n \in [-1,1]$, then
    $$
    \left\lvert\prod_{i=1}^{n}a_{i}-\prod_{i=1}^{n}b_{i}\right\rvert\le\sum_{i=1}^{n}\lvert a_{i}-b_{i}\rvert
    $$
\end{lemma}
\begin{proof}
    \begin{align*}
    \left\lvert\prod_{i=1}^{n}a_{i}-\prod_{i=1}^{n}b_{i}\right\rvert &= \left\lvert\prod_{i=1}^{n}a_{i}-b_{1}\prod_{i=2}^{n}a_{i}+b_{1}\prod_{i=2}^{n}a_{i}-\prod_{i=1}^{n}b_{i}\right\rvert \\
    &\le \lvert a_{1}-b_{1}\rvert\left\lvert\prod_{i=2}^{n}a_{i}\right\rvert+\lvert b_{1}\rvert\left\lvert\prod_{i=2}^{n}a_{i}-\prod_{i=2}^{n}b_{i}\right\rvert \\
    &\le \lvert a_{1}-b_{1}\rvert+\left\lvert\prod_{i=2}^{n}a_{i}-\prod_{i=2}^{n}b_{i}\right\rvert
    \end{align*}
    Then the result follows by induction.    
\end{proof}

Recall: 
\begin{align*}
    m(V,X,Y) &= \E\qty[\prod_{i=1}^{I}\prod_{j=1}^{N_i} \int_{\Theta \times \Phi \times \mathcal{H}} f(y_{ij}, x_{ij},  v_i \mid \boldsymbol{\theta}, \boldsymbol{\phi}, \boldsymbol{\eta}) \widetilde{F}_i(d\boldsymbol{\theta}, d\boldsymbol{\phi}, d\boldsymbol{\eta})] \\
    m^{(K,L,M)}(V,X,Y) &= \E\qty[\prod_{i=1}^{I}\prod_{j=1}^{N_i} \int_{\Theta \times \Phi \times \mathcal{H}} f(y_{ij}, x_{ij},  v_i \mid \boldsymbol{\theta}, \boldsymbol{\phi}, \boldsymbol{\eta}) \widetilde{F}_i^{(K,L,M)}(d\boldsymbol{\theta}, d\boldsymbol{\phi}, d\boldsymbol{\eta})]
\end{align*}
Then using Lemma \ref{lemma:TV_ineq} and \ref{lemma:prod_ineq}, we have
\begin{align*}
    &TV\left(m^{(K,L,M)}(V,X,Y),m(V,X,Y)\right) \\
    &\le TV\left(\E\left[\prod_{i=1}^{I}\prod_{j=1}^{N_{i}}\widetilde{F}_{i}(\boldsymbol{\theta}_{ij}, \boldsymbol{\phi}_{ij}, \boldsymbol{\eta}_i) \right],\E\left[\prod_{i=1}^{I}\prod_{j=1}^{N_{i}}\widetilde{F}_{i}^{(K,L,M)}(\boldsymbol{\theta}_{ij}, \boldsymbol{\phi}_{ij}, \boldsymbol{\eta}_i)\right]\right) \tag{$\because$ Lemma \ref{lemma:TV_ineq}}\\
    &= \sup_{(\boldsymbol{\theta}_{ij}, \boldsymbol{\phi}_{ij}, \boldsymbol{\eta}_i)\in \Theta \times \Phi \times \mathcal{H}}\left\lvert\E\left[\prod_{i=1}^{I}\prod_{j=1}^{N_{i}}\widetilde{F}_{i}(\boldsymbol{\theta}_{ij}, \boldsymbol{\phi}_{ij}, \boldsymbol{\eta}_i)-\prod_{i=1}^{I}\prod_{j=1}^{N_{i}}\widetilde{F}_{i}^{(K,L,M)}(\boldsymbol{\theta}_{ij}, \boldsymbol{\phi}_{ij}, \boldsymbol{\eta}_i)\right]\right\rvert \\
    &\le \sup_{(\boldsymbol{\theta}_{ij}, \boldsymbol{\phi}_{ij}, \boldsymbol{\eta}_i)\in \Theta \times \Phi \times \mathcal{H}}\E\left[\left\lvert\prod_{i=1}^{I}\prod_{j=1}^{N_{i}}\widetilde{F}_{i}(\boldsymbol{\theta}_{ij}, \boldsymbol{\phi}_{ij}, \boldsymbol{\eta}_i)-\prod_{i=1}^{I}\prod_{j=1}^{N_{i}}\widetilde{F}_{i}^{(K,L,M)}(\boldsymbol{\theta}_{ij}, \boldsymbol{\phi}_{ij}, \boldsymbol{\eta}_i)\right\rvert\right] \\
    &\le \E\left[\sup_{(\boldsymbol{\theta}_{ij}, \boldsymbol{\phi}_{ij}, \boldsymbol{\eta}_i)\in \Theta \times \Phi \times \mathcal{H}}\left\lvert\prod_{i=1}^{I}\prod_{j=1}^{N_{i}}\widetilde{F}_{i}(\boldsymbol{\theta}_{ij}, \boldsymbol{\phi}_{ij}, \boldsymbol{\eta}_i)-\prod_{i=1}^{I}\prod_{j=1}^{N_{i}}\widetilde{F}_{i}^{(K,L,M)}(\boldsymbol{\theta}_{ij}, \boldsymbol{\phi}_{ij}, \boldsymbol{\eta}_i)\right\rvert\right] \\
    &\le \E\left[\sup_{(\boldsymbol{\theta}_{ij}, \boldsymbol{\phi}_{ij}, \boldsymbol{\eta}_i)\in \Theta \times \Phi \times \mathcal{H}}\sum_{i=1}^{I}\sum_{j=1}^{N_{i}}\left\lvert \widetilde{F}_{i}(\boldsymbol{\theta}_{ij}, \boldsymbol{\phi}_{ij}, \boldsymbol{\eta}_i)-\widetilde{F}_{i}^{(K,L,M)}(\boldsymbol{\theta}_{ij}, \boldsymbol{\phi}_{ij}, \boldsymbol{\eta}_i)\right\rvert\right] \tag{$\because$ Lemma \ref{lemma:prod_ineq}}\\
    &\le \E\left[\sum_{i=1}^{I}\sum_{j=1}^{N_{i}}\sup_{(\boldsymbol{\theta}_{ij}, \boldsymbol{\phi}_{ij}, \boldsymbol{\eta}_i)\in \Theta \times \Phi \times \mathcal{H}}\left\lvert \widetilde{F}_{i}(\boldsymbol{\theta}_{ij}, \boldsymbol{\phi}_{ij}, \boldsymbol{\eta}_i)-\widetilde{F}_{i}^{(K,L,M)}(\boldsymbol{\theta}_{ij}, \boldsymbol{\phi}_{ij}, \boldsymbol{\eta}_i)\right\rvert\right] \\
    &=\E\left[\sum_{i=1}^{I}\sum_{j=1}^{N_{i}} TV\left(\widetilde{F}_{i},\widetilde{F}_{i}^{(K,L,M)}\right)\right] \\
    &\le N\left(\left(\frac{\astar}{1+\astar}\right)^{M}+\left(\frac{\atheta}{1+\atheta}\right)^{K}+\left(\frac{\aphi}{1+\aphi}\right)^{L}\right) \tag{$\because$ Theorem \ref{thm:finite_approximation}} 
\end{align*}
Here $N=\sum_{i=1}^{I}\sum_{j=1}^{N_{i}}1$ denotes the entire sample size.

\section{Binary post-treatment confounders}
When the post-treatment confounder $D_{ij}(a)$ is binary, we assume there exists a continuous latent random variable $Z_{ij}(a)$ such that $D_{ij}(a)=1$ if $Z_{ij}(a) \geq 0$ and the joint distribution of continuous latent variables follows a multivariate Gaussian copula model,
\begin{equation*}
    \begin{split}
        F_{\mathbf{Z}_{i}(1),\mathbf{Z}_{i}(0)}(\mathbf{z},\mathbf{z}' \mid \mathbf{C}_{i},N_i) =& \; \Phi_{2N_i} \left[  \Phi_1^{-1}\qty{F_{Z_{i1}(1)}(z_1) \mid \mathbf{C}_{i1},N_i) }, 
        \Phi_1^{-1}\qty{F_{Z_{i1}(0)}(z_1') \mid \mathbf{C}_{i1},N_i) }, \ldots, \right. \\
        & \left. \ \ \ \ \ \
        \Phi_1^{-1}\qty{F_{Z_{iN_i}(1)}(z_{N_i}) \mid \mathbf{C}_{iN_i},N_i) }, 
        \Phi_1^{-1}\qty{F_{Z_{iN_i}(0)}(z_{N_i}') \mid \mathbf{C}_{iN_i},N_i) }, \,  \boldsymbol{\Omega} \right],
    \end{split}
    \end{equation*}
    for any $\mathbf{z},\mathbf{z}'$.
This formulation induces the counterfactual dependence structure between $D_{ij}(a)$ and $D_{ij'}(a')$, while enabling the data-augmentation approach when the post-treatment confounders are binary.

\section{Comparison with the nested dependent Dirichlet process}
\label{sec:comparison_nddp}
\citet{ohnishi2025} proposed the nested dependent Dirichlet process (nDDP) prior, which is defined using the atom processes indexed by cluster- and individual-level covariates and a copula that defines the dependence structure between two arbitrary atom processes. Note that the copula here is not the one that we used to model the cross-world dependence of the post-treatment confounders in the main manuscript. Their method is specifically designed to model the conditional distribution of multilevel outcomes in CRTs. 
One advantage of the nDDP over the CA-EDP is its capacity to allow dependence structures between atom processes \citep{ohnishi2025}, such as those induced by a Gaussian copula, thereby the nDDP may capture complex data‐generating mechanisms of outcomes that the CA-EDP cannot. Additionally, the nDDP allows the stick-breaking weights to depend on cluster-level covariates, which could also enhance the flexibility in a context where the cluster-level unobserved heterogeneity is well explained by cluster-level covariates.

Additionally, the nDDP takes a conditional approach. Consequently, when computing the marginal expectation of the outcome, it forgoes an explicit probabilistic model for the covariates and instead approximates their joint distribution with the empirical distribution observed in the data. 
This is often natural in Bayesian analyses, since modeling covariates can demand substantial effort and may not yield commensurate benefits.
However, it is important to note that their estimator targets the \emph{mixed average treatment effect} (MATE) \citep{FanLi2023} rather than the desired population average treatment effect (PATE), that is, $\E[Y] =\int_{\mathrm{N}}
          \int_{\mathcal{X}}
            \E\left[Y \mid N=n,X=x\right]
          dF_{X \mid N}(x)  dF_{N}(n)
        \approx
        \frac{1}{I}\sum_{i=1}^{I}\frac{1}{N_i}\sum_{j=1}^{N_i}
          \E\left[Y \mid N_i,X_{ij}\right]$. 
The MATE is a good approximation to the PATE when the number of clusters and individuals is sufficiently large and the empirical distribution is `close' to the truth; however, in CRTs with only a small number of clusters, the data may not support reliable estimation of the between-cluster heterogeneity needed to recover the PATE.

A key advantage of the CA-EDP is its ability to model covariate distributions at both the cluster and individual levels.
The CA-EDP specifies a hierarchical prior for the full data-generating mechanism, $p(N,X,Y)=p(Y\mid N,X)p(N,X)$, and marginalizes $E[Y]$ over the posterior of $(N,X)$ rather than replacing $p(N,X)$ with an empirical point mass. 
Modeling covariates offers several  benefits. First, finite-sample efficiency can be better when the underlying model for $(N,X)$ is complex with a modest number of cluster such that the sample used in the empirical distribution does not approximate it well. Also, the posterior for the expectation of $Y$ automatically propagates uncertainty in the covariate distribution and the cluster size.
Additionally, our estimands accommodate informative cluster size by taking the expectation of cluster-averaged potential outcomes with respect to the cluster-size distribution $F_N$. The nDDP approximates $F_N$ by its empirical distribution and therefore does not capture the nuanced influence of informative cluster size, whereas the CA-EDP explicitly models $F_N$ and thereby captures these nuances.
Finally, the CA-EDP naturally accommodates ignorable covariate missingness.

Taken together, these approaches involve a familiar trade-off. With few clusters, the nDDP's use of empirical covariate and cluster-size distributions can make the MATE a poor proxy for the PATE and can under-propagate uncertainty; in this setting, the CA-EDP's joint modeling provides more complete uncertainty quantification, including informative cluster size and missingness, though it introduces additional modeling work and may lose efficiency when the data are weakly informative. When the number of clusters is sufficiently large, the nDDP typically yields an accurate approximation to the PATE, benefiting from its flexible dependence structure, yet uncertainty may still be somewhat under-propagated; correspondingly, the incremental efficiency gains from CA-EDP may be limited. The choice should reflect study size, the complexity of $p(N,X)$, and the desired balance between modeling burden and robustness in uncertainty quantification. 

\section{Details of the blocked Gibbs sampler}
\label{sec:gibbs_details}
\subsection{Blocked Gibbs sampler for nEPDM}

The posterior inference is carried out with a three-level blocked Gibbs sampler
Let $\zeta^{(n)}_i \in \{1,\ldots,K\}$, $\zeta^{(y)}_{ij} \in \{1,\ldots,L\}$, and $\zeta^{(x)}_{ij} \in \{1,\ldots,M\}$ denote the latent group indicators; that is, the cluster-level indicator and the individual-level indicators, respectively.
These follow multinomial distributions 
$\zeta^{(n)}_i \sim \mathrm{MN}(\boldsymbol{\pi}^{*})$, 
$\zeta^{(y)}_{ij} \sim \mathrm{MN}(\mathbf{w}^{\theta}_{\zeta^{(n)}_{i}})$, and 
$\zeta^{(x)}_{ij} \sim \mathrm{MN}(\mathbf{w}^{\phi}_{\zeta^{(n)}_{i}\zeta^{(y)}_{ij}})$,
where 
$\boldsymbol{\pi}^{*} = (\pi^{*}_{1},\ldots,\pi^{*}_{K})^\top$,  
$\mathbf{w}^{\theta}_{\zeta^{(n)}_{i}} = (w^{\theta}_{\zeta^{(n)}_{i}1},\ldots,w^{\theta}_{\zeta^{(n)}_{i}L})^\top$, and 
$\mathbf{w}^{\phi}_{\zeta^{(n)}_{i}\zeta^{(y)}_{ij}} = (w^{\phi}_{\zeta^{(n)}_{i}\zeta^{(y)}_{ij}1},\ldots,w^{\phi}_{\zeta^{(n)}_{i}\zeta^{(y)}_{ij}M})^\top$  are the vectors of  weights from the CA-EDP. 
This section presents the algorithm for continuous outcomes, mediators, and post-treatment confounders.  In what follows, we assume the observed data are conditioned on unless noted otherwise.

\subsubsection{Sample $\boldsymbol\zeta^{(n)}$ (cluster level)}

Conditional on all the other parameters and observed data, the full conditional probabilities of the group indicator $\zeta^{(n)}_i$ for cluster $i=1,\ldots,I$ are proportional to
\begin{align*}
P\left(\zeta^{(n)}_i=k\mid\cdot\right)
  &\propto
    \pi^{*}_k
    p\left(N_i \mid \lambda^{(n)}_{k}\right)
    \prod_{j=1}^{N_i}
      \left[
        \sum_{l=1}^{L}
        w^{\theta}_{kl}
        p\left(D_{ij},M_{ij},Y_{ij}\mid \boldsymbol{\theta}_{l}\right)
        \left\{\sum_{m=1}^{M}
        w^{\phi}_{klm}
        p\left( \mathbf{X}_{ij}\mid \phi_{m} \right)
        \right\}
      \right],
\end{align*}
where $p\left(N_i \mid \lambda^{(n)}_{k}\right)$, $p\left(D_{ij},M_{ij},Y_{ij}\mid \boldsymbol{\theta}_{l}\right)$ and $p\left( \mathbf{X}_{ij}\mid \phi_{m} \right)$ are the densities in the CA-EDP.

\subsubsection{Sample $\boldsymbol\zeta^{(y)}$ (D/M/Y level)}
For each individual $j$ in cluster $i$ with current $\zeta^{(n)}_{i}$, draw from the conditional probabilities for $l=1,\ldots,L$:
\begin{align*}
    P\left(\zeta^{(y)}_{ij}= l\mid\cdot\right)
  \propto
  w^{\theta}_{\zeta^{(n)}_{i}l}
        p\left(D_{ij},M_{ij},Y_{ij}\mid \boldsymbol{\theta}_{l}\right)
        \left\{\sum_{m=1}^{M}
        w^{\phi}_{\zeta^{(n)}_{i}lm}
        p\left( \mathbf{X}_{ij}\mid \phi_{m} \right)
        \right\}.
\end{align*}

\subsubsection{Sample $\boldsymbol\zeta^{(x)}$ (X level)}

For each coordinate $t=1,\ldots,p$ of individual $j$ in cluster $i$ with current values $\zeta^{(n)}_{i}$ and $\zeta^{(y)}_{ij}$, draw from the conditional probabilities over $m=1,\ldots,M$: 
\begin{align*}
    P\left(\zeta^{(x)}_{ij}= m\mid\cdot\right)
  \propto
        w^{\phi}_{\zeta^{(n)}_{i}\zeta^{(y)}_{ij}m}
        p\left( \mathbf{X}_{ij}\mid \phi_{m} \right)
\end{align*}

\subsubsection{Sample stick-breaking weights
$(s^{*}_k,v^{\theta}_{kl},v^{\phi}_{klm})$ and
$(\pi^{*}_k,w^{\theta}_{kl},w^{\phi}_{klm})$}

Let $s^{*}_{K}=1$. Given $\astar$ and $\zeta^{(n)}_i$, draw $s^{*}_k$ for $k=1,\ldots,K-1$ from
\begin{align*} 
    s^{*}_k \sim \text{Be}\left(1 + \sum_{i=1}^{I} \mathbbm{1}(\zeta^{(n)}_i = k), \astar + \sum_{i=1}^{I} \mathbbm{1}(\zeta^{(n)}_i > k)\right).
\end{align*}
Then update $\pi^{*}_k = s^{*}_k \prod_{j=1}^{k-1} (1 - s^{*}_j)$. Then,
for each group $k$, let $v^{\theta}_{kL}=1$. 
Given $\atheta$ and $\zeta^{(n)}_i$, draw  $v^{\theta}_{kl}$ for $l=1,\ldots,L-1$ from
\begin{align*} 
    v^{\theta}_{kl} &\sim \text{Be}\left(1 + \sum_{i=1}^{I}\sum_{j=1}^{N_i} \mathbbm{1}(\zeta^{(y)}_{ij} = l, \zeta^{(n)}_i = k), \atheta + \sum_{i=1}^{I}\sum_{j=1}^{N_i} \mathbbm{1}(\zeta^{(y)}_{ij} > l, \zeta^{(n)}_i = k) \right). 
\end{align*}
Then update $w^{\theta}_{kl} = v^{\theta}_{kl} \prod_{j=1}^{l-1} (1 - v^{\theta}_{kj})$ for $k=1,\ldots,K$. 
Finally, for each $k$ and $l$, let $v^{\phi}_{klM}=1$
Given $\aphi$, $\zeta^{(n)}_i$ and $\zeta^{(y)}_{ij}$, draw  $v^{\phi}_{klm}$ for $m=1,\ldots,M-1$ from
\begin{align*} 
    v^{\phi}_{klm} \sim \text{Be} & \left( 1 + \sum_{i=1}^{I}\sum_{j=1}^{N_i} \mathbbm{1}(\zeta^{(x)}_{ij} = m, \zeta^{(y)}_{ij} = l, \zeta^{(n)}_i = k) \right.\\
    &\left. , \aphi + \sum_{i=1}^{I}\sum_{j=1}^{N_i} \mathbbm{1}(\zeta^{(x)}_{ij} > m, \zeta^{(y)}_{ij} = l, \zeta^{(n)}_i = k) \right). 
\end{align*}
Then update $w^{\phi}_{klm} = v^{\phi}_{klm} \prod_{t=1}^{m-1} (1 - v^{\phi}_{klt})$ for $k=1,\ldots,K$ and $l=1,\ldots,L$.

\subsubsection{Update concentration parameters $\astar$, $\atheta$ and $\aphi$}
Assuming the conjugate priors 
$\astar \sim \mathrm{Ga}(a_{\astar}, b_{\astar})$, 
$\atheta \sim \mathrm{Ga}(a_{\atheta}, b_{\atheta})$ and 
$\aphi \sim \mathrm{Ga}(a_{\alpha_{\phi}}, b_{\alpha_{\phi}})$, 
update the concentration parameters $\astar$, $\atheta$ and $\aphi$ as follows:
\begin{align*} 
    \astar &\sim \mathrm{Ga}\left( a_{\astar} + K - 1, b_{\astar} - \sum_{k=1}^{K - 1} \ln(1 - s^{*}_k)  \right), \\
    \atheta &\sim \mathrm{Ga}\left( a_{\atheta} + K(L - 1), b_{\atheta} - \sum_{k=1}^{K}\sum_{l=1}^{L - 1} \ln(1 - v^{\theta}_{kl}) \right), \\
    \alpha_{\phi} &\sim \mathrm{Ga}\left( a_{\alpha_{\phi}} + KL(M - 1), b_{\alpha_{\phi}} - \sum_{k=1}^{K}\sum_{l=1}^{L}\sum_{m=1}^{M-1} \ln(1 - v^{\phi}_{klm}) \right)
\end{align*}

\subsubsection{Sample $\lambda^{(n)}_{k}$}

For $k=1,\ldots,K$, assuming a Gamma base measure $\lambda^{(n)}_{k}\sim\mathrm{Ga}(a_{0N},b_{0N})$, 
\begin{align*}
    \lambda^{(n)}_{k} \sim
\mathrm{Ga}\left(a_{0N}+\sum_{i:\zeta^{(n)}_i=k}N_i,
                   b_{0N}+n_k\right).
\end{align*}

\subsubsection{Update atoms $\boldsymbol{\theta}_{l}$}

For $l=1,\ldots,L$, let
$\mathbf C^{(d)}_{l}$,
$\mathbf C^{(m)}_{l}$ and
$\mathbf C^{(y)}_{l}$ be the design matrices (including cluster means) and
$\mathbf D_{l}$, $\mathbf M_{l}$, $\mathbf Y_{l}$ the response
vectors for the observations currently assigned to $\zeta^{(y)}_{ij}=l$.
For $\bullet \in \{d,m,y\}$,
$(\sigma^{\bullet}_{l})^{2}\sim\mathrm{IG}(a_{0,\bullet},b_{0,\bullet})$ and
$\boldsymbol\beta^{\bullet}_{l}
     \sim\mathrm N(\boldsymbol\mu^{\bullet}_{0},\boldsymbol\Sigma^{\bullet}_{0})$,
the posterior updates are
\begin{align*}
(\sigma^{\bullet}_{l})^{2}\sim
  \mathrm{IG}\left(a^{\bullet}_{0}+{\frac12}n_{l},
                     b^{\bullet}_{0}+{\frac12}\|{\mathbf r^{\bullet}}\|^2\right), ~~~ \boldsymbol\beta^{\bullet}_{l}\sim
  \mathrm N\left(\boldsymbol\mu_{l}^{\bullet},
                    \boldsymbol\Sigma_{l}^{\bullet}\right),
\end{align*}
with $\mathbf r^{\bullet}$ the residuals at the current draw and
$
\boldsymbol\Sigma_{l}^{\bullet}=
  (\boldsymbol\Sigma_{0,\bullet}^{-1}
   +\frac{1}{(\sigma^{\bullet}_{l})^{2}}  \mathbf C^{\bullet\top}_{l}\mathbf C^{\bullet}_{l})^{-1},
\boldsymbol\mu_{l}^{\bullet}=
  \boldsymbol\Sigma_{l}^{\bullet}
  ((\boldsymbol\Sigma^{\bullet}_{0})^{-1}\boldsymbol\mu^{\bullet}_{0}
   +\frac{1}{(\sigma^{\bullet}_{l})^{2}} \mathbf R_{l}^\bullet),
$
where $\mathbf R_{l}^\bullet$ is the least square estimate for $\beta^{\bullet}_{l}$, e.g., $\mathbf R_{l}^{(d)} = (\mathbf C^{d\top}_{l}\mathbf C^{(d)}_{l})^{-1}\mathbf C^{d\top}_{l}\mathbf D_{l}$.

\subsubsection{Update atoms $\boldsymbol{\phi}_{m}$}
For $m=1,\ldots,M$, $\phi_{m}$ is sampled from the following conditinal.
\begin{align*}
    \phi_{m} \sim \qty(\prod_{i,j:\zeta^{(x)}_{ij}=m}p\qty(\mathbf{X}_{ij} \mid \phi_{m})) g_{\phi}(\phi_{m}),
\end{align*}
where $g_{\phi}$ is the prior distribution of $\phi_{m}$ induced by the base measure $G_{\phi}$ and $p\qty(\mathbf{X}_{ij} \mid \phi_{m}))$ is the application-specific distribution of $\mathbf X_{ij}$.

\subsection{G-computation and post-processing}
\label{sec:gcomp_postprocessing}

For each retained MCMC draw we compute the causal estimands using $T=100$ synthetic clusters. The number of synthetic cluster is chosen such that the posterior estimates are stable. For $t=1,\ldots,T$, we first estimate the same-world correlation parameters $\gamma_a$ given the observed marginals  using the Metropolis-Hastings (MH) alrorithm in step \ref{step:estimate_gamma}. Given the estimated $\gamma_a$, we move on to the g-computation steps, where we iterate steps\ref{step:begin}-\ref{step:last}, and then aggregate the sample in step \ref{step:agg}

\begin{enumerate}[label=(\alph*)]
\item \emph{Sample the same-world correlations $\gamma_1$ and $\gamma_0$ using the MH algorithm.}
\label{step:estimate_gamma}
Using observed data under $A_i=a\in\{1,0\}$, compute the rank uniforms and latent normals $u^{(a)}_{ij} = F^{(a)}_{kl}\big(D_{ij}(a)\big)$ and $z^{(a)}_{ij} = \Phi^{-1}\big(u^{(a)}_{ij}\big)$ for $j=1,\ldots,N_i,$
and collect $\mathbf z_i^{(a)}=(z^{(a)}_{i1},\ldots,z^{(a)}_{iN_i})^\top$.
Under an equicorrelation Gaussian copula for the same world, let $R_{i}^{(a)}(\gamma_a) = (1-\gamma_a)I_{N_i} + \gamma_a J_{N_i},$
and the copula (log-)likelihood for world $a$ is given by:
\begin{align*}
    \ell_i^{(a)}(\gamma_a)=\log c\big(\mathbf u_{i}^{(a)};R_{i}^{(a)}(\gamma_a)\big)
    = -\frac{1}{2}\log|R_{i}^{(a)}|
    -\frac{1}{2}\,\mathbf z_i^{(a)\top}\big((R_{i}^{(a)})^{-1}-I_{N_i}\big)\mathbf z_i^{(a)}.
\end{align*}
For $R_{i}^{(a)}(\gamma_a)$, the necessary terms are:
\begin{align*}
    \log|R_{i}^{(a)}|=(N_i-1)\log(1-\gamma_a)+\log\big(1+(N_i-1)\gamma_a\big),\\  
    (R_{i}^{(a)})^{-1}-I_{N_i} =\frac{\gamma_a}{1-\gamma_a}I_{N_i}
    -\frac{\gamma_a}{(1-\gamma_a)\big(1+(N_i-1)\gamma_a\big)}\,J_{N_i}.
\end{align*}
Write $s^{(a)}_{i1}=\sum_{j=1}^{N_i} z^{(a)}_{ij}$ and $s^{(a)}_{i2}=\sum_{j=1}^{N_i} \big(z^{(a)}_{ij}\big)^2$. Then, aggregating across all clusters $i=1,\ldots,I$, the log-likelihood is given by:
\begin{align*}
    \ell^{(a)}(\gamma_a)= \sum_{i:A_i=a} \ell_{i}^{(a)}(=\gamma_a)
    &-\frac{1}{2} \sum_{i:A_i=a} \left[\Big\{(N_i-1)\log(1-\gamma_a)+\log\big(1+(N_i-1)\gamma_a\big)\Big\} \right.\\
    &\left. +\left\{
    \frac{\gamma_a}{1-\gamma_a}\,s^{(a)}_{i2}
    -\frac{\gamma_a}{(1-\gamma_a)\big(1+(N_i-1)\gamma_a\big)}\,(s^{(a)}_{i1})^2
    \right\} \right].
\end{align*}

For priors, we use independent uniforms respecting the PD constraint e.g., $\gamma_a \sim \mathrm{Unif}(0,1),$ for $a\in\{1,0\}$, which satisfy the condition \eqref{eq:cond_posdef}. For the proposal distribution, we sample from the prior: $\gamma_1^{\mathrm{prop}}\sim \mathrm{Unif}(0,1)$, $\gamma_0^{\mathrm{prop}}\sim \mathrm{Unif}(0,1)$.

The joint log-posterior (up to an additive constant) is:
\begin{align*}
    \mathcal L(\gamma_1,\gamma_0) =\ell^{(1)}(\gamma_1)+\ell^{(0)}(\gamma_0)+\log \pi(\gamma_1)+\log \pi(\gamma_0),
\end{align*}
with $\pi(\cdot)$ the uniform prior.
Accept $(\gamma_1^{\mathrm{prop}},\gamma_0^{\mathrm{prop}})$ with probability
\begin{align*}
    \alpha=\min\left\{1,
    \exp\qty(\mathcal L(\gamma_1^{\mathrm{prop}},\gamma_0^{\mathrm{prop}})
    -\mathcal L(\gamma_1^{\mathrm{prev}},\gamma_0^{\mathrm{prev}})) \right\}.
\end{align*}
The initial values of $\gamma_a$ are drawn from the prior.

\item \emph{Draw latent groups.}
\label{step:begin}
      Sample $k,l,m$ from the weights
      $(\pi_k,w^{\theta}_{kl},w^{\phi}_{klm})$.

\item \emph{Draw cluster size and baseline covariates.}
\label{step:clustersize}
      Draw
      $N\sim\mathrm{Pois}(\lambda^{(n)}_{k})$.
      For $j=1,\dots,N$ draw
      $\mathbf X_{j}\sim
      \mathrm p\left(\cdot \mid \boldsymbol\phi_{m}\right)$.

\item \emph{Post-treatment confounder $D$.}
      Compute the design matrices
      $\mathbf C^{(d)}_{j}(a)=(1,a,N,\mathbf X_j)$  and sample
      $D_{j}(a)\sim \mathrm N\left(\mathbf C^{(d)}_{j}(a)\boldsymbol\beta^{(d)}_{l}, (\sigma^{(d)}_{l})^2\right)$
      for $a\in\{0,1\}$.

\item \emph{Counterfactual post-treatment confounder $\mathbf D_i(0)$}
For cluster size $N$, let $F^{(a)}$ denote the current CA-EDP-mixture marginal CDF of $D_j(a)$ for arm $a\in\{0,1\}$ and $j=1,\ldots,N$, evaluated at unit $j$'s covariates $\mathbf C^{(d)}_{j}(a)$. Under the finite approximation, we have $F^{(a)}(D_j(a)) = \sum_{k}^{K}\sum_{l}^{L}\pi_k w_{kl}^{\theta} \Phi(D_j(a) \mid \mathbf C^{(d)}_{j}(a), \boldsymbol{\theta}_l)$, where $\Phi$ is the Gaussian CDF, parameterized by  $\boldsymbol{\theta}_l$. Proceed as follows:
\begin{enumerate}[label*=\arabic*.]

\item For each $j=1,\ldots,N$, compute $p_{j}=F^{(1)}\big(D_{j}(1)\big)$ and  $z_{j}=\Phi^{-1}(p_{j})$
and collect $\mathbf z_1=(z_{1},\ldots,z_{N})^\top$.

\item Sample the cross-world parameter $\rho$ from the prior \eqref{eq:prior_rho}, that respects the bound in (\ref{eq:cond_posdef}),
and set $\rho^*=\frac{\gamma_1+\gamma_0}{2} \rho$.
Form the same-world equicorrelation blocks and cross block
\begin{equation*}
    R_1=(1-\gamma_1)I_{N}+\gamma_1 J_{N},\qquad
    R_0=(1-\gamma_0)I_{N}+\gamma_0 J_{N},\qquad
    B=(\rho-\rho^*)I_{N}+\rho^* J_{N}.
\end{equation*}
Using Gaussian conditioning, construct $\mathbf Z(0)\mid \mathbf Z(1)=\mathbf z_1 \sim \mathrm N\big(\mu,\Sigma\big),$ where $\mu=B\,R_1^{-1}\mathbf z_1$, $\Sigma=R_0 - B\,R_1^{-1}B.$
Here $I_{N}$ is the identity and $J_{N}$ is the all-ones matrix.

\item Draw $\mathbf z_0\sim \mathrm N(\mu,\Sigma)$, set $u_{0j}=\Phi(z_{0j})$, and compute $D_{j}(0)=\big(F^{(0)}\big)^{-1}(u_{0j})$ for $j=1,\ldots,N$,
where each inverse is obtained by solving $\min_{d}\big(F^{(0)}(d)-u_{0j}\big)^2$ with an optimization method (e.g., BFGS) with tolerance $\varepsilon_{\text{tol}}$.
\end{enumerate}

\item \emph{Mediator $M$.}
      Construct
      $\mathbf C^{(m)}_{j}(a,d)=\left(1,a,N,\mathbf X_j,d,\bar d_{-j}\right)$,
      where $\bar d_{-j}$ is the leave-one-out cluster mean, and draw
      $M_{j}(a,d)\sim
        \mathrm N\left(\mathbf C^{(m)}_{j}(a,d)\boldsymbol\beta^{(m)}_{l},
                        (\sigma^{(m)}_{l})^2\right)$.

\item \emph{Outcome $Y$.}
\label{step:last}
      Construct
      $\mathbf C^{(y)}_{j}(a,m)=
        (1,a,N,\mathbf X_j,d,m,\bar m_{-j})$.

       \noindent 
      Then compute $\E\left[ Y_{\cdot j} \mid \mathbf{M}=\mathbf{m}, \mathbf{D}=\mathbf{d}, A=a, \mathbf{C},N \right]$
      in Threorem \ref{thm:identification}.

\item \emph{Aggregate.}
\label{step:agg}
      Cluster averages of
      $
        \left\{Y_{j}(1,1,1),
               Y_{j}(1,1,0),
               Y_{j}(1,0,0),
               Y_{j}(0,0,0)\right\}
      $
      yield Monte-Carlo draws of
      $\text{TE},\text{NIE},\text{NDE},\text{SME},\text{IME}$,
      which are stored as posterior samples.
\end{enumerate}

\section{Baseline simulation details}
\label{sec:simulation_details}

This section provides additional details about the data-generating process for our simulation study, which involves hierarchical data with clusters and individuals, covariates, treatments, mediators, and outcomes. 
We consider a total of $K=40$ clusters (or groups), indexed by $i=1,2,\dots,I$. For each cluster $i$, the cluster-level covariate $V_i \sim \mathrm{N}\left( \dfrac{3N_i}{50}, 1 \right)$, and the cluster-level treatment $A_i \sim \text{Bernoulli}(0.5).$

\subsection{Mediators}
\label{sec:DGP_mediator}

We consider post-treatment confounders and mediators, $D$ and $M$, for each individual. We consider a scenario where $D$ and $M$ are correlated within the same units, and the same type of mediators are correlated between units within the same cluster as well.

For each individual $(i,j)$, we calculate the mediator mean parameters based on cluster-level and individual-level variables:
\begin{align*}
    \theta_{D_{ij}}(A_i) &= 1.5 \left\{ -2 + 2A_i + \left(0.5 + 0.5A_i\right) \dfrac{N_i}{50} + 0.5X_{1,ij} - 0.5X_{2,ij} + X_{2,ij} + 0.5V_i \right\}, \\
    \theta_{M_{ij}}(A_i) &= -\theta_{D_{ij}}(A_i).
\end{align*}
Additionally, we consider the following correlation structure between mediators for units $j,k$.
\begin{align*}
    \begin{pmatrix}
      D_{ij}(1) \\
      D_{ij}(0) \\
      M_{ij}(1) \\
      M_{ij}(0) \\
      D_{ik}(1) \\
      D_{ik}(0) \\
      M_{ik}(1) \\
      M_{ik}(0) \\
    \end{pmatrix}
    \sim \mathrm{MVN}\left( 
    \begin{pmatrix}
      \theta_{D_{ij}}(1) \\
      \theta_{D_{ij}}(0) \\
      \theta_{M_{ij}}(1) \\
      \theta_{M_{ij}}(0) \\
      \theta_{D_{ik}}(1) \\
      \theta_{D_{ik}}(0) \\
      \theta_{M_{ik}}(1) \\
      \theta_{M_{ik}}(0) \\
    \end{pmatrix},  
    \sigma^2 \begin{pmatrix}
      R & S  \\
      S & R
    \end{pmatrix} \right),
\end{align*}
where the correlation matrices are defined as
\[
    R = \begin{pmatrix}
      1 & \alpha_1 & \alpha_0 & \alpha_2 \\
      \alpha_1 & 1 & \alpha_2 & \alpha_0 \\
      \alpha_0 & \alpha_2 & 1 & \alpha_1 \\
      \alpha_2 & \alpha_0 & \alpha_1 & 1
    \end{pmatrix}, 
    \quad
    S = \begin{pmatrix}
      \rho_0 & 0 & \rho_1 & 0 \\
      0 & \rho_0 & 0 & \rho_1 \\
      \rho_1 & 0 & \rho_0 & 0 \\
      0 & \rho_1 & 0 & \rho_0 \\
    \end{pmatrix}.
\]
We let $\sigma^2=1.0$, $\alpha_0=\alpha_2=0.05$, $\alpha_1=0.03$,  $\rho_0=0.03$ and $\rho_1=0.0$. 

\subsection{Outcome variable}
\label{sec:GDP_outcome}
In Scenarios S1-S6, we consider a mixture model with nonlinear components, described below.
For each individual $(i,j)$, the outcome $Y_{ij}$ is generated based on a function of treatments, mediators, covariates, and random effects.
We consider mixture models of nonlinear data-generating processes with non-gaussian errors for potential outcomes. 
We first compute the following location parameters with nonlinear and interaction terms.
\begin{align*}
    \theta_{1_{ij}} &= 1.0 + A_i + \left(0.5 + 0.5A_i\right) \dfrac{N_i}{50} + 0.5\overline{D_i} - 0.5\overline{M_i} + D_{ij} - M_{ij} \\
    &+ 0.3X_{1,ij}A_i - 0.3X_{2,ij}A_i + 0.1X_{1,ij}^2 + 0.1X_{2,ij}^2 + 0.1X_{1,ij}X_{2,ij} + 0.5X_{3,ij} + 0.5V_i, \\
    \theta_{2_{ij}} &= -1.0 - A_i - \left(0.5 + 0.5A_i\right) \dfrac{N_i}{50} - 0.5\overline{D_i} + 0.5\overline{M_i} - D_{ij} + M_{ij} \\
    &- 0.3X_{1,ij}A_i + 0.3X_{2,ij}A_i - 0.1X_{1,ij}^2 - 0.1X_{2,ij}^2 - 0.1X_{1,ij}X_{2,ij} + 0.5X_{3,ij} + 0.5V_i, \\
    \theta_{3,ij} &= 1.0 + A_i + \left(0.3 + 0.3A_i\right) \dfrac{N_i}{50} + 0.3\overline{D_i} - 0.3\overline{M_i} + D_{ij} - M_{ij} \\
    &+ 0.1X_{1,ij}A_i - 0.1X_{2,ij}A_i + 0.1X_{1,ij}^2 + 0.1X_{2,ij}^2 + 0.1X_{1,ij}X_{2,ij} + 0.3X_{3,ij} + 0.3V_i, \\
    \theta_{4_{ij}} &= -1.0 - A_i - \left(0.3 + 0.3A_i\right) \dfrac{N_i}{50} - 0.3\overline{D_i} + 0.3\overline{M_i} - D_{ij} + M_{ij} \\
    &- 0.1X_{1,ij}A_i + 0.1X_{2,ij}A_i - 0.1X_{1,ij}^2 - 0.1X_{2,ij}^2 - 0.1X_{1,ij}X_{2,ij} + 0.3X_{3,ij} + 0.3V_i, \\
    \theta_{5,ij} &= -0.5 \theta_{1_{ij}}, 
    \theta_{6,ij} = -1.0\theta_{2_{ij}}, 
    \theta_{7,ij} = -1.5 \theta_{3_{ij}}, 
    \theta_{8,ij} = -2.0 \theta_{4_{ij}},
\end{align*}
where $\overline{D_i}$ and $\overline{M_i}$ are the cluster-level means of the mediators: $\overline{D_i} = \dfrac{1}{N_i} \sum_{j=1}^{N_i} D_{ij}, \quad \overline{M_i} = \dfrac{1}{N_i} \sum_{j=1}^{N_i} M_{ij}.$
The outcome $Y_{ij}$ is then generated from a mixture distribution according to a  group latent variable $G_i\sim \text{Categorical}(0.2, 0.3, 0.5),$
    \begin{itemize}
        \item If $G_i=1$, $Y_{ij} \sim 0.5\mathrm{t}( \theta_{1,ij}, 1.5 ) + 0.5  \mathrm{t}( \theta_{2,ij}, 1.5 ).$
        \item If $G_i=2$, $Y_{ij} \sim 0.5\mathrm{t}( \theta_{3,ij}, 1.5 ) + 0.25  \mathrm{t}( \theta_{4,ij}, 1.5 ) + 0.25  \mathrm{t}( \theta_{5,ij}, 1.5 ).$
        \item If $G_i=3$, $Y_{ij} \sim 0.5\mathrm{t}( \theta_{6,ij}, 1.5 ) + 0.25  \mathrm{t}( \theta_{7,ij}, 1.5 ) + 0.25  \mathrm{t}( \theta_{8,ij}, 1.5 ),$
    \end{itemize}
    where $X \sim \mathrm{t}( \theta, \nu )$ represents a t-distributed random variable with location parameter $\theta$ and $\nu$ degrees of freedom.

In Scenario S7, we consider a simple scenario where the parametric model is correctly specified. Specifically, $Y_{ij}$ is generated from:
\begin{align*}
    Y_{ij} \sim \mathrm{N}(\theta_{ij}, 1.0),
\end{align*}
where $\theta_{ij}=1.0 + A_i  + 0.5\overline{D_i} - 0.5\overline{M_i}  + 0.5 D_{ij} - 0.5 M_{ij} + 0.3X_{1,ij} - 0.3X_{2,ij} + 0.3X_{3,ij} + 0.3 N_i$. The following parametric model is fitted: $\E[Y_{ij} \mid A_i, \mathbf{D}_i, \mathbf{M}_i, \mathbf{C}_i, N_i] = \beta_0 + \beta_1A_i  + \beta_2\overline{D_i} + \beta_3\overline{M_i}  + \beta_4 D_{ij} + \beta_5 M_{ij} + \beta_6 X_{1,ij} + \beta_7 X_{2,ij} + \beta_8 X_{3,ij} + \beta_9 N_i.$

\end{document}